%% file: main.tex
\documentclass{article}

\usepackage{arxiv}

\usepackage[utf8]{inputenc} 
\usepackage[T1]{fontenc}    
\usepackage{hyperref}       
\usepackage{url}            
\usepackage{booktabs}       
\usepackage{amsfonts}       
\usepackage{nicefrac}       
\usepackage{microtype}      
\usepackage{lipsum}		
\usepackage{graphicx}
\usepackage{natbib}
\usepackage{doi}
\usepackage{amsthm}
\newtheorem{definition}{Definition}
\newtheorem{theorem}{Theorem}

\newtheorem{proposition}{Proposition}

\usepackage[linesnumbered,ruled]{algorithm2e}
\usepackage{diagbox}
\usepackage{multirow}
\usepackage{subfigure}
\usepackage{bbm}
\usepackage{amsmath}
\usepackage{makecell}
\usepackage{enumitem}

\usepackage{amsmath}
\DeclareMathOperator*{\argmax}{arg\,max}

\newcommand{\svtfull}{sparse vector technique\xspace}
\newcommand{\dpfull}{differential privacy\xspace}
\newcommand{\etal}{\textit{et al.}\xspace}
\newcommand{\eg}{\textit{e.g.}\xspace}
\newcommand{\ie}{\textit{i.e.}\xspace}

\title{Unleash the Power of Ellipsis: Accuracy-enhanced Sparse Vector Technique with Exponential Noise}

\author{ 
{Yuhan Liu}\\
	Renmin University of China\\
	\texttt{liuyh2019@ruc.edu.cn} \\
	\And
    {Sheng Wang} \\
	Alibaba Group\\
	\texttt{sh.wang@alibaba-inc.com} \\
 	\And
    {Yixuan Liu} \\
	Remin University of China\\
	\texttt{liuyixuan@ruc.edu.com} \\
 	\And
 {Feifei Li} \\
	Alibaba Group\\
	\texttt{lifeifei@alibaba-inc.com} \\
 	\And
 {Hong Chen} \\
	Renmin university of China\\
	\texttt{chong@ruc.edu.cn} \\
}

\renewcommand{\shorttitle}{}

\hypersetup{
pdftitle={A template for the arxiv style},
pdfsubject={q-bio.NC, q-bio.QM},
pdfauthor={David S.~Hippocampus, Elias D.~Striatum},
pdfkeywords={First keyword, Second keyword, More},
}

\begin{document}
\maketitle

\begin{abstract}
	\input{abstract}
\end{abstract}


\input{introduction}
\input{preliminaries}
\input{privacy-analysis}
\input{construction}
\input{evaluation}
\input{related_work}
\input{conclusion}
\input{appendix}

\input{main.bbl}
\bibliographystyle{unsrtnat}






\end{document}

%% file: abstract.tex
The Sparse Vector Technique~(SVT) is one of the most fundamental tools in differential privacy~(DP).
It works as a backbone for adaptive data analysis by answering a sequence of queries on a given dataset, and gleaning useful information in a privacy-preserving manner.
Unlike the typical private query releases that directly publicize the noisy query results, SVT is less informative---it keeps the noisy query results to itself and only reveals a binary bit for each query, indicating whether the query result surpasses a predefined threshold.
To provide a rigorous DP guarantee for SVT, prior works in the literature adopt a \textit{conservative} privacy analysis by assuming the direct disclosure of noisy query results as in typical private query releases.
This approach, however, hinders SVT from achieving higher query accuracy due to an overestimation of the privacy risks, which further leads to an excessive noise injection using the Laplacian or Gaussian noise for perturbation.
Motivated by this, we provide a new privacy analysis for SVT by considering its less informative nature.
Our analysis results not only broaden the range of applicable noise types for perturbation in SVT, but also identify the exponential noise as optimal among all evaluated noises (which, however, is usually deemed non-applicable in prior works).
The main challenge in applying exponential noise to SVT is mitigating the sub-optimal performance due to the bias introduced by noise distributions.
To address this, we develop a utility-oriented optimal threshold correction method and an appending strategy, which enhances the performance of SVT by increasing the precision and recall, respectively.
The effectiveness of our proposed methods is substantiated both theoretically and empirically, demonstrating significant improvements up to $50\%$ across evaluated metrics.

%% file: introduction.tex
\section{Introduction}
\begin{figure}[tbp!]
    \centering
    \includegraphics[width=0.7\linewidth]{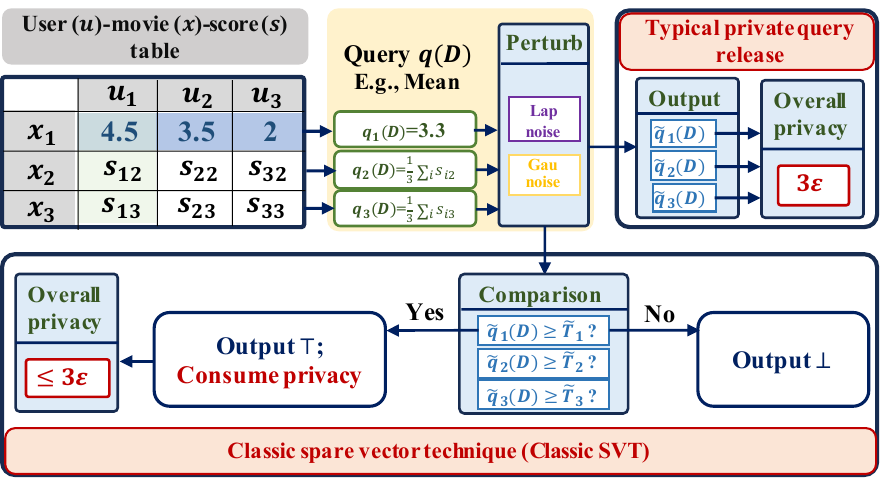}
    \caption{Classic sparse vector technique~(SVT) and typical private query release on publicizing movies with the top-c highest scores. Classic SVT: after perturbing the movie score $q_i(D)$ and the predefined threshold $T_i$ with either Laplacian or Gaussian noise, SVT compares the noisy score $\tilde{q}_i(D)$ and the noisy threshold $\tilde{T}_i$. If $\tilde{q}_i(D)\geq \tilde{T}_i$, $\top$ is output. Otherwise, $\bot$ is output. Typical private query releases: all $\tilde{q}_{i}$ are released and sorted to obtain the noisy top-c movies.}
    \label{fig: comparison_svt_typical}
\end{figure}
The \svtfull~(SVT)~\cite{dwork2009complexity, dwork2014algorithmic} is one of the most fundamental algorithmic tools in DP, serving as a backbone for adaptive data analysis in many applications, such as feature selection~\cite{bassily2018model}, stream data analysis~\cite{hasidim2020adversarially}, and top-$c$ selection~\cite{lyu2017understanding}.
At a high level, SVT answers a sequence of queries on a given dataset and extracts useful information in a privacy-preserving manner.
Unlike typical private query releases that directly reveal the noisy query results, \textit{SVT discloses less information}---it outputs only a binary bit for each query, indicating whether the query result surpasses a predefined threshold~(Cf. Figure~\ref{fig: comparison_svt_typical}).
The major advantage of SVT over the typical private query releases is that only the positive queries~(those whose results are above the predefined threshold) consume privacy, thus potentially allowing an infinite number of queries on a single dataset.

Despite the widespread usage of SVT, it still suffers from a low query accuracy~\cite{zhu2020improving}
due to a conservative privacy analysis, resulting in a sub-optimal noise selection.
Specifically, prior works in the literature na\"ively approximate the privacy budget consumption of positive queries in SVT with that of differentially private query result releases~(Cf. Section~\ref{subsec: noise-distribution}).
However, since SVT discloses less information in each query~(\ie, only a binary bit instead of a noisy real-valued query result), this conservative privacy analysis approach overestimates its privacy risk.
Consequently, SVT is constrained to use noise with large variance as in typical query releases for threshold and query perturbation, such as the Laplacian~\cite{lyu2017understanding} or Gaussian~\cite{zhu2020improving} noise. Such noise types then lead to an excessive noise injection and a sub-optimal query accuracy.

Motivated by this, our first contribution is providing a new privacy analysis (Cf. Theorem~\ref{theo: privacy-svt-lips}), which captures the \textit{less informative nature of SVT} by precisely computing the privacy loss of the private comparison between thresholds and query results. Informally speaking, our analysis demonstrates that any noise whose cumulative distribution function satisfies the Lipschitz condition can be applied to SVT, allowing the use of noise types with smaller variance under the same privacy level.
Based on our analysis, we identify the exponential noise, previously deemed non-applicable to SVT, as the optimal choice among all considered noise types. Exponential noise satisfies Theorem~\ref{theo: privacy-svt-lips} tightly and benefits from a smaller variance compared to the others~(Cf. Figure~\ref{fig: variance}).

However, applying the exponential noise to SVT is challenging due to the bias introduced to the query results.
The na\"ive bias correction method, which subtracts the expected value of random noise from the noisy query results, yields a sub-optimal performance.
This is primarily because SVT requires correcting each noisy query result individually, whereas the na\"ive correction method is designed for correcting the aggregation of noisy query results~\cite{erlingsson2014rappor}.
When random noises are aggregated~(as in the aggregation of noisy query results), their summed value approximates the expected value. In contrast, individual noisy query results are less concentrated around their expectations, resulting in less accurate corrections (Cf. Section~\ref{subsec: correction_motivation}).

To address this challenge, our second contribution involves developing an optimal threshold correction method\footnote{As we further discussed in Section~\ref{subsec: overview}, correcting the threshold is equivalent to correcting the noisy query results in SVT.} and an appending strategy, which effectively mitigate introduced bias and significantly enhance query accuracy.
Instead of correcting the bias of each query result individually, our method focuses on optimizing the precision of the output binary bit vector.
Specifically, we search for an optimal correction term that maximizes the probability of SVT distinguishing every true positive query from the true negative ones~(Cf. Equation~\ref{eq: success-rate}).
Additionally, we introduce an appending strategy to complement the threshold correction method and boost the recall of SVT. Concretely, each query with a noisy negative output is appended to the end of the query queue for another round of querying.
The intuition behind this is that, true positive queries are more likely to be identified as positive queries by SVT compared to the true negative ones after multiple rounds of querying, thus increasing the recall.

In conclusion, our main contributions are summarized as follows:
\begin{enumerate}[leftmargin=2em]
    \item We provide a new privacy analysis that precisely captures the privacy loss of SVT. This analysis not only broadens the noise choices for query perturbation but also identifies exponential noise as the optimal choice among all considered ones.
    \item Building on our privacy analysis, we propose a new SVT variant with exponential noise, where an optimal threshold correction method and an appending strategy are developed to mitigate the introduced bias and boost the query accuracy.
    \item Comprehensive experiments conducted on both synthetic datasets and real-world datasets demonstrate that our proposed methods significantly increase the query accuracy of SVT by $2\%\sim50\%$ across evaluated metrics.
\end{enumerate}

%% file: preliminaries.tex
\section{Preliminaries}\label{sec: preliminary}
In this section, we first recall the definition~(Cf. Definition~\ref{def: DP}), mechanisms~(Cf. Equation~\ref{eq: dp_mechanism}), and properties of the \dpfull~(\ie, post-processing~(Cf. Proposition~\ref{prop: post_processing}) and composition~(Cf. Theorem~\ref{theo: composition})).
Then, we proceed to describe the \svtfull algorithm~(Cf. Algorithm~\ref{alg: basic-svt}) and some of its predominant variants.

\subsection{Differential Privacy}\label{subsec: background_dp}
The \dpfull~(DP) was first introduced by Dwork \etal~\cite{dwork2006differential} and recently became a \textit{de facto} standard for privacy protection. Roughly speaking, DP ensures that the likelihood of any specific output from a random algorithm varies little with sensitive neighboring inputs. The formal definition is as follows:
\begin{definition}[Differential Privacy~\cite{dwork2006differential}]\label{def: DP}
A randomized algorithm $\mathcal{M}:\mathbb{D}\rightarrow\mathbb{R}$ satisfies $(\varepsilon,\delta)$-differential privacy if for any two neighboring datasets $D$ and $D^{\prime}$, and any output $o\subseteq\mathbb{O}$
$$\Pr[\mathcal{M}(D)\in o]\leq e^{\varepsilon}\cdot\Pr[\mathcal{M}(D^{\prime})\in o]+\delta.$$
\end{definition}
When $\delta=0$, we say that $\mathcal{M}$ satisfies pure DP~(denoted by $\varepsilon$-DP). Otherwise, it guarantees an approximate DP.

One of the typical methods of achieving DP is through the additive-noise mechanism, which corrupts and releases query results with additive noise randomly drawn from certain distributions~\cite{geng2019optimal}. That is, given a dataset $D$, to guarantee DP, a randomized algorithm $\mathcal{M}$ releases:
\begin{equation}\label{eq: dp_mechanism}
    \mathcal{M}(D) = q(D) + X,
\end{equation}
wherein $q(D)$ is the result of a query carried on $D$ and $X$ is the additive noise drawn from a probability distribution $\mathcal{N}$. The noise scale is calibrated to the sensitivity of $q(D)$ defined in the following:
\begin{definition}[$l_p$-sensitivity]
    For a real-valued query $q:\mathbb{D}\rightarrow \mathbb{R}$, the $l_p$-sensitivity of $q$ is defined as:
    $$\Delta_p=\max\limits_{D,D^{\prime}\in\mathbb{D}}\Vert q(D)-q(D^{\prime})\Vert_p,$$
    where $\Vert \cdot \Vert_{p}$ denotes the $l_p$ norm and $D,D^{\prime}$ is a pair of neighboring datasets that differ by one element.
\end{definition}

The \textit{Laplace mechanism} is one of the most classic additive-noise mechanisms that achieves $\varepsilon$-DP by letting $X$ follow a Laplace distribution centered at $0$ with a noise scale $b=\frac{\Delta_1}{\varepsilon}$, denoted as $X \sim Lap\left(\frac{\Delta_1}{\varepsilon}\right)$, where $\Delta_1$ is the sensitivity of the query. Meanwhile, the \textit{Gaussian mechanism} mechanism typically guarantees $(\varepsilon,\delta)$-DP by drawing the noise $X$ from a normal distribution $N(0,\sigma^{2})$, where $\sigma$ is proportional to $\frac{\Delta_2}{\varepsilon}$.

DP provides a rigorous privacy guarantee mathematically and is straightforward to achieve. Additionally, there are other two properties that contribute to the wide-ranging applications of DP. Firstly, it is immune to \textit{post-processing}, which is formally described as follows:
\begin{proposition}[Post-processing\cite{dwork2014algorithmic}]\label{prop: post_processing}
Let $\mathcal{M}:\mathbb{D}\rightarrow \mathcal{R}$ be a randomized algorithm that is $(\varepsilon,\delta)$ differentially private. Let $f:\mathbb{R}\rightarrow\mathbb{R}$ be an arbitrary randomized mapping. Then $f\circ\mathcal{M}$ is $(\varepsilon,\delta)$-differentially private.
\end{proposition}

Secondly, different differentially private building blocks can be elegantly combined for designing more sophisticated algorithms, as described by the following theorem:
\begin{theorem}[Sequential Composition\cite{dwork2006differential}]\label{theo: composition}
    Let $\mathcal{M}_i:\mathbb{N}^{\lvert \chi\rvert}\rightarrow \mathcal{R}_{i}$ be an $(\varepsilon_i,\delta_i)$-differential privacy algorithm for $i\in [k]$. If $\mathcal{M}_k:\mathbb{N}^{\lvert \chi\rvert}\rightarrow\prod\limits_{i=1}^{k}\mathcal{R}_i$ is defined to be $\mathcal{M}_{[k]}(x)=(\mathcal{M}_1(x),\ldots,\mathcal{M}_k(x))$, then $\mathcal{M}_{[k]}$ is $(\sum\limits_{i=1}^{k}\varepsilon_i,\sum\limits_{i=1}^{k}\delta_i)$-differentially private.
\end{theorem}

\textbf{Other composition theorems.} 
In addition to sequential composition, more advanced composition theorems with tighter composition results are also utilized in the literature~\cite{dwork2010boosting, kairouz2015composition, meiser2018tight, zhu2020improving, gopi2021numerical}.
Among these, composition under \textit{R\'enyi differential privacy}~(RDP)~\cite{mironov2017renyi} is one of the widely adopted methods in the literature, owing to its concise and tight analysis result. RDP employs R\'enyi divergence as a tool for privacy measurement, and combining the RDP notion with Theorem~\ref{theo: composition} yields a tighter privacy bound.
Zhu~\etal~\cite{zhu2020improving} later demonstrate that SVT can be combined with RDP for improved privacy guarantees.
Although it is straightforward to extend the existing composition theorems applied to SVT to our proposed SVT variants~\cite{dwork2014algorithmic, zhu2020improving}, for ease of analysis, we derive our privacy analysis in this work based on Theorem~\ref{theo: composition}. While we do not dive deep into the theoretical analysis of other composition theorems in this work, an empirical study on the performance of our proposed method under RDP composition is provided in Appendix~\ref{sec: eva_rdp} to offer further insights.

%
\subsection{Sparse Vector Technique}\label{subsec: background_svt}
\begin{algorithm}
    \caption{SVT. Privately indicates if query results are above thresholds.}\label{alg: basic-svt}
    \LinesNumbered
    \KwIn{$Q=\{q_1,q_2,\ldots\}$,$\Delta$,$\varepsilon_1$,$\varepsilon_2$,$c$,$k_{max}$,$T=\{T_1,T_2,\ldots\}$, option \texttt{RESAMPLE}}
    $\rho\sim\mathcal{N}_1(\varepsilon_1,\Delta)$;\ $n_c=0$;$n_a$=0\;\label{line: basic-svt-threshold-noise}
    \label{line:basic-svt-threshold-perturbation}
    \For{$i=1,2,3,\ldots,k_{max}$}{
        $n_a=n_a+1$\;
        $v_i\sim \mathcal{N}_2(\varepsilon_2,\Delta)$\;\label{line: basic-svt-query-noise}
        $\tilde{q}_i(D)=q_i(D)+v_i$\label{line: basic-svt-query-perturb}\tcp*{Query perturbation}
        $\tilde{T}_i= T_i+\rho$\label{line: basic-svt-threshold-perturb}\tcp*{Threshold perturbation}
    \eIf{$\tilde{q}_i(D)\geq\tilde{T}_i$\hfill\tcp{Private comparison}\label{line: basic-svt-compare}}
        {Output $a_i=\top$\label{line: basic-svt-positive-outcome}\;$n_c=n_c+1$\;
        if \texttt{RESAMPLE}, $\rho\sim \mathcal{N}_1(\varepsilon_1,\Delta)$\label{line: basic-svt-resample}\;
        \textbf{Abort} \text{ if $n_c\geq c$ or $n_a\geq k_{max}$\label{line: basic-svt-halt}}\label{line:Abort}\;}{Output $a_i=\bot$\label{line: basic-svt-negative-outcome}}}
\end{algorithm}
In principle, SVT works as follows. 
Given a sequence of queries $q_i(D)$ on a dataset $D$ and their corresponding predefined thresholds $T_i$, SVT first perturbs them with random noise drawn from noise distribution $\mathcal{N}_2$ and $\mathcal{N}_1$~(Line~\ref{line: basic-svt-query-noise}, Line~\ref{line: basic-svt-query-perturb}, Line~\ref{line: basic-svt-threshold-noise}, and Line~\ref{line: basic-svt-threshold-perturb}), respectively.
Then, it compares each noisy query result $\tilde{q}_i(D)$ with the noisy threshold $\tilde{T}_{i}$~(Line~\ref{line: basic-svt-compare}). If $\tilde{q}_i(D)$ is no smaller than $\tilde{q}_i(D)$, $\top$ is output as a positive indicator~(Line~\ref{line: basic-svt-positive-outcome}). Otherwise, $\bot$ is output as a negative indicator~(Line~\ref{line: basic-svt-negative-outcome}).
The algorithm halts either when the number of positive outcomes reaches its maximum value $c$, or when the total number of queries exceeds its maximum value $k_{max}$~(Line~\ref{line: basic-svt-halt}).

Different variants in the literature set this indicator to different values. For instance, Dwork et al. \cite{dwork2009complexity} set \texttt{RESAMPLE} to \texttt{True}, whereas Lyu et al. \cite{lyu2017understanding} argue that the overall privacy cost is significantly reduced by setting it to \texttt{False}, therefore significantly boosting the performance of SVT. Furthermore, Zhu et al. \cite{zhu2020improving} alternate \texttt{True} and \texttt{False} periodically for certain applications.

In this work, we provide a privacy analysis of our proposed method under both \texttt{True} and \texttt{False} settings~(Cf. Theorem~\ref{theo: privacy-guarantee}). As the empirical evaluation results for both settings demonstrate similar trends and \texttt{RESAMPLE=True} yields better performance in general~\cite{lyu2017understanding}, we only present the results of \texttt{RESAMPLE} set to \texttt{True} in Section~\ref{sec: evaluation}.

\subsubsection{Private Comparison}\label{subsec: noise-distribution}
As mentioned above, to provide a rigorous DP guarantee, certain types of noise are required to ensure that for any noisy threshold $\tilde{T}_i$,
\begin{equation}\label{eq: svt-dp}
    \frac{\Pr[\tilde{q}_i(D)\geq \tilde{T}_i]}{\Pr[\tilde{q}_i(D^{\prime})\geq\tilde{T}_i]}\leq e^{\varepsilon}
\end{equation}
holds with high probability.

Previous works~\cite{dwork2009complexity, lyu2017understanding, zhu2020improving} bound the probability ratio on the left of Inequality~\ref{eq: svt-dp} by na\"ively assuming that Inequality~\ref{eq: svt-dp} is equivalent to the following:
\begin{equation}\label{eq: classic-query-dp}
   \frac{\Pr[\tilde{q}_i(D)\in o]}{\Pr[\tilde{q}_i(D^{\prime})\in o]}\leq e^{\varepsilon},
\end{equation}
where $o$ is any possible output in the output space.

To ensure Inequality~\ref{eq: classic-query-dp} hold, \textit{Laplacian}~($Lap\left(\frac{\Delta f}{\varepsilon_{\varepsilon_1/\varepsilon_2}}\right)$)~\cite{dwork2009complexity, lyu2017understanding} and \textit{Gaussian} noise~($N\left(0,\sigma_{1/2}\right)$, where $\sigma_b$ for $b\in\{1,2\}$ is a function of $\varepsilon_b$)~\cite{zhu2020improving} are then adopted for both the query result and predefined threshold perturbation.

\textbf{However, this privacy analysis overestimates the privacy risk in SVT, leading to an overly conservative choice of noise.}
Note that Inequality~\ref{eq: svt-dp} is easier to achieve compared to Inequality~\ref{eq: classic-query-dp}:
while Inequality~\ref{eq: classic-query-dp} requires a rigorous bound on the probability ratio over every possible subset~($o$) in the output space, Inequality~\ref{eq: classic-query-dp} splits the output space into two~(\ie, $\geq \tilde{T}_i$ and $<\tilde{T}_i$) and only requires a bound on the \textit{accumulative} probability over the two parts.
Thus, the high probability ratios~(\ie, the left side term in Inequality~\ref{eq: classic-query-dp}) incurred by some extreme subsets $o$ in Inequality~\ref{eq: classic-query-dp} are averaged down in Inequality~\ref{eq: svt-dp} by those with relatively low probability ratios.
In other words, let $\mathcal{O}_i$ be the output space of $\tilde{q}_i(D)$, which is split into two parts: $\mathcal{R}:=\{z\in\mathcal{O}_i|z\geq T_i\}$ and $\mathcal{L}:=\mathcal{O}_i\backslash\mathcal{R}$.
Then, by further splitting $\mathcal{R}$ into $n$ subsets that $\mathcal{R}=o_1\cup,\ldots,\cup o_n$, we have the following inequality holds:
\begin{equation*}
 \frac{\Pr[\tilde{q}_i(D)\geq \tilde{T}_i]}{\Pr[\tilde{q}_i(D^{\prime})\geq \tilde{T}_i} = 
    \frac{\sum_{j=1}^{n}\Pr\left[\tilde{q}_i(D)\in o_j\right]}{\sum_{j=1}^n \Pr\left[\tilde{q}_i(D^{\prime})\in o_j\right]}\leq
    \max_{j\in[n]}\frac{\Pr\left[\tilde{q}_i(D)\in o_j\right]}{\Pr\left[\tilde{q}_i(D^{\prime})\in o_j\right]}
\end{equation*}

Intuitively, there may be some noise distributions that, although they do not satisfy Inequality~\ref{eq: classic-query-dp}, can still ensure that Inequality~\ref{eq: svt-dp} holds with high probability while introducing less distortion.

\subsubsection{Advantage of SVT Over Typical Private Releases}
One of the most unique properties of SVT, as well as its variants, is that only the positive outcomes incur privacy costs. Specifically, the overall privacy budget $\varepsilon$ is independent of the value of $k_{max}$ depends instead on the noise scale (\ie, the variance of the noise) of each noisy query and the number of positive outcomes $n_c$.

\textbf{Example.}
Consider a scenario where a data analyst wants to identify the top-$c$ most popular movies~(Cf. Figure~\ref{fig: comparison_svt_typical}). Each movie is rated by a group of individuals who have watched it. Directly releasing the score of each movie~(\ie, $q_i(D)$ in Figure~\ref{fig: comparison_svt_typical}) may reveal sensitive information about whether an individual has watched a particular movie, thus necessitating the use of DP to protect individual privacy. 
\textit{Typical private query releases} perturb and release the score of all candidate movies, which incurs a privacy budget proportional to the number of all candidates $n$. In contrast, using \textit{SVT} to approximate the top-$c$ movies by only outputting the indices of first $c$ films whose scores exceed a predefined threshold $T$ results in an overall privacy cost proportional to $c$ rather than $n$. When $c\ll n$, this approach significantly reduces the privacy cost.

Note that Figure~\ref{fig: comparison_svt_typical} uses $\texttt{MEAN}(D)$ as an example of query $q_i(D)$. In practice, $q_i(D)$ can represent any queries with real-valued answers, such as $\texttt{SUM}(D)$, $\texttt{COUNT}(D)$, $\texttt{MAX}(D)$. Furthermore, while we use top-$c$ selection to demonstrate the effectiveness of our proposed method, SVT can be applied to many other scenarios where queries need to be evaluated against specific criteria~(predefined threshold) in a privacy-preserving manner, such as feature selection~\cite{bassily2018model}, steaming data analysis~\cite{hasidim2020adversarially}.

\subsubsection{Utility Metric}\label{subsec: utility-metric}
The utility of SVT is measured by a specialized utility metric, namely $(\alpha,\beta)$-accuracy~\cite{dwork2014algorithmic}.
\begin{definition}[$(\alpha,\beta)$-accuracy]\label{def: svt-acc}
    An algorithm which outputs a stream of answers $a_1,\ldots,\in\{\top,\bot\}^{*}$ in response to a stream of $k$ queries $q_1,\ldots,q_k$ is $(\alpha,\beta)$-accurate with respect to a threshold $T$ if except with probability at most $\beta$, the algorithm does not halt before $q_k$, and for all $a_i=\top$:$$q_i(D)\geq T-\alpha$$
    and for all $a_i=\bot$:$$q_i(D)\leq T+\alpha.$$
\end{definition}
Definition~\ref{def: svt-acc} specifies that given an error tolerance parameter $\alpha$, SVT achieves a success probability of at least $1-\beta$. Particularly, queries with results falling within the interval $[T-\alpha,T+\alpha]$ are allowed to be misclassified. For instance, if a query $q(D)$ falls within $(T,T+\alpha]$, the SVT algorithm is still considered successful even if it incorrectly classifies $q(D)$ as negative by outputting $\bot$. Specifically, it is demonstrated by prior works that the SVT algorithm, when using Laplacian noise~(\ie, $\mathcal{N}_1$ and $\mathcal{N}_2$ are Laplace distributions), is $(\frac{8(\ln k+\ln \frac{2}{\beta})}{\varepsilon},\beta)$-accurate~\cite{dwork2014algorithmic}.

%% file: privacy-analysis.tex
\section{Privacy Analysis Revisit}\label{sec: privacy-revisit}
As discussed in Section~\ref{subsec: noise-distribution}, it has been observed that previous privacy proofs tend to overestimate the privacy risk in SVT by bounding the privacy over all possible subsets in the output space of query results. The underlying assumption is that disclosing a binary bit is as risky as directly revealing the query result itself. However, this may not hold true: intuitively, a binary bit leaks less information than the complete query result and, therefore, should pose a lower risk.

Motivated by this, we revisit the privacy analysis of SVT, taking into account its less informative nature. We provide a new analysis result in Theorem~\ref{theo: privacy-svt-lips}. Informally speaking, our results indicate that for query perturbation, SVT poses a less stringent constraint on eligible noise distributions compared to typical private query releases, thereby accommodating a broader range of noise distributions.
\begin{theorem}[Privacy of SVT]\label{theo: privacy-svt-lips}
    Algorithm~\ref{alg: basic-svt} satisfies differential privacy if for any real numbers $b_1$, $b_2$, there are two positive real numbers $k_1$ and $k_2$ such that the following inequalities hold:
\begin{equation}\label{eq: pure_dp}
    \vert \ln(f_1(x))-\ln(f_1(x+b_1))\vert \leq k_1\vert b_1\vert,
\end{equation}
\begin{equation}\label{eq: lispchitz}
    \vert \ln(1-F_2(x))-\ln(1-F_2(x+b_2))\vert \leq k_2\vert b_2\vert,
\end{equation}
where $f_1(\cdot)$ and $F_2(\cdot)$ are the probability density function of $\mathcal{N}_1$ and the cumulative function of the $\mathcal{N}_2$, respectively.
\end{theorem}
We defer the proof of Theorem~\ref{theo: privacy-svt-lips} to Appendix~\ref{sec: proof-of-lips-privacy} and provide only the key takeaways here.

\textbf{Takeaways from Theorem~\ref{theo: privacy-svt-lips}.}
\underline{First}, Theorem~\ref{theo: privacy-svt-lips} \textit{broadens} the range of noise distributions for query perturbation. Specifically, as indicated by Equation~\ref{eq: lispchitz}, distributions such as the exponential and Gumbel distribution are now eligible for query perturbation~(\ie, $\mathcal{N}_2$ in Algorithm~\ref{alg: basic-svt}), whereas these types of noise were previously considered unsuitable for SVT.
\underline{Second}, although Theorem~\ref{theo: privacy-svt-lips} relaxes the constrains on noise distributions for \textit{query result perturbation}, the eligible noise distributions for \textit{threshold perturbation}~(\ie, $\mathcal{N}_1$ in Algorithm~\ref{alg: basic-svt}) remain \textit{unchanged} compared to the previous work. This is because, despite the query results not being directly released, the predefined threshold is public information. To prevent negative queries from consuming privacy, the true threshold value, which is compared with the noisy query result, must also be obscured. Hence, the same noise distributions, such as Laplacian or Gaussian, are used in traditional private query releases.

To clarify, we summarize some of the most commonly used noise distributions for both threshold perturbation and query result perturbation in Table~\ref{tbl: noise-combination}.
\begin{table}[htbp!]
  \centering
  \caption{Feasible noise distributions for SVT}
    \begin{tabular}{ccccc}
    \toprule
      & Laplace & Gaussian & Exponential & Gumbel\\
    \midrule
    Query Result & \checkmark &\checkmark &\checkmark &\checkmark\\
    \midrule
    Threshold &\checkmark& \checkmark & $\times$ & $\times$\\
    \bottomrule
    \end{tabular}%
  \label{tbl: noise-combination}%
\end{table}%

%% file: construction.tex
\begin{figure*}[t!]
    \centering
    \includegraphics[width=\linewidth]{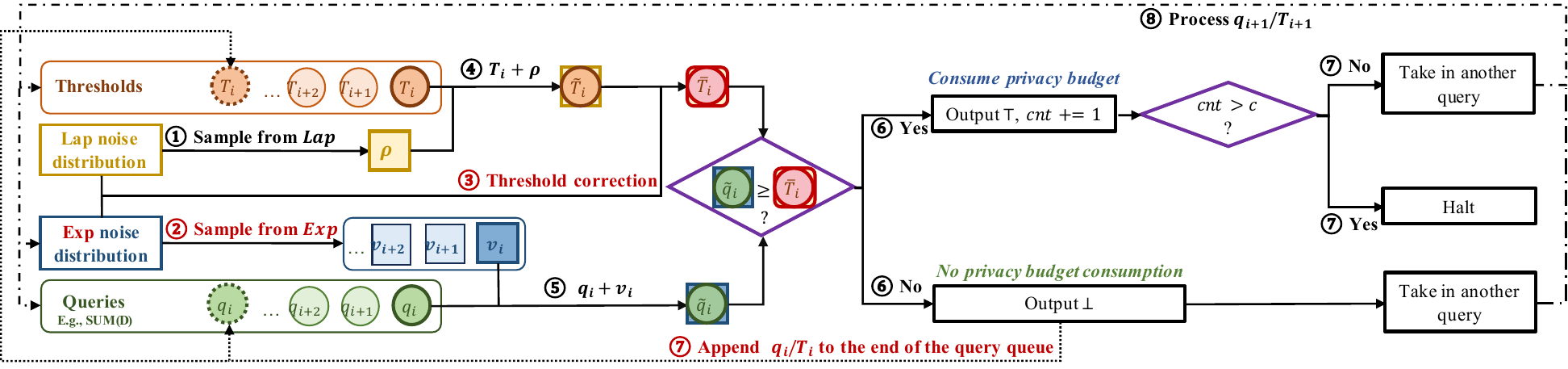}		
    \caption{The overall design of Algorithm \ref{alg: exp-svt}, with our newly proposed methods in this work emphasized with red color. The \texttt{RESAMPLE} is set to \texttt{False}. In step \textcircled{3}, the optimal threshold correction term $r^{op}$ is computed based on the Laplace and exponential distribution and added to the current threshold. The threshold~(step \textcircled{4}) and the query~(step \textcircled{5}) are then perturbed with random noise sampled in step \textcircled{1} and \textcircled{2}, respectively. After comparing the noisy threshold with the noisy query in step \textcircled{6}, the algorithm outputs $\top$ and compares the number of accumulated positive outcomes~($cnt$) with a positive threshold $c$ if $\tilde{q}_i\geq \bar{T}_i$. Otherwise, $\bot$ is output. If \texttt{APPEND} is set to \texttt{True}, the negative query is appended to the end of the query queue for an additional round of querying. If the algorithm does not halt in step \textcircled{7}, it takes in another threshold-query pair and repeats the procedure~(step \textcircled{8}).}
    \label{fig: framework}
\end{figure*}
\section{Construction}
Based on the two crucial takeaways discussed in Section~\ref{sec: privacy-revisit}, we propose an enhanced SVT algorithm with improved query accuracy, referred to as SVT-Exp. This algorithm is summarized in Algorithm~\ref{alg: exp-svt}, with key changes in the algorithm design highlighted by \underline{underlines}. Additionally, the main steps of Algorithm~\ref{alg: exp-svt} are illustrated in Figure~\ref{fig: framework} for better understanding.
\begin{algorithm}
    \caption{SVT with exponential noise, optimal threshold correction, and appending strategy~(SVT-Exp).}\label{alg: exp-svt}
    \LinesNumbered
    \KwIn{$Q=\{q_1,q_2,\ldots\}$,$\Delta$,$\varepsilon_1$,$\varepsilon_2$,$\lambda$,$c$,$k_{max}$,$T=\{T_1,T_2,\ldots\}$,$\alpha$, $\varepsilon$, $m$,$e$,$k$ option \texttt{RESAMPLE}, option \texttt{APPEND}.}
    $\varepsilon=\varepsilon_1+\varepsilon_2$; $n_c=0$; $n_a=0$\;
    $b=\frac{\Delta}{\varepsilon_1}$;$\rho\sim Lap(b)$\;\label{line: exp-svt-threshold-noise}
    \label{line:exp-svt-threshold-perturbation}
        \underline{$r^{op}=\texttt{CorrectionTerm}(b,\lambda,\alpha,m,e,k)$}\tcp*{Threshold correction term computation}\label{line: exp-svt-threshold-correction}
    \For{$i=1,2,3,\ldots$}{
        $n_a=n_a+1$\;
        \underline{$\lambda=\frac{\varepsilon_2}{2c\Delta}$;\ 
        $v_i\sim Exp(\frac{1}{\lambda})$}\tcp*{Query perturbation}\label{line: exp-svt-query-noise}
        $\tilde{q}_i(D)=q_i(D)+v_i$\;
        $\tilde{T}_i=T_i+\rho$\tcp*{Threshold perturbation}
    \eIf{\underline{$\tilde{q}_i(D)\geq \tilde{T}_i+r^{op}$}\hfill\tcp{Private comparison}\label{line: exp-svt-query-perturb}}
        {Output $a_i=\top$\label{line: exp-svt-positive-outcome}\;$n_c=n_c+1$\;
        $S_{\top}=S_{\top}\cup i$\;
        if \texttt{RESAMPLE}, $\rho\sim Lap(b)$\label{line: exp-svt-resample}\;
        \textbf{Abort} \text{ if $n_c\geq c$ \text{or} $n_a\geq k_{max}$}\label{line:Abort}\;}{Output $a_i=\bot$\label{line: exp-svt-negative-outcome}\;
        \If{\underline{\texttt{APPEND}}\label{line: start_appending}}{
        \underline{$Q=Q\cup\{q_i\}$;$T=T\cup\{T_i\}$}\tcp*{Appending qeuries with noisy negative outcomes}\label{line: end_appending}}
        }}
\end{algorithm}
\subsection{Overview}\label{subsec: overview}
In summary, our proposed enhanced SVT algorithm introduces three key improvements:

\textcircled{1} \textcircled{1} \textbf{The algorithm uses exponential noise for query perturbation}~(Line~\ref{line: exp-svt-query-noise}). This choice is based on the fact that the cumulative probability function of exponential noise tightly satisfies Equation~\ref{eq: lispchitz} in Theorem~\ref{theo: privacy-svt-lips}. Meanwhile, it exhibits a much smaller variance compared to other noise distributions considered, resulting in less data distortion and enhanced query accuracy.

\textcircled{2}: \textbf{An optimal threshold correction term is computed by maximizing the $(\alpha,\beta)$-accuracy of the SVT with exponential noise}~(Line \ref{line: exp-svt-threshold-correction}), and added to the noisy thresholds $\tilde{T}_i$~(Line \ref{line: exp-svt-query-perturb}). 
Correcting the threshold is crucial when using exponential noise (or other noise distributions centered at non-zero values) as it introduces bias into the query results\footnote{Note that correcting the threshold is equivalent to correcting the query results. In essence, as SVT compares $\tilde{q}_i-\tilde{T}_i$ with $0$. Therefore, adding the correction term to the predefined threshold $\tilde{T}_i$ is the same as subtracting it from the query result $\tilde{q}_i(D)$.}.
Intuitively, the correction term that maximizes the $(\alpha,\beta)$-accuracy ensures the highest success probability~(\ie, $1-\beta$) for SVT, thereby improving query accuracy.

\textcircled{3}: 
\textbf{An appending strategy has been developed}, where noisy negative queries~(\ie, $q_i(D)$ for $\tilde{q}_i(D)<\tilde{T}_i(D)$) are appended to the query queue for an additional round of querying~(Line~\ref{line: start_appending} to Line~\ref{line: end_appending}). 
Generally speaking, while our optimal threshold correction method effectively increases the precision of SVT, the appending strategy further increases its recall. The rationale behind this strategy is that for a positive query where $q_i(D)\geq T_{i}$, comparing the noisy threshold $\tilde{T}_i$ with $\tilde{q}_i(D)$ multiple times increases the likelihood of $q_i(D)$ identified as a positive query.

The concrete details of our developed methods and strategy are presented in Section~\ref{subsec: exp-svt}, Section~\ref{subsec: threshold-correction}, and Section~\ref{subsec: appending}.
Additionally, the privacy guarantee and the utility guarantee of Algorithm~\ref{alg: exp-svt} are provided in Theorem~\ref{theo: privacy-guarantee} and Theorem~\ref{theo: utility-guarantee}, respectively. The proof of these theorems is deferred to Appendix~\ref{sec: exp-privacy-proof} and Appendix~\ref{sec: exp-utility-proof}.

\begin{theorem}[Privacy guarantee]\label{theo: privacy-guarantee}
    Let $c$ denote the number of positive outcomes output by Algorithm~\ref{alg: exp-svt}. Algorithm~\ref{alg: exp-svt} satisfies $(\varepsilon_1+\varepsilon_2)$-differential privacy when \texttt{RESAMPLE} is set to \texttt{False}, and $(c\varepsilon_1,\varepsilon_2)$-differential privacy when \texttt{RESAMPLE} is set to \texttt{True}, where $\varepsilon_1$ and $\varepsilon_2$ are the privacy budgets for threshold and query perturbation, respectively.
\end{theorem}

\begin{figure}[tbp!]
    \centering
    \includegraphics[width=0.6\linewidth]{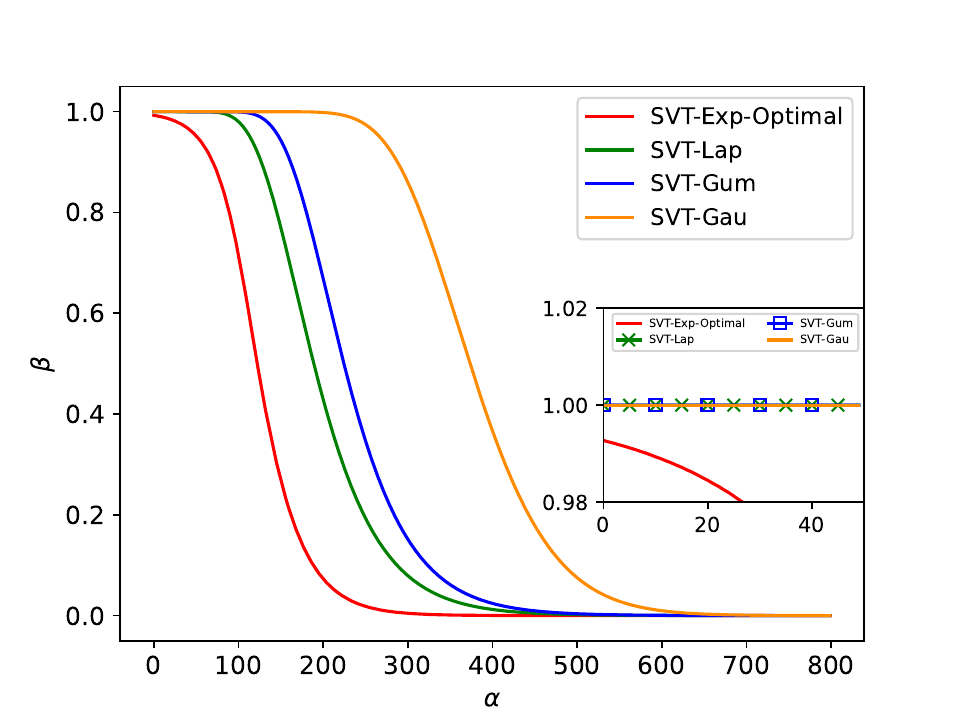}		
    \caption{The $(\alpha,\beta)$-accuracy of SVT algorithms with four different types of noise: the SVT with the exponential noise and the optimal threshold correction; SVT with the Laplacian noise; SVT with the Gumbel noise and the mean threshold correction; and SVT with Gaussian noise. Parameter $k$ in Theorem~\ref{theo: utility-guarantee} is set to $50$ and the overall privacy budget $\varepsilon=1$.}
    \label{fig: accuracy}
\end{figure}

For the utility guarantee, following the convention~\cite{dwork2014algorithmic,lyu2017understanding}, $(\alpha,\beta)$-accuracy is adopted as the utility metric. The parameters in Algorithm~\ref{alg: exp-svt} are set as $\Delta=1$, $\varepsilon_1=\varepsilon_2=\frac{\varepsilon}{2}$ and $c=1$ for ease of computation and comparison. As demonstrated in Theorem~\ref{theo: utility-guarantee}, compared to SVT algorithms where the Laplacian noise is adopted for query perturbation, Algorithm~\ref{alg: exp-svt} demonstrates a clear advantage.
\begin{theorem}[Utility guarantee]\label{theo: utility-guarantee}
    Given any $k$ records such that $\lvert \{i<k:d_i\geq t-\alpha\}\rvert=0$~(\ie, the record above and closest to the threshold is the last one), Algorithm~\ref{alg: exp-svt} is at least $(\frac{4(\ln k+\ln \frac{2}{\beta})}{\varepsilon},\beta)$-accurate.
\end{theorem}

Additionally, note that the assumptions made on parameter settings in Theorem~\ref{theo: utility-guarantee} may not fully align with real-world scenarios. Therefore, we further provide a numerical utility analysis under a more practical setting in Figure~\ref{fig: accuracy}, which shows that the value of $\beta$ of Algorithm~\ref{alg: exp-svt} is consistently below that of other baselines for a fixed $\alpha$, further demonstrating the effectiveness of our method.

\subsection{Query Perturbation with Exponential Noise}\label{subsec: exp-svt}
In this section, we justify the use of exponential noise in Algorithm~\ref{alg: exp-svt} for the following two reasons:

First,
exponential noise guarantees DP for SVT, as its cumulative function,
$$F(x;\lambda)=\begin{cases}1-\exp{(-\lambda x)} & x\geq0\\0&x<0\end{cases},$$
tightly satisfies Equation~\ref{eq: lispchitz} with $k_2=\lambda$ for any $\lambda>0$. However, note that the exponential distribution does not satisfy Equation~\ref{eq: pure_dp}, and therefore cannot be used for threshold perturbation. Since the Laplace distribution is one of the most commonly and frequently adopted noise distributions in SVT, we use it to perturb the predefined threshold, following convention.

\begin{figure}[tbp]
    \centering
    \includegraphics[width=0.6\linewidth]{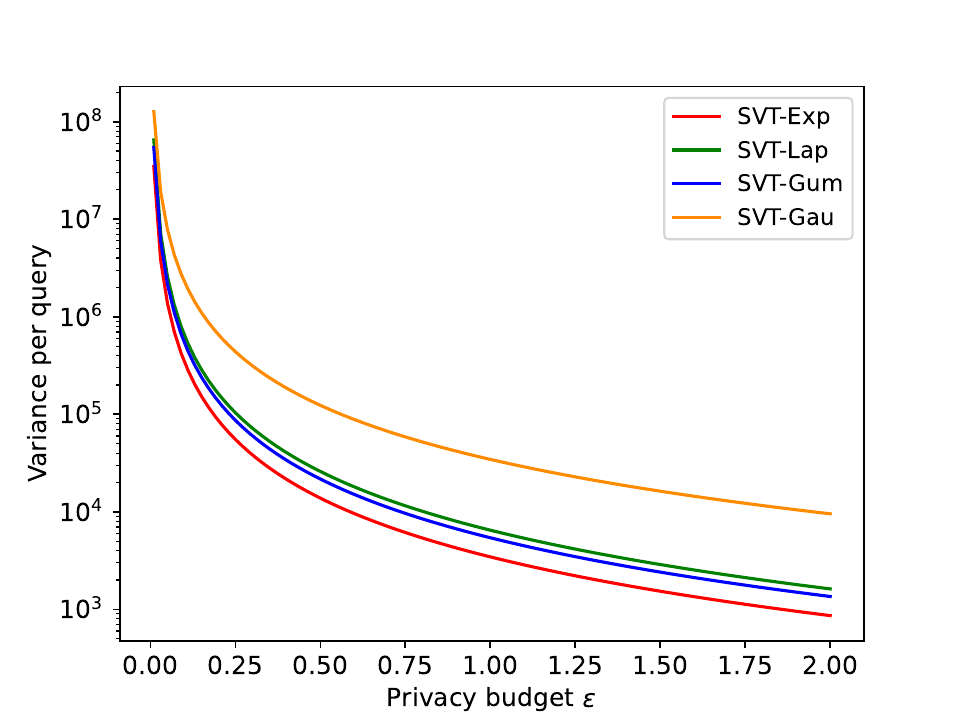}			
    \caption{The variance of SVT with four different types of noise~(\ie, exponential noise, Laplacian noise, Gumbel noise, and Gaussian noise). The privacy budget $\varepsilon$ varies from $0.01$ to $2$. $c$ is set to $50$. The y-axis is log-based. Note that the threshold correction does not affect the query variance.}
    \label{fig: variance}
\end{figure}

Second, exponential noise yields the smallest variance on the private comparison result~(\ie, $\tilde{q}_{i}(D)\geq \tilde{T}_i+r^{op}$) in Algorithm~\ref{alg: exp-svt} among all evaluated noises, yielding the least data distortion among all compared methods.
Specifically, as also stated in \cite{lyu2017understanding}, the variance of the private comparison result $V$ is written as follows:
\begin{equation}\label{eq: variance}
    V=\texttt{Var}\left(\texttt{Lap}\left(\frac{\Delta}{\varepsilon_1}\right)\right)+\texttt{Var}\left(\texttt{Exp}\left(\frac{2c\Delta}{\varepsilon_2}\right)\right).
\end{equation}
A smaller $V$ indicates a smaller data distortion, thus a potentially higher query accuracy.
Given an overall privacy budget $\varepsilon$, by varying $w$ such that $\varepsilon_2=w\varepsilon_1$ and $\varepsilon=\varepsilon_1+\varepsilon_2$, $V$ can be minimized for any fixed pair of noise distributions~\cite{lyu2017understanding}~(Cf. Appendix~\ref{sec: privacy-allocation-proof}).
As shown in Figure~\ref{fig: variance} where the minimal $V$ for different $\mathcal{N}_2$ in Algorithm~\ref{alg: exp-svt} are depicted, 
exponential distribution yields the smallest private comparison variance.
\subsection{Optimal Threshold Correction}\label{subsec: threshold-correction}
In this section, we elaborate on our optimal threshold correction by: (1) providing the motivations necessitating our design~(Cf. Section~\ref{subsec: correction_motivation}); (2) detailing our correction methodology, which incorporates derivation of the success probability of SVT and computation of the optimal correction term~(Cf. Section~\ref{subsec: correction_methodology}); and (3) presenting a numerical computation framework that generalizes our correction method to broader applications~(Cf. Section~\ref{subsec: correction_numerical}).

\subsubsection{Motivations.}\label{subsec: correction_motivation}
Though the exponential noise demonstrates properties that can theoretically boost query accuracy, applying it to SVT is challenging due to the introduced bias.
Specifically, the expectation of the difference between a noisy query and its corresponding noisy threshold is written as follows: $$E\left(\tilde{q}_i(D)-\tilde{T}_i\right)=q_i(D)-T_i+\frac{2c\Delta}{\varepsilon_2}.$$
Due to the introduction of a positive bias $\frac{2c\Delta}{\varepsilon_2}$, the algorithm is likely to mistakenly classify queries $q_{i}(D)$ as positives if they fall within the range $T_i\geq q_{i}(D)\geq T_i-\frac{2c\Delta}{\varepsilon_2}$.
As a result, the query accuracy is compromised for two reasons. First, as mentioned above, the false positive rate increases. Second, as described in Algorithm \ref{alg: exp-svt}, to limit the overall privacy budget consumption, no more than $c$ records are allowed to be identified as positive queries. Since the false positive queries occupy these $c$ spots, the subsequent true positive queries are unable to be output, thereby decreasing the precision.

A na\"ive solution is to directly subtract the bias from each $q_{i}(D)$ or add the bias to each $T_i$~\cite{wang2017locally,ye2019privkv,liu2022collecting,liu2023echo}.
However, this method yields a sub-optimal performance when applied to SVT. The primary reason is that this approach is designed for correcting the aggregation of $n$ noisy query results, whereas our focus is on correcting each noisy query result individually.
According to the law of large numbers, the mean of the aggregation with larger $n$ tends to converge to its expectation. 
That is to say, as $n$ increases, the mean of the aggregated results approaches the noise expectation more closely. Consequently, correcting the aggregate of noisy query results by subtracting the expectation often achieves satisfactory accuracy. However, since our approach targets individual noisy results~(\ie, $n=1$), this method does not provide the same level of performance in SVT as it does with aggregated noisy queries.

\subsubsection{Methodology.}\label{subsec: correction_methodology}
Motivated by this, we have developed a new correction method that focuses on improving the precision of the output binary vector rather than correcting each individual query result.
At a high level, inspired by the definition of $(\alpha,\beta)$-accuracy~(Cf. Definition~\ref{def: svt-acc}), we compute the success probability of the SVT algorithm by multiplying the probability of each query being corrected classified. We then derive an optimal threshold correction term by maximizing this success probability.
In essence, our optimal correction method increases the threshold value to enhance the precision of SVT, which is further elucidated in the analysis presented in Figure~\ref{fig: optimal-correction-term}.

\textbf{Success probability of SVT.} We begin by considering a simple case where $c=1$. Assume there are $k$ true negative queries before the true positive query. Given a fixed error tolerance parameter $\alpha$, the success probability of Algorithm~\ref{alg: exp-svt} with a threshold correction term $r$ is defined as follows:
\begin{equation}\label{eq: optimal_threshold_correction}
\begin{aligned}
        p(r) &= \prod_{i=1}^{k}\Pr\left[\tilde{q}_{i}(D)-\alpha<\tilde{T}_i+r\right]\cdot\Pr\left[\tilde{q}_{j}(D)+\alpha\geq\tilde{T}_j+r\right],
\end{aligned}
\end{equation}
where $\tilde{q}_{i/j}(D)=q_{i/j}(D)+\texttt{Exp}(\frac{1}{\lambda})$, and $\tilde{T}_{i/j}=T_{i/j}+\texttt{Lap}(b)$. Note that here we abuse the notion $\texttt{Exp}\left(\cdot\right)$ and $\texttt{Lap}\left(\cdot\right)$ to denote the random variables drawn from exponential distribution and Laplace distribution, respectively.

\textbf{Takeaways from Equation~\ref{eq: optimal_threshold_correction}.}
\underline{First}, Equation~\ref{eq: optimal_threshold_correction} is directly related to the $(\alpha,\beta)$-accuracy. Specifically, for a fixed $\alpha$, the value of $\beta$ with a threshold correction term $r$ is given by $1-p(r)$. \underline{Second}, by maximizing $p(r)$, we maximize the query accuracy of SVT, which in turn improves empirical performance. \underline{Third}, Equation~\ref{eq: optimal_threshold_correction} can be easily extended to scenarios where $c>1$ by: \textcircled{1} dividing all queries into subroutines, each containing exactly one positive query; and \textcircled{2} multiplying the success probability of each subroutine.

However, Equation~\ref{eq: optimal_threshold_correction} cannot be directly adopted as it is data-dependent, which could lead to privacy leakage.
To address this issue, we adopt a worst-case assumption by setting $q_{i/j}(D)=T_{i/j}$.
Thus, the success probability $p(r)$ can be expressed as:
\begin{equation}\label{eq: success-rate}
\begin{aligned}
    p(r) = \left(\Gamma\left(r+\alpha\right)\right)^{k}\cdot\left(1-\Gamma\left(r-\alpha\right)\right),
\end{aligned}
\end{equation}
where $\Gamma\left(\cdot\right)$ is the cumulative distribution function~(CDF) of the random variable $Z=X-Y$. Here $X$ and $Y$ are random variables drawn from distribution $\texttt{Exp}\left(\frac{1}{\lambda}\right)$ and $\texttt{Lap}\left(b\right)$, respectively. The shape of the probability density function~(PDF) of $Z$ is illustrated in Figure~\ref{fig: pdf-z}, and the explicit expression of $p\left(r\right)$ is detailed in Appendix~\ref{sec: pr_analytical}.

\textbf{Optimal correction term.} After demonstrating that the maximum value of $p(r)$ exists~(Cf. Appendix~\ref{sec: proof-of-maximum-pr}), the optimal threshold correction is obtained by:
\begin{equation}\label{eq: optimal-term}
    r^{op} = \argmax p(r).
\end{equation}
\begin{figure}[tbp!]
    \centering
    \subfigure[The probability distribution of the random variable $Z=X-Y$ with $X\sim\mathrm{Lap}(\frac{1}{\varepsilon_1})$ and $Y\sim\mathrm{Exp}(\frac{1}{\varepsilon_2})$. The parameters $\alpha$ and $k$ are set to $0$ and $10$, respectively.]{
    \includegraphics[width=0.4\linewidth]{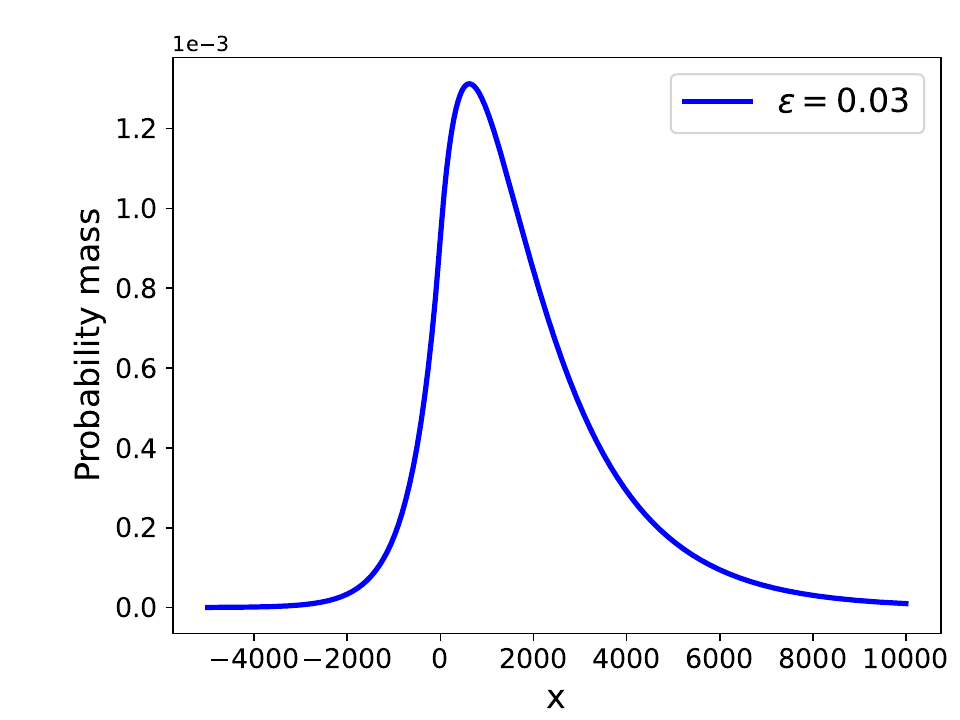}\label{fig: pdf-z}
    }
    \subfigure[$p(r)$ with varying $k$. The parameter $\alpha$ is set to $0$ and the overall privacy budget $\varepsilon$ is set to $0.1$.]{
    \includegraphics[width=0.4\linewidth]{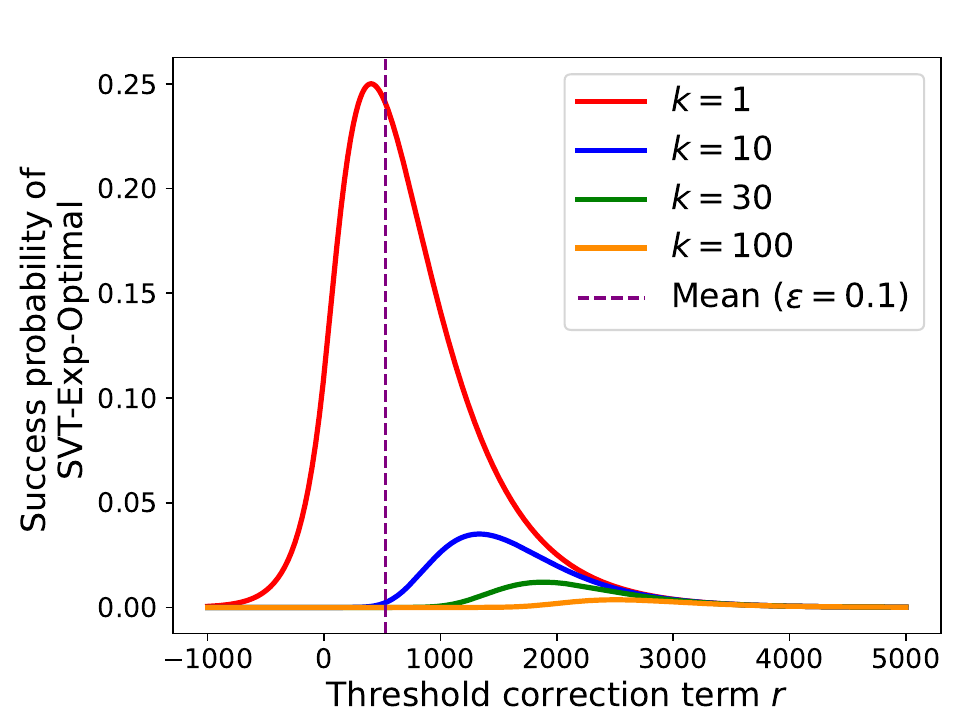}\label{fig: correction-with-k}
    }
    \subfigure[$p(r)$ with varying $\alpha$. The overall privacy budget $\varepsilon$ is set to $0.1$ and the parameter $k$ is set to $10$.]{
    \includegraphics[width=0.4\linewidth]{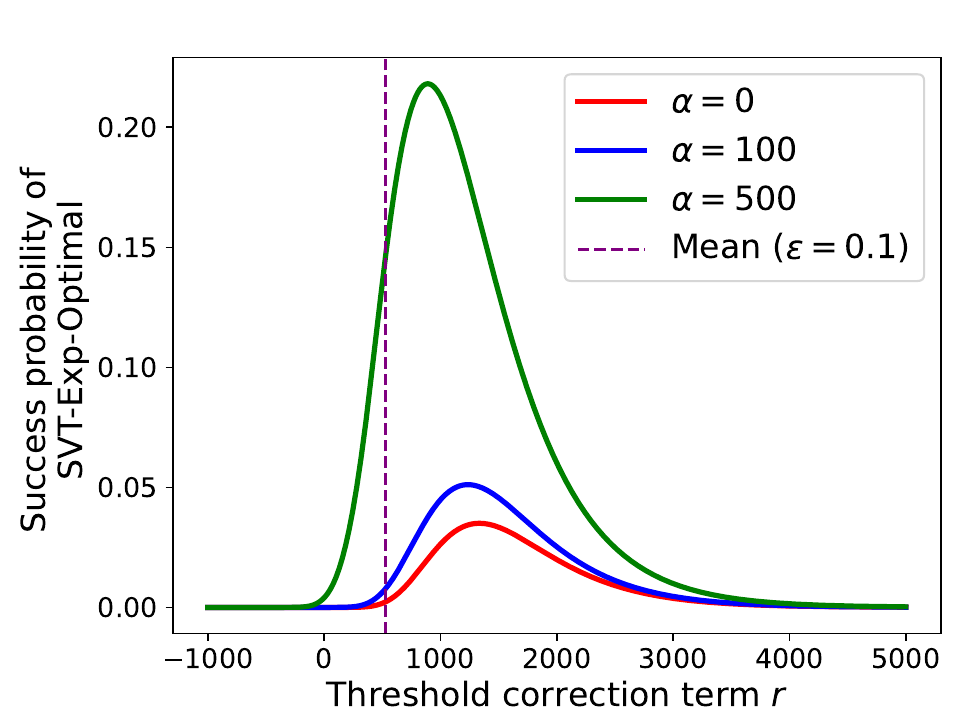}\label{fig: correction-with-alpha}
    }
    \subfigure[$p(r)$ with varying overall privacy budget $\varepsilon$. The parameters $\alpha$ and $k$ are set to $0$ and $10$, respectively.]{
    \includegraphics[width=0.4\linewidth]{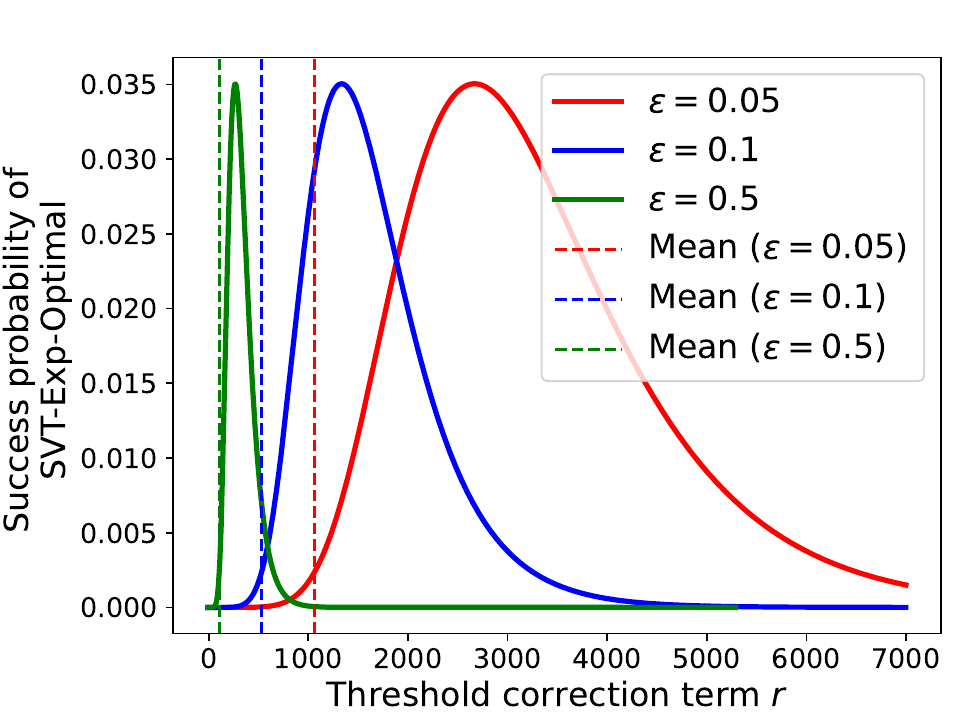}\label{fig: correction-with-eps}
    }
    \caption{The numerical analysis of Equation \ref{eq: success-rate}. Dashed lines are the mean of the exponential noise distribution. Each privacy $\varepsilon$ has a corresponding mean.}
    \label{fig: optimal-correction-term}
\end{figure}
To offer a better understanding, we perform a detailed analysis of Equation~\ref{eq: success-rate} and Equation~\ref{eq: optimal-term}~(Cf. Figure~\ref{fig: optimal-correction-term}), and the key findings are summarized as follows:

\underline{First}, as shown in Figure~\ref{fig: correction-with-k}, an increasing number of negative queries~(k) leads to a larger threshold correction term~($r^{op}$). As $k$ grows, SVT is more likely to halt before the last true positive query, which is referred to as `early halting'. Consequently, a larger correction term is necessary to mitigate early halting, thereby enhancing the precision of the SVT algorithm.
    
\underline{Second}, as shown in Figure~\ref{fig: correction-with-alpha}, an increase in the error tolerance parameter $\alpha$ results in a higher success probability and a smaller threshold correction term~$r^{op}$. The rationale is that a larger $\alpha$ allows more false negative or false positive queries during the query process. Consequently, a smaller threshold correction term is sufficient to achieve a high success probability. Note that $\alpha$ is a hyperparameter chosen by the data analysts. For generality, this work defaults to $\alpha=0$, irrespective of specific data distributions.

\underline{Third}, as shown in Figure~\ref{fig: correction-with-eps}, an increase in the privacy budget results in a smaller threshold correction term. This is because a larger privacy budget leads to lower noise variances, which reduces data distortion and results in noisy query results that are more tightly concentrated around their expectations. Consequently, a smaller threshold correction term is adequate to prevent the early-halting issue and achieve an optimal success probability.

\subsubsection{Numerical Computation Framework.}\label{subsec: correction_numerical}
To improve the scalability of the optimal threshold correction method, we further propose a numerical computation framework, which is summarized in Algorithm~\ref{alg: threshold_correction} and Algorithm~\ref{alg: discretizer}. Instead of deriving the analytical formula for $\Gamma(\cdot)$ in Equation~\ref{eq: success-rate}, we first discretize the noise distributions $\mathcal{N}_1$ and $\mathcal{N}_2$ used in Algorithm~\ref{alg: exp-svt}. We then convolve these discretized distributions using the Fast Fourier Transform (FFT) to obtain the numerical representation of $\Gamma$. The optimal threshold correction term, $\bar{r}^{op}$, is determined by maximizing $\bar{p}(r)$, which is calculated based on this numerical $\Gamma(\cdot)$.
\begin{algorithm}[htbp]
    \caption{CorrectionTerm: Numerical threshold correction.}\label{alg: threshold_correction}
    \LinesNumbered
    \KwIn{$b$, $\lambda$, $\alpha$, $m$, $e$, $k$}
    $B= F_{Exp}^{-1}(1-e)$;\tcp{Boundary}\label{line: boundary}
    $\bar{L}=\texttt{Discretizer}(Lap(b),m,B)$;\tcp*{Discretize Laplace distribution}
    $\bar{E}=\texttt{Discretizer}(Exp(\frac{1}{\lambda}),m,B)$;\tcp*{Discretize exponential distribution}
    $\bar{\Gamma}=\texttt{CONVOLVE}(\bar{L},\bar{E})$;\tcp*{Convolution with FFT}
    Compute $\bar{r}^{op}$ that maximize $\left(\bar{\Gamma}(\bar{r}+\alpha)\right)^{k}\cdot(1-\bar{\Gamma}(\bar{r}-\alpha))$\;
    \Return $\bar{r}^{op}$\;
\end{algorithm}
\begin{algorithm}[htbp]
    \caption{Discretizer}\label{alg: discretizer}
    \LinesNumbered
    \KwIn{\texttt{CDF}, $m$, $B$}
    $u=\frac{B}{m-1}$\;\label{line: chunk}
    \For{i=-m+1 to m-2}{
    $h_i=\texttt{CDF}((i+1)\cdot u)-\texttt{CDF}(i\cdot u)$\;\label{line: prob_mass}
    }
    $h_{-\infty} = \texttt{CDF}((1-m)\cdot u)$\;\label{line: negative_bracket}
    $h_{+\infty}=1-\texttt{CDF}(m-1\cdot u)$\;\label{line: positive_bracket}
    $\bar{D}=\{h_{-\infty}, h_{i} \text{ for } i \text{ in } (-m,m)\cap\mathbb{Z},  h_{+\infty}\}$\;
    \Return $\bar{D}$\;
\end{algorithm}
Hereinafter, we use the combination of Laplace and exponential distributions as an example to illustrate the concrete steps of our framework.
First, to convert the continuous noise distributions into discretized sequences, we define a boundary $B$ within which most of the probability mass~(\eg, $1-e$ as indicated in Line~\ref{line: boundary} of Algorithm~\ref{alg: threshold_correction}) is concentrated.
Next, in Algorithm~\ref{alg: discretizer}, we discretize the event space within $B$ into $u$ chunks, with each chunk having a mesh size $m$ (Line~\ref{line: chunk}). The events in each chunk are aggregated into a new event with a probability mass $h_i$ (Line~\ref{line: prob_mass}). Additionally, the probability mass of events exceeding $B$ is assigned to the positive infinity bracket, while events below $-B$ are placed in the negative infinity bracket (Lines~\ref{line: negative_bracket} and Line~\ref{line: positive_bracket} in Algorithm~\ref{alg: discretizer}).
Subsequently, we use Fast Fourier Transform~(FFT)~\cite{nussbaumer1982fast} to compute the discretized distribution $\bar{\Gamma}$ by convolving the discretized Laplace distribution $\bar{L}$ and the discretized exponential distribution $\bar{E}$, where we define that $\infty+x=\infty$ for any $x\in \mathbb{R}$. 
Finally, the correction term $\bar{r}$ is determined using $\bar{\Gamma}$ based on Equation~\ref{eq: optimal-term}.

The primary advantage of this numerical computation framework is its strong scalability to various types of noise distributions. While we can derive the explicit formula for $\Gamma(\cdot)$ when using Laplace and exponential distributions for query and threshold perturbation, respectively (Cf. Appendix~\ref{sec: pr_analytical}), deriving such formulas becomes complex when combining other distributions, such as Gaussian and exponential. This numerical framework allows privacy practitioners to apply the optimal threshold correction method to any noise distributions.
It is important to note that our numerical framework may introduce additional computation costs due to the discretization and convolution steps. To provide further insights, we analyze the trade-off between the running time of our algorithm and the estimation accuracy of $\bar{p}(r)$ and $\bar{r}^{op}$ in Appendix~\ref{sec: efficiency_numerical}. In summary, our analysis demonstrates that the proposed method can achieve highly accurate estimations with minimal additional computation cost.

\subsection{Appending Strategy}\label{subsec: appending}
As discussed earlier, our threshold correction method enhances the performance of SVT by increasing the value of the threshold, which in turn improves the precision of the algorithm. However, a higher threshold can also result in a reduced recall. This occurs because true positive queries with relatively small results are less likely to exceed the adjusted threshold.

To further enhance recall, we propose an appending strategy. As shown in Figure~\ref{fig: framework} Step~\textcircled{7} and Algorithm~\ref{alg: exp-svt} from Line~\ref{line: start_appending} to Line~\ref{line: end_appending}, each query identified as negative by the algorithm is appended to the end of the query queue for another round of querying. In settings where there is an infinite number of queries~(\eg, streaming data analysis, where new data and queries continuously arrive), these queries may be randomly inserted back into the queue.

The rationale behind this strategy is twofold. First, as stated in Theorem~\ref{theo: privacy-guarantee}, re-querying the outputs identified as negative incurs no additional privacy cost.
Second, conducting multiple rounds of querying increases the likelihood of correctly identifying true positive queries that were misclassified initially. Simultaneously, it also enlarges the probability gap between positive queries correctly identified as positive and negative queries mistakenly classified as positive. Concretely, for each query $q_{i}(D)$, the probability of it identified as positive in each round is $$p_{i}=\Pr[q_{i}(D)+\texttt{Exp}(\frac{1}{\lambda})\geq T_{i}+r+\texttt{Lap(b)}].$$ 
After $t$ rounds of querying, this probability increases to $tp_{i}$. Also note that the larger the difference $q_{i}(D)-T_i$, the higher the $p_{i}$. Therefore, after multiple rounds, the probability of $q_{i}(D)$ being identified as positive is $t\left(p_{i}-p_{j}\right)$ higher than for $q_{j}\left(D\right)$, where $q_{i}(D)-T_{i}\geq q_{j}(D)-T_{j}$.

%% file: evaluation.tex
\section{Evaluation}\label{sec: evaluation}
Before presenting our main results in Section~\ref{subsec: evaluation_results}, we detail the used datasets in Section~\ref{subsec: datasets}, the adopted utility metrics in Section~\ref{subsec: evaluation_utility_metric}, the compared baselines in Section~\ref{subsec: baselines}, and the selected parameters for our experiments in~Section \ref{subsec: parameter_selection}.
\subsection{Datasets and Implementations}\label{subsec: datasets}
This work uses six different datasets: three real-world and three synthetic. Detailed information is provided in Table~\ref{tbl: datasets}.

\begin{table}[!htbp]
\caption{Dataset details. \# $\cdot$ denotes the number of $\cdot$, and the `Threshold' is the predefined threshold we use in SVT for each dataset.}\label{tbl: datasets}
    \centering
    \scalebox{0.88}{
    \begin{tabular}{lcccc}
    \toprule
        ~ & \# of records & \# of items & Threshold & Type\\ \hline
        Binary & - & 10,000 & 500 &\multirow{3}{*}{synthetic}\\
        Zipf & - & 10,000 & 200 & \\ 
        T40I10D100K~\cite{qty} & 100,000 & 942 & 11,850 &\\ \hline
        BMS-POS~\cite{zheng2001real} & 515,597 & 1,657 & 13,600 &\multirow{3}{*}{real-world} \\ 
        Kosarak~\cite{aggarwal2009frequent} & 990,002 & 41,270 & 10,500 &\\
        Adult~\cite{adult} & 48,843 & 123 & 200 & \\
        \bottomrule
    \end{tabular}}
\end{table}

For the real-world datasets and the T40I10D100K dataset, each item's score is based on its frequency across records. Specifically, for each item $I_{i}$ in the dataset, the score $s_i=\sum\limits_{j=1}^{n}\mathbbm{1}(I_j^{i}=1)$, where $I_{j}^{i}$ is an indicator of whether the record $R_j$ contains item $I_{i}$, and $n$ is the total number of records.
For the remaining synthetic datasets, scores are assigned directly. In the Binary dataset, each positive query is assigned a score of 1,000, while negative queries receive a score of 0.
In the Zipf dataset, each item's score is proportional to $\frac{1}{i}$, with the score calculated as $s_i=\frac{1}{i}\times10,000$.
We use $m$ to denote the total number of items.
Additionally, Figure~\ref{fig: dataset-properties} plots the item scores for each dataset for further insights.
\begin{figure}[tbp!]
    \centering
    \includegraphics[width=0.6\linewidth]{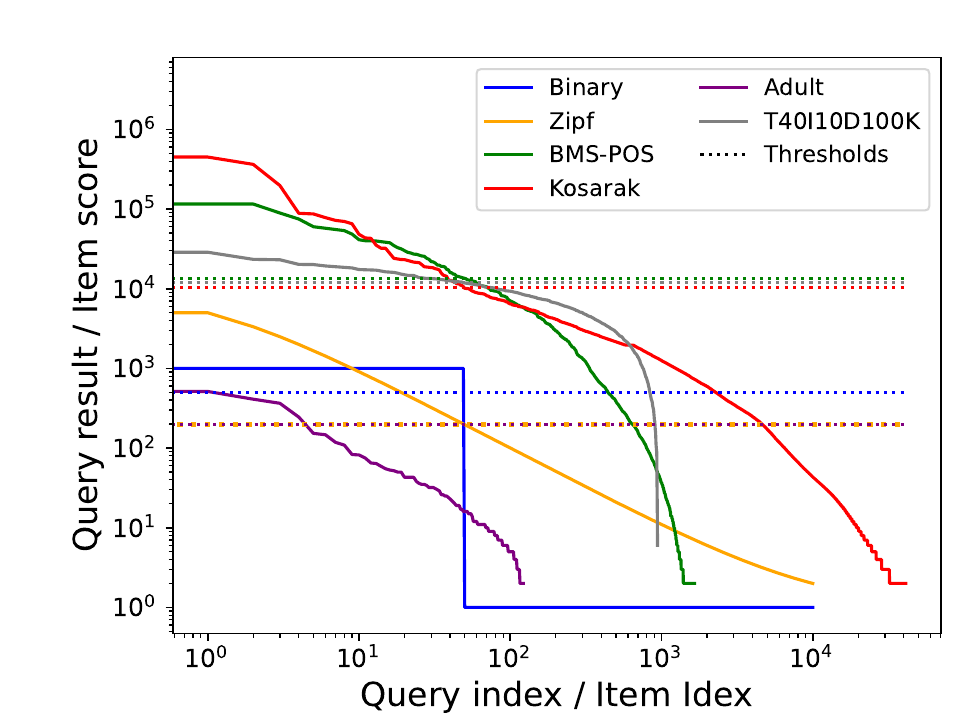}
    \caption{The scores of items of 6 datasets. All the items are in descending order based on their scores, which are plotted in a log-based manner for clarity. Dashed lines in Corresponding colors are the predefined thresholds we adopt in this work.}
    \label{fig: dataset-properties}
\end{figure}

All experiments are conducted on a laptop~(6-core Intel Core i7 CPU at 2.2 GHz with 16-GB RAM).

\subsection{Queries and Utility Metric}\label{subsec: evaluation_utility_metric}
\subsubsection{Queries}
The effectiveness of our proposed method is general but is validated here within the top-$c$ selection problem, a common application for SVT \cite{lyu2017understanding, zhu2020improving}. In this scenario, we use SVT to approximately query items with the top-$c$ highest scores in each dataset. 

Specifically,
we first shuffle the items randomly, then query each item's score~(\ie, $q_i(D)=s_{i}$) and compare it to the threshold $T$.
If $q_i(D)\geq T$, the corresponding query index $i$ is output. The algorithm halts either when $c$ indices is output or the number of queries reaches a maximum limit $k_max$. The choice of $T$ is crucial for query accuracy, but determining an appropriate threshold is beyond this work's scope. Various methods for threshold determination are discussed in the literature~\cite{lee2014top, carvalho2020differentially}.

\subsubsection{Utility Metrics}
We evaluate SVT's performance on the top-$c$ selection problem using two primary metrics: \textbf{F1-score}~\cite{fawcett2006introduction} and normalized cumulative rank (NCR).
The F1-score is computed as:
$$F1=\frac{2\cdot\text{precision} \cdot \text{recall}}{\text{precision} + \text{recall}}=\frac{2TP}{2TP+FP+FN},$$
where $TP$ is the number of true positive queries, $FP$ is the number of false positive queries,and $FN$ is the number of false negative queries.
However, the F1-score is more suited for unordered settings, where missing a top result incurs the same penalty as missing a lower-ranked result~\cite{lyu2017understanding}. To address this, we also use the Normalized Cumulative Rank (NCR)~\cite{lyu2017understanding}. 
In NCR, each query $q_i$ is assigned a rank score defined as follows: the top query~(\ie, top-$1$) receives a score of $c$, the next receives $c-1$, and so on~\cite{lyu2017understanding}. Queries below the threshold receive a score of $0$. The total score for positive outcomes is then normalized to the range $[0,1]$ by dividing by $\frac{c(c+1)}{2}$, the maximum possible score.
Since the results for the F1-score demonstrate similar trends as NCR, detailed F1-score results are provided in Appendix~\ref{sec: eva_f1}.

\subsection{Baselines}\label{subsec: baselines}
Table~\ref{tbl: baselines} lists all the methods compared in this work. We evaluate our method, which uses exponential noise for query perturbation, against \textbf{SVT-Lap}~\cite{lyu2017understanding}, \textbf{SVT-Gau}~\cite{zhu2020improving}, and \textbf{SVT-Gum}. While SVT-Lap and SVT-Gau are two of the most frequently used SVT variants in the literature, SVT-Gum is a new variant proposed in this work. SVT-Gum meets the constraints in Theorem~\ref{theo: privacy-svt-lips} but has a slightly larger variance compared to SVT-Exp. More details about SVT-Gum are provided in Appendix~\ref{sec: privacy-svt-gum}.
Additionally, to provide more insights, we also compare our method with the `upper bound' of query accuracy, which is obtained by directly ranking the noisy query results perturbed with random noise drawn from $\texttt{Exp}\left(\frac{\Delta}{\varepsilon_2}\right)$.
Regarding threshold correction, we compare SVT-Exp with optimal threshold correction to SVT-Exp with no threshold correction and to SVT-Exp with mean correction~(\ie, the na\"ive solution in Section~\ref{subsec: correction_motivation}).

\begin{table}[!ht]
\caption{Baseline methods. \texttt{Lap}, \texttt{Exp}, \texttt{Gau}, and \texttt{Gum} refer to Laplace, exponential, Gaussian, and Gumbel distribution, respectively. 'Optimal' represents our optimal threshold correction method, while 'Mean' represents the na\"ive correction method by subtracting the mean of the noise.}\label{tbl: baselines}
    \centering
    \scalebox{0.9}{
    \begin{tabular}{cccc}
    \toprule
        ~ & \makecell{Noise \\ for query} & \makecell{Noise \\ for threshold} & \makecell{Threshold \\ correction}  \\ \hline
        \makecell{SVT-Exp upper bound} & - & \texttt{Exp} & -  \\ 
        \makecell{SVT-Exp (optimal)~(Alg. \ref{alg: exp-svt})} & \texttt{Lap} & \texttt{Exp} & Optimal  \\ 
        \makecell{SVT-Exp (mean)} & \texttt{Lap} & \texttt{Exp} & Mean  \\ 
        \makecell{SVT-Exp (no)} & \texttt{Lap} & \texttt{Exp} & No  \\ 
        \makecell{SVT-Gumbel} & \texttt{Lap} & \texttt{Gum} & Mean  \\ 
        SVT-Lap & \texttt{Lap} & \texttt{Lap} & No  \\ 
        SVT-Gau & \texttt{Gau} & \texttt{Gau} & No \\
        \bottomrule
    \end{tabular}}
\end{table}

\subsection{Parameter Selection}\label{subsec: parameter_selection}
This work involves several privacy parameters: the failure rate $\delta$, the overall privacy budget $\varepsilon$, the privacy budget for threshold perturbation~(\ie, $\varepsilon_1$), and the privacy budget for query perturbation~(\ie, $\varepsilon_2$).
For SVT-Gau, $\delta$ is set to $\frac{1}{n}$, following convention~\cite{dwork2014algorithmic,liu2022collecting}, while in other settings, $\delta$ is set to $0$.
As also mentioned in Section~\ref{subsec: exp-svt}, $\varepsilon$ is divided into $\varepsilon_1$ and $\varepsilon_2$ such that $\varepsilon_2=w\varepsilon_1$ and $\varepsilon=\varepsilon_1+\varepsilon_2$. The parameter $w$ is computed by minimizing Equation~\ref{eq: variance}~\cite{lyu2017understanding}. Table~\ref{tbl: privacy-alloation} lists $w$ for different baselines, with its proof provided in Appendix~\ref{sec: privacy-allocation-proof}.
\begin{table}[!htbp]
\caption{Optimal Privacy Allocation. While $\varepsilon$ is the overall privacy budget consumption, $\varepsilon_1$ and $\varepsilon_2$ are the privacy budget for threshold perturbation and query perturbation, respectively. $w$ is an indicator of privacy budget allocation.}\label{tbl: privacy-alloation}
    \centering
    \scalebox{1}{
    \begin{tabular}{ccccc}
    \toprule
        ~ & SVT-Exp & SVT-Gum & SVT-Lap & SVT-Gau\\ \midrule
        $\varepsilon$ & \multicolumn{4}{c}{$\varepsilon = \varepsilon_1+\varepsilon_2,\ \varepsilon_2=w\varepsilon_1$}\\ \midrule
        $w$ & $\left(\sqrt{2}c\right)^{2/3}$ & $(\frac{\pi c}{\sqrt{3}})^{2/3}$ & $\left(2c\right)^{2/3}$ &$(2c)^{2/3}$ \\ \bottomrule
    \end{tabular}}
\end{table}

\subsection{Evaluation Results}\label{subsec: evaluation_results}
\begin{figure*}[tbp!]
    \centering
    \subfigure[Binary dataset]{
    \includegraphics[width=7cm]{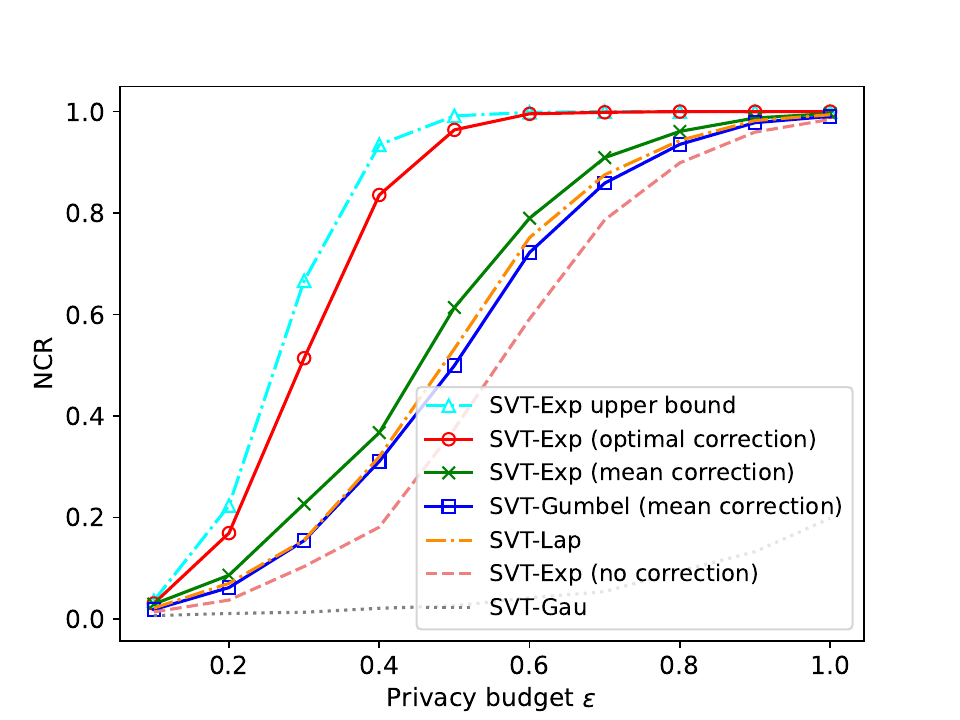}
    }
    \subfigure[Zipf dataset]{
    \includegraphics[width=7cm]{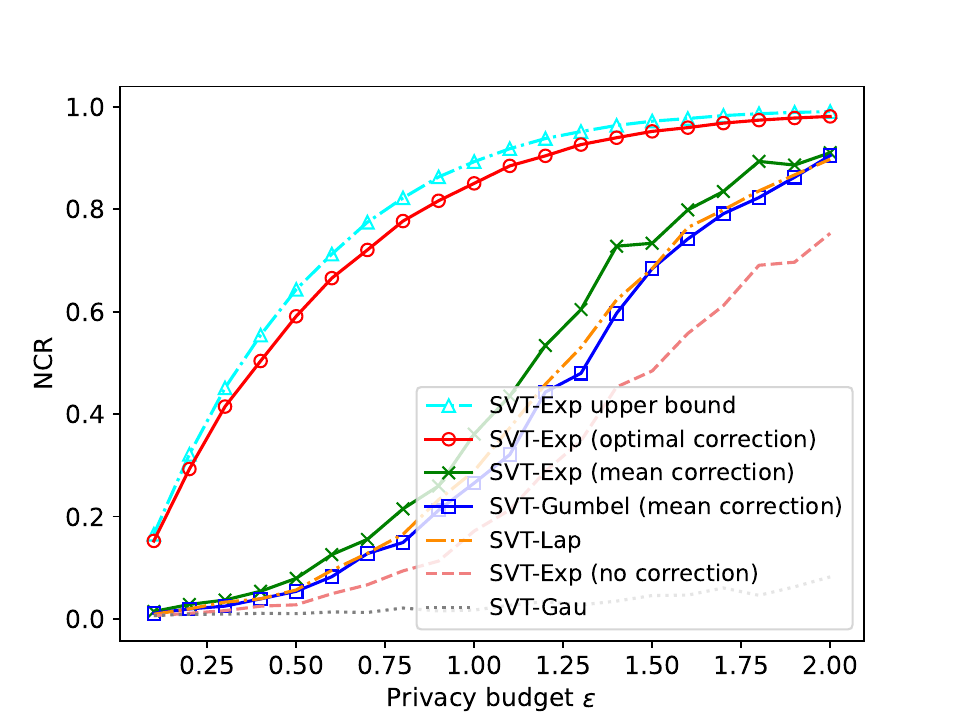}\label{fig: overall-zip}
    }
    \subfigure[BMS-POS dataset]{
    \includegraphics[width=7cm]{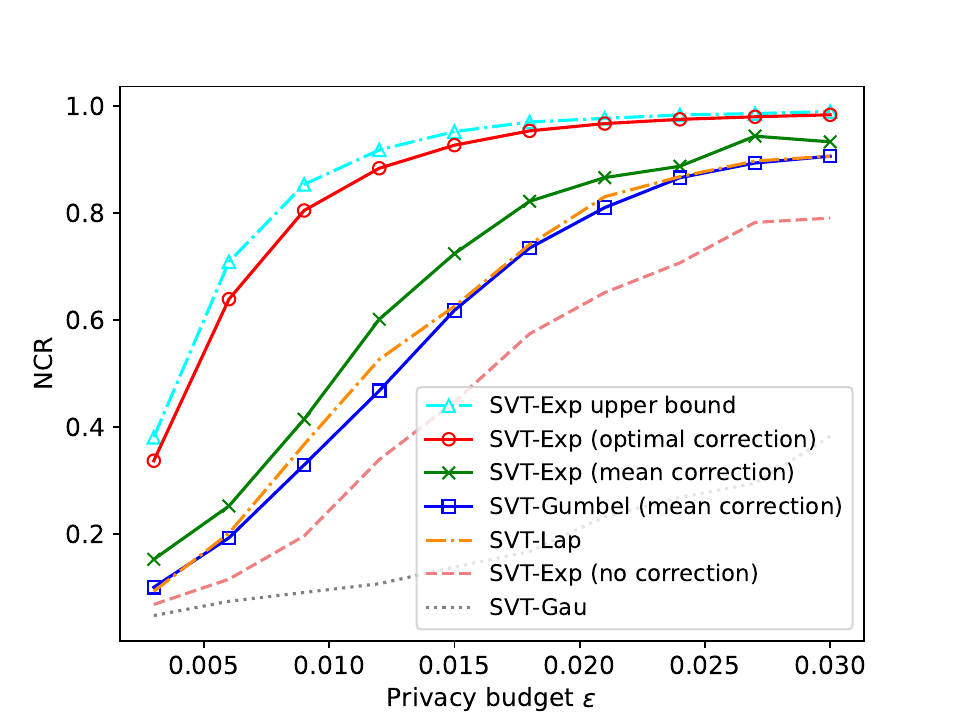}
    }
    \subfigure[Kosarak dataset]{
    \includegraphics[width=7cm]{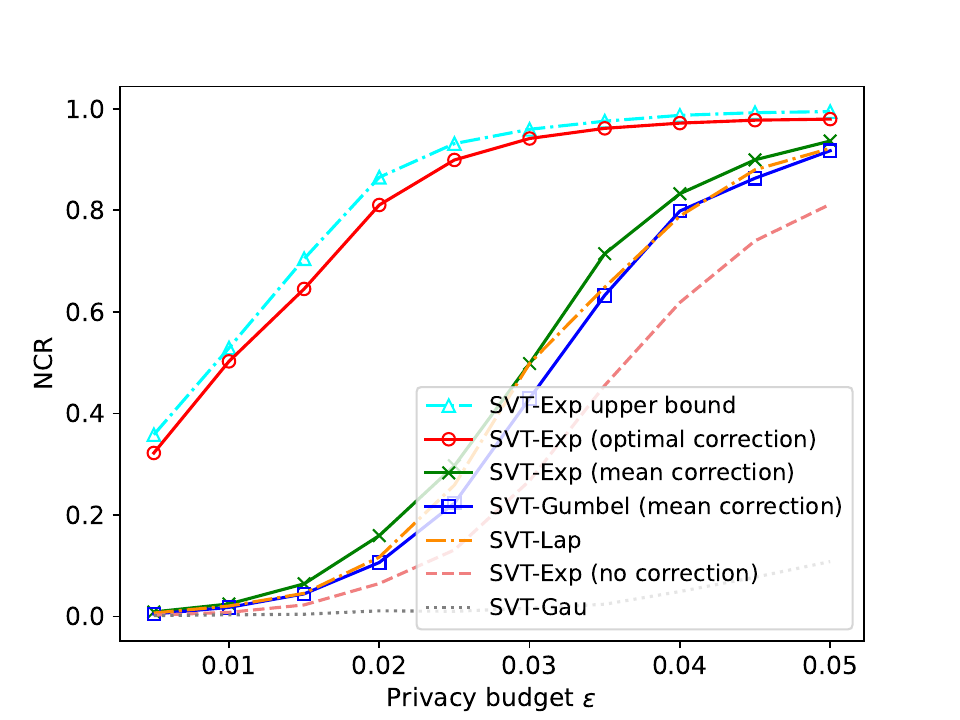}
    }
    \subfigure[Adult dataset]{
    \includegraphics[width=7cm]{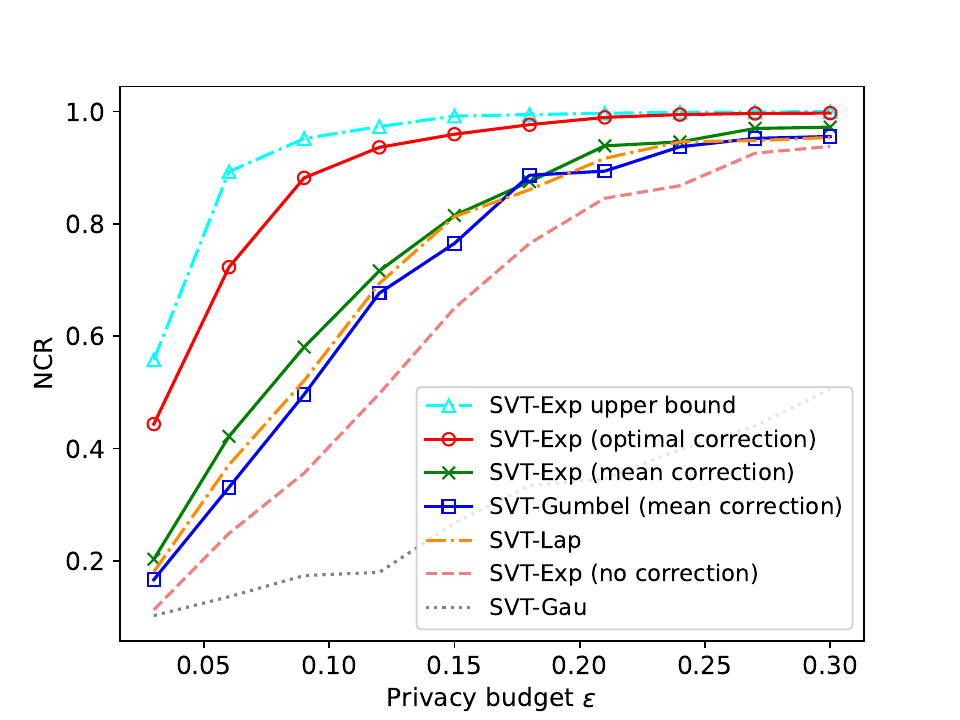}
    }
    \subfigure[T40I10D100K dataset]{
    \includegraphics[width=7cm]{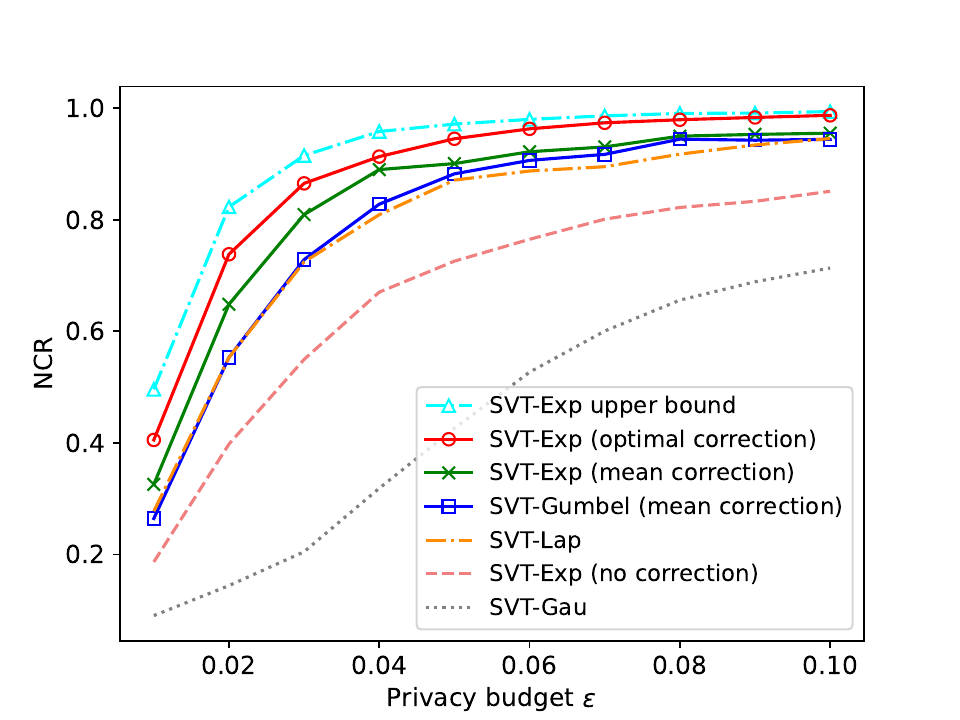}
    }
    \caption{NCR on six datasets with $c=5$ for Adult dataset and $c=50$ for the remaining datasets, $\alpha=0$, $k=\lfloor\frac{m}{c}\rfloor$, \texttt{RESAMPLE}=\texttt{False}, and \texttt{APPEND}=\texttt{True}. Sequential composition theorem is adopted for computing $\varepsilon$.}
    \label{fig: overall-50}
\end{figure*}

Hereinafter, we demonstrate the effectiveness of our proposed methods through experiments on six datasets, with an in-depth analysis of the evaluation results.
Our main results, presented in Figure~\ref{fig: overall-50}, show a significant advantage of Algorithm~\ref{alg: exp-svt} over other baselines. The correctness and effectiveness of our optimal threshold correction method are shown in Table~\ref{tbl: threshold-correction-terms} and Figure~\ref{fig: overall-50}. Finally, the effectiveness of our appending strategy and the trade-off it yields between efficiency and query accuracy is illustrated in Figure~\ref{fig: traverse}.

\subsubsection{Effectiveness of Algorithm~\ref{alg: exp-svt}.}
First, as illustrated in Figure~\ref{fig: overall-50}, our proposed method~(\ie, SVT-Exp (optimal correction), shown with red-circle line) significantly outperforms others baselines by up to $50\%$ on NCR across all tested privacy regions and datasets.
Moreover, compared to other baselines, the performance of our proposed method closely matches the empirical SVT-Exp upper bound, further demonstrating its effectiveness.
Second, our method shows a particular larger advantage when the gaps between predefined threshold and query results are relatively large~(Cf. Figure~\ref{fig: overall-zip} and Figure~\ref{fig: dataset-properties}). This is because our threshold correction better filters out true positive queries far above the threshold, while our appending strategy effectively distinguishes smaller true positives from true negatives far below the threshold.
Third, all compared variants, including ours, demonstrate higher NCR on datasets where the gaps between the threshold and query results are large even under smaller privacy budget~(\eg, $\varepsilon=0.05$ on Kosarak), as these datasets are generally more robust to noise.
Fourth, the NCR of other baselines decreases from SVT-Lap, SVT-Gumbel, to SVT-Gau, which aligns with our theoretical analysis in Figure~\ref{fig: accuracy} and Figure~\ref{fig: variance}.

\subsubsection{Effectiveness of the Optimal Threshold Correction}\label{subsec: evaluation-optimal-threshold-correction}
\begin{figure*}[tbp!]
    \centering
    \subfigure[Binary dataset]{
    \includegraphics[width=7cm]{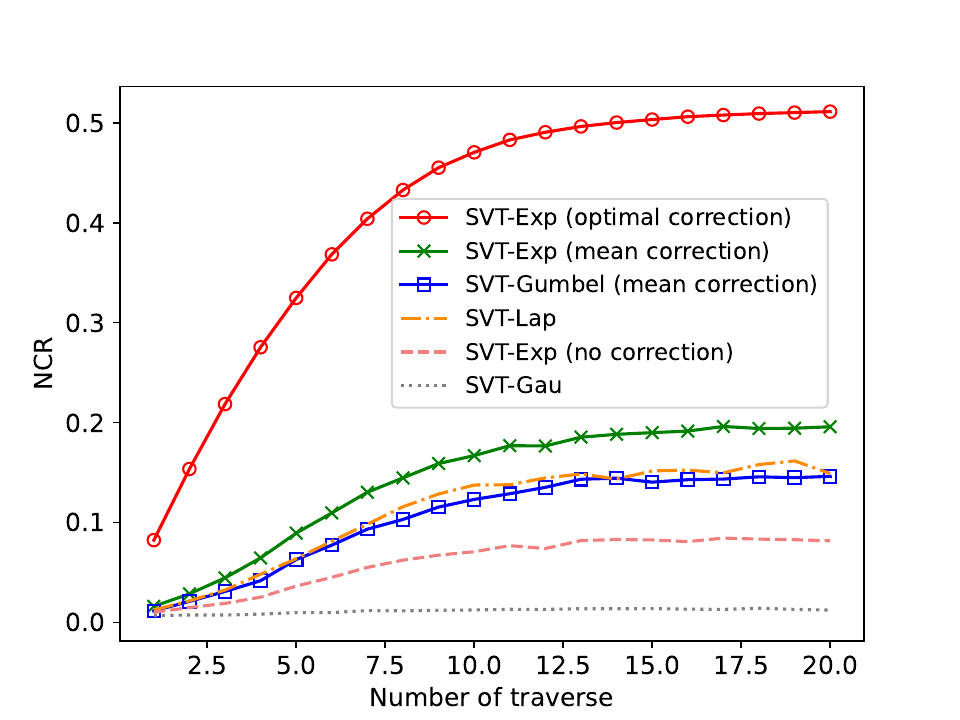}
    }
    \subfigure[Zipf dataset]{
    \includegraphics[width=7cm]{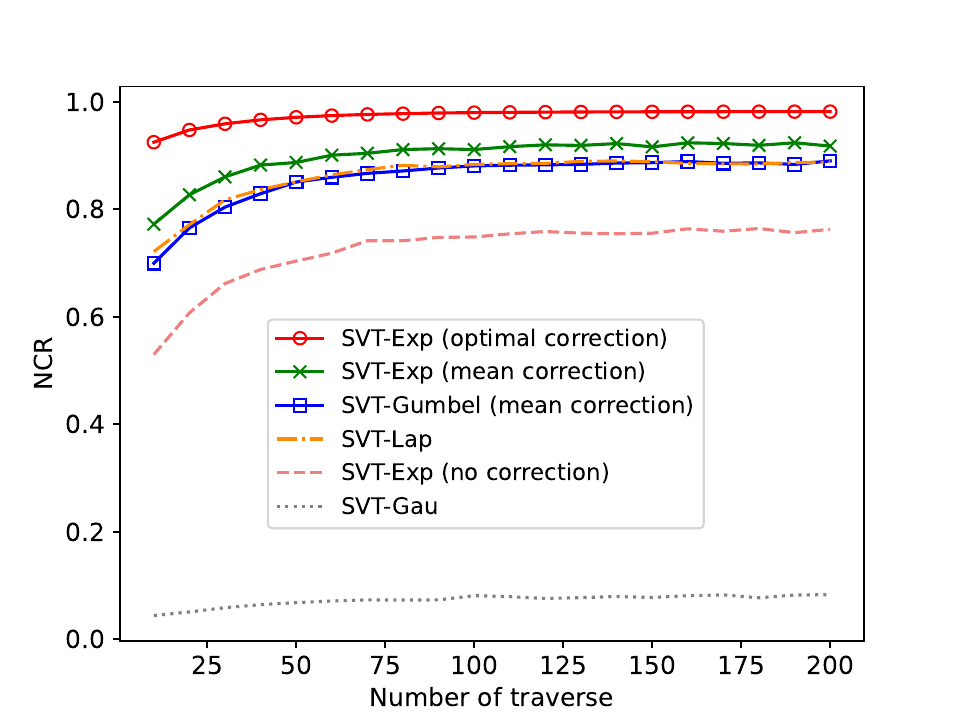}
    }
    \subfigure[BMS-POS dataset]{
    \includegraphics[width=7cm]{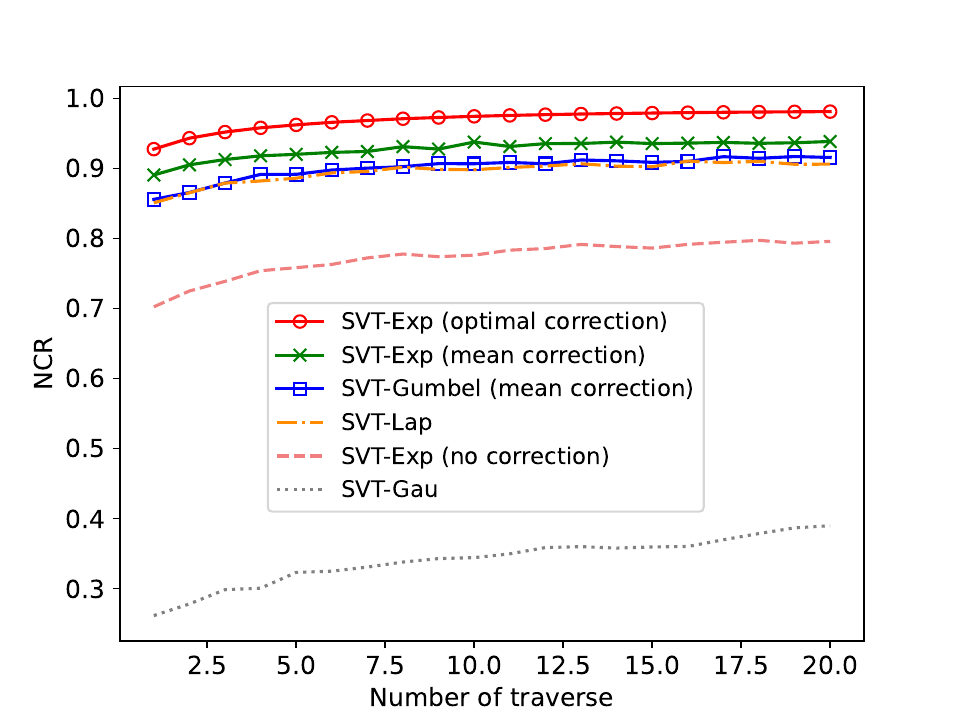}
    }
    \subfigure[Kosarak dataset]{
    \includegraphics[width=7cm]{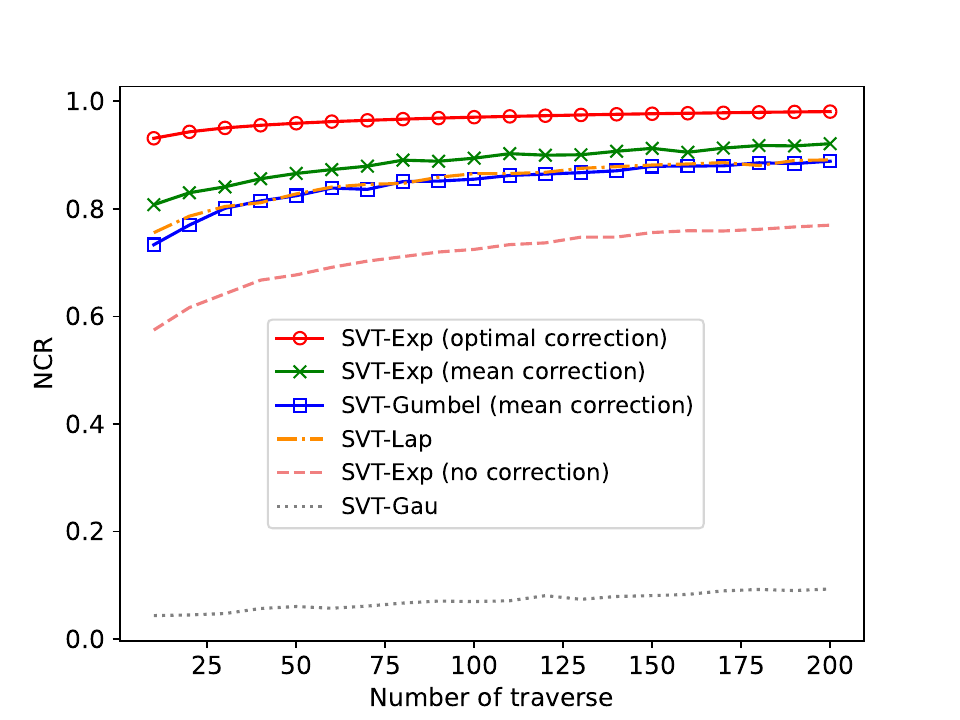}
    }
    \subfigure[Adult dataset]{
    \includegraphics[width=7cm]{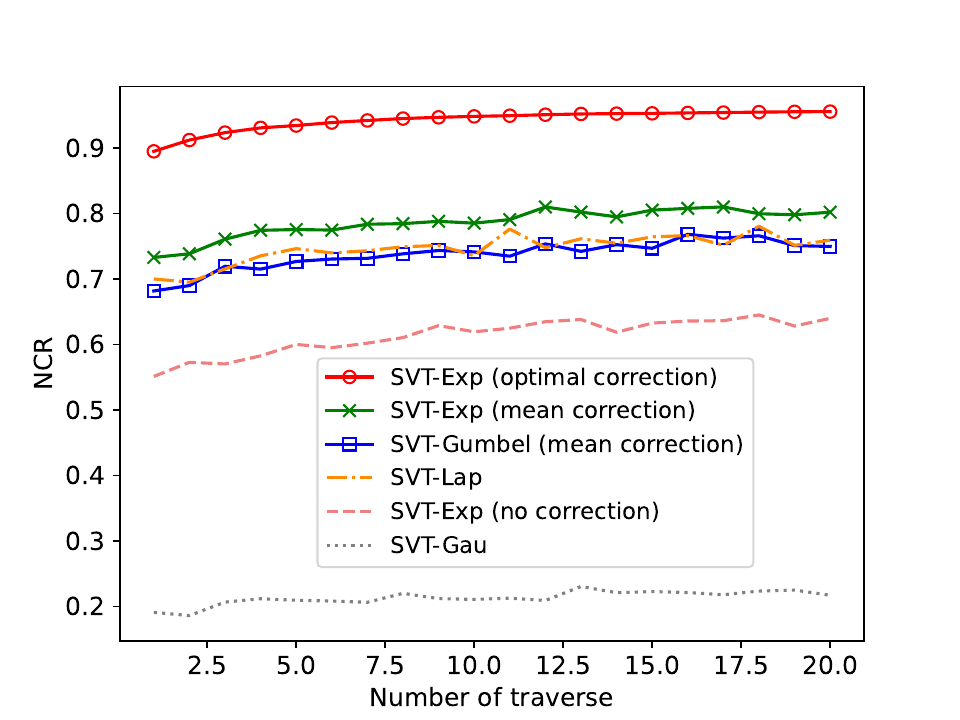}
    }
    \subfigure[T40I10D100K dataset]{
    \includegraphics[width=7cm]{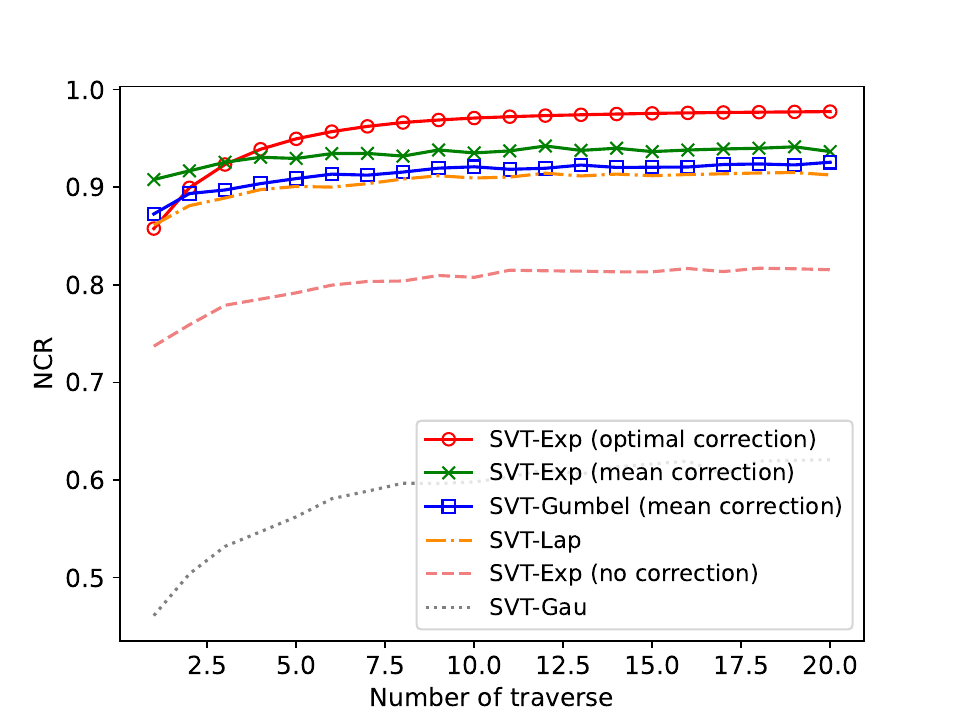}
    }
    \caption{NCR on six datasets with varying number of traverses. Parameter $c$ is set to $5$ for Adult and $50$ for the remaining datasets, $\alpha=0$, $k=\lfloor\frac{m}{c}\rfloor$. \texttt{RESAMPLE}=\texttt{False}, and \texttt{APPEND}=\texttt{True}. Sequential composition theorem is adopted for computing $\varepsilon$.}
    \label{fig: traverse}
\end{figure*}
\begin{table}[!htbp]
    \caption{Comparison of the threshold correction term between the optimal threshold correction method and the mean correction method, where $c=50$ and $\alpha=0$.}
    \centering
    \begin{tabular}{cccccc}
    \toprule
        $\varepsilon$ & 0.01 & 0.05 & 0.1 & 1 & 2  \\ \midrule
        Optimal & 35450.45 & 7147.15 & 2803.30 & 280.78 & 138.89  \\ 
        Mean & 5332.08 & 1046.42 & 523.21 & 52.32 & 26.16 \\ \bottomrule
    \end{tabular}
    \label{tbl: threshold-correction-terms}
\end{table}
We compare our SVT-Exp~(optimal correction) with SVT-Exp~(mean correction) and SVT-Exp~(no correction) in Figure~\ref{fig: overall-50}, while Table~\ref{tbl: threshold-correction-terms} presents the values of optimal correction term $r^{op}$ with different values of $\varepsilon$.
First, as shown in Table~\ref{tbl: threshold-correction-terms}, the optimal term is usually larger than the mean of the injected noise, aligning with our analysis in Section~\ref{subsec: correction_methodology}.
Second, SVT-Exp without correction demonstrates relatively low NCR even compared to, \eg, SVT-Lap, highlighting the need for a threshold correction method when using exponential noise in SVT.
Third, SVT-Exp with our optimal correction terms drastically outperforms SVT-Exp with mean correction, which slightly outperforms the other baselines. This demonstrates the effectiveness of both exponential noise and optimal threshold correction methods.

\subsubsection{Effectiveness of the Appending Strategy.}
To demonstrate the effectiveness of our appending strategy, we compare the performance of all our baselines at the same overall privacy budget consumption with varying numbers of traverse~(\ie, the number of times a query with the noisy negative outcome compared with the threshold)\footnote{For fairness, under each number of traverses, we ensure that each baseline consumes the same amount of the privacy budget before comparing their performance. Hence, the overall privacy budget consumption varies with a different number of traverses.}. As shown in Figure~\ref{fig: traverse}, our proposed method performs better across most datasets under all tested traverse numbers. Notably, our method shows slight inferiority on the T40I10D100K dataset under a very small number of traverses. One possible reason is the relatively small gap between positive query results and the predefined threshold on T40I10D100K~(Cf. Figure~\ref{fig: dataset-properties}): due to our relatively high optimal threshold correction term, more queries are needed to filter out the true positive queries, as explained in Section~\ref{subsec: correction_methodology}.
Even though, it is worth noting that our method outperforms the others on the T40I10D100K dataset with only a few more traverses~(\eg, less than $5$). That is also to say, our proposed method yields better performance compared to other baselines with only a marginal additional computation cost.

%% file: related_work.tex
\section{Related Work}

Our work focuses on enhancing one of the most fundamental DP algorithms initially introduced by Dwork et al.~\cite{dwork2009complexity,dwork2014algorithmic, roth2010interactive}, namely the \svtfull~(SVT).
Benefiting from its key feature where only positive outcomes consume privacy, SVT and its variants~\cite{kaplan2021sparse} have become crucial components in numerous algorithms across various domains and scenarios~\cite{lee2014top,chen2015differentially,stoddard2014differentially,shokri2015privacy,bun2017make,hasidim2020adversarially,zhang2021wide,kaplan2023differentially}. 
The applications of SVT span shared-parameter selection in deep learning models~\cite{shokri2015privacy}, time-stamp selection for streaming data~\cite{hasidim2020adversarially}, online query answering~\cite{bun2017make}, among others.
By improving the accuracy of SVT in a broad context, our work has the potential to enhance the performance of the aforementioned algorithms.

To gain deeper insights into the \svtfull, Lyu~\etal summarize all prevalent SVT variants in~\cite{lyu2017understanding}. Apart from the thorough privacy analysis under the classic notion of DP, they also propose an enhanced SVT with the optimal privacy budget allocation scheme. Shortly after, Zhu~\etal~\cite{zhu2020improving} revisit the privacy analysis of the SVT algorithm under a relaxed privacy notion, namely R\'enyi differential privacy~(RDP)~\cite{mironov2017renyi}. Their research explores the utility of Gaussian noise for both threshold and query perturbation within the SVT framework. They argue that Gaussian noise may outperform Laplace noise in specific scenarios, such as when query results predominantly fall below predefined thresholds or when the number of queries is limited and not excessively large.
In alignment with Lyu~\etal~\cite{lyu2017understanding}, our work conducts the privacy analysis through the lens of DP, aiming for a generic result. Moreover, our work focuses more on noise selection rather privacy allocation, which distinguishes us from Lyu~\etal. Also, we demonstrate
the advantage of leveraging
the exponential noise in comparison to other alternatives, which further distinguishes us from Zhu~\cite{zhu2020improving}.

Another line of research improves the SVT by exploiting information revealed in the algorithm for free. Concretely, Ding~\etal~\cite{ding2023free} pinpoint that releasing the gap between the noisy thresholds and the noisy query results incurs no extra privacy costs. Hence, they use the free gap as guidance to design a better privacy budget allocation scheme. Kaplan~\etal.~\cite{kaplan2021sparse} improve SVT by iteratively deleting elements that contribute to current positive outcomes, which enables a more refined privacy accountant.
Since both methods can be applied on top of our proposed method, we leave them out for comparison in this work.

Notice that we are not the pioneer effort in utilizing the exponential noise for achieving differential privacy in the literature. Previous studies from Durfee \etal\cite{durfee2019practical} and Shekelyan~\etal\cite{shekelyan2022differentially} demonstrate the use of exponential in achieving the exponential mechanism~(EM), another differentially private algorithm primarily employed for privacy-preserving top-k selection tasks. 
Concretely, they~\cite{durfee2019practical,shekelyan2022differentially} leverage the exponential noise to Oneshot~\cite{qiao2021oneshot} mechanism where items with the top-k highest score are selected by injecting noise into each score ranking all noisy scores in descending order~\cite{wasserman2010statistical,ding2021permute,mckenna2020permute}.
While the exponential mechanism (EM) offers superior accuracy guarantees compared to SVT for the top-k selection problem, its application is limited: EM struggles with queries on streaming data, where the total number of queries might be infinite. 
Therefore, our focus remains on SVT in this work, given its broader applicability as a more generic algorithm.
Ding \etal~\cite{ding2023free} briefly look into applying the exponential noise to SVT. However, their approach, which involves injecting exponential noise into the threshold, has been demonstrated to compromise privacy by our findings.
To the best of our knowledge, we are the first attempt that applies the exponential noise in SVT, accompanied by rigorous privacy and utility guarantee.

%% file: conclusion.tex
\section{Conclusion}
This work aims to enhance the query accuracy of the SVT algorithm by utilizing exponential noise. We begin with revisiting the privacy analysis of SVT algorithms and expanding the range of noise options for query perturbation by considering its less informative nature. Our analysis identifies exponential noise as the most effective, both theoretically and empirically, among the considered noise distributions. Additionally, we develop a generic optimal threshold correction method and an appending strategy to ensure both a high query precision and recall for SVT with marginal computation cost. The effectiveness of these methods is thoroughly validated through comprehensive experiments on both real-world and synthetic datasets.

%% file: appendix.tex
\appendix
\section{Proof of Theorem~\ref{theo: privacy-svt-lips}}\label{sec: proof-of-lips-privacy}
Before starting our proof, we first define two different query settings for the SVT algorithm: the monotonic query setting and the non-monotonic query setting~(or monotonic setting and non-monotonic setting in short).

For the monotonic setting, when the dataset $D$ changes into its neighboring dataset $D^{\prime}$, all the query results change in the same direction. That is to say, we have either $\forall_{i\in\left[n\right]}q_i(D)\geq q_i(D^{\prime})$ or $\forall_{i\in\left[n\right]}q_{i}(D)\leq q_{i}(D^{\prime})$, where $n$ is number of different queries in total.
By contrast, for the non-monotonic setting, the above statement is not necessary to be true: the query results can change in different directions when $D$ changes into $D^{\prime}$.

In the following, we first show our privacy analysis for the non-monotonic setting. We then elaborate on how to extend it to monotonic setting.

To begin with, we restate Theorem~\ref{theo: privacy-svt-lips} as follows:
\begin{theorem}[Restate of Theorem~\ref{theo: privacy-svt-lips}]\label{theo: restate-privacy-svt-lips}
    Algorithm~\ref{alg: basic-svt} satisfies differential privacy if for any real numbers $b_1$, $b_2$, there are two positive real numbers $k_1$ and $k_2$ such that the following inequalities hold:
\begin{equation}\label{eq: pure_dp}
    \vert \ln(f_1(x))-\ln(f_1(x+b_1))\vert \leq k_1\vert b_1\vert,
\end{equation}
\begin{equation}\label{eq: lispchitz}
    \vert \ln(1-F_2(x))-\ln(1-F_2(x+b_2))\vert \leq k_2\vert b_2\vert,
\end{equation}
where $f_1(\cdot)$ and $F_2(\cdot)$ are the probability density function of $\mathcal{N}_1$ and the cumulative function of the $\mathcal{N}_2$, respectively.
\end{theorem}

Now, we provide our proof under the non-monotonic settings.

\begin{proof}[Proof for the non-monotonic setting.]
Suppose there are a pair of neighboring datasets $D$ and $D^{\prime}$ that differ by a single record, a sequence of $n$ queries $Q=\{q_1,\ldots,q_n\}$ that have the largest sensitivity $\Delta$,
and corresponding predefined thresholds $T=\{T_1,\ldots,T_n\}$.
Let $\mathcal{M}$ denotes Algorithm~\ref{alg: basic-svt}, $o\in \mathcal{O}$ is a possible output, and $S_{\bot}$, $S_{\top}$ denote the set of negative and positive outcomes, respectively. We have that

\begin{equation}\label{eq: initial-privacy-analysis}
\begin{aligned}
&\Pr[\mathcal{M}(D)=o]\\
=&\int_{-\infty}^{+\infty}\Pr[\rho=t]\Pr[\max\limits_{i\in S_{\bot}} \tilde {q_i}(D)<\tilde{T_i}]\Pr[\min\limits_{j\in S_{\top}} \tilde {q_j}(D)\geq \tilde{T_j}] dt\\
=&\int_{-\infty}^{+\infty}\Pr[\rho=t]\underbrace{\prod\limits_{i\in S_{\bot}}\Pr[\tilde {q_i}(D)<\tilde{T_i}]\prod\limits_{j\in S_{\top}}\Pr[\tilde {q_j}(D)\geq \tilde{T_j}]}_{(*)} dt,
\end{aligned}
\end{equation}
where $\tilde{q_{i}}(D)=q_{i}(D)+v_i$ and $\tilde{T_{i}}=T_i+\rho$ for each $i\in[n]$ are the perturbed query result and the perturbed threshold, respectively. In particular, $t$ is an instance of the random varible $\rho$.

In the worst case, $0\leq q_{i}(D^{\prime})-q_{i}(D)\leq\Delta$ holds for $i\in S_{\bot}$ and $0\leq q_{j}(D)-q_{j}(D^{\prime})\leq\Delta$ holds for $j\in S_{\top}$. Therefore, it is further derived that
\begin{equation}\label{eq: worst-case-replace}
    (*)\leq \prod\limits_{i\in S_{\bot}}\Pr[\tilde {q_i}(D^{\prime})-\Delta<\tilde{T_i}]\prod\limits_{j\in S_{\top}}\Pr[\tilde {q_j}(D^{\prime})+\Delta\geq \tilde{T_j}].
\end{equation}
By first inserting Inequality~\ref{eq: worst-case-replace} into Equation~\ref{eq: initial-privacy-analysis} and then converting the integral variable $t$ into $u=t+\Delta$, it is obvious that Equation \ref{eq: initial-privacy-analysis} is no larger than
\begin{equation}\label{eq: after-change-integral}
    \begin{aligned}
        &\int_{-\infty}^{+\infty}\Pr[\rho=t^{\prime}]\underbrace{\prod\limits_{i\in S_{\bot}}\Pr[\tilde {q_i}(D^{\prime})<\tilde{T_i}]\prod\limits_{j\in S_{\top}}\Pr[\tilde {q_j}(D^{\prime})\geq \tilde{T_j}-2\Delta]}_{(**)} du,
    \end{aligned}
\end{equation}
where $t^{\prime}=u-\Delta$, $\tilde{T_i}=T_i+u$, and $\tilde{T_j}=T_j+u$. 

If the distribution of the additive noise for threshold perturbation meets the condition in Theorem~\ref{theo: privacy-svt-lips}, by setting $k_1$ in Equation~\ref{eq: pure_dp} to $\frac{\varepsilon_1}{\lvert b_1\rvert}$, we will therefore reach that
\begin{equation}\label{eq: privay-threshold}
    \Pr[\rho=t]\leq e^\varepsilon_1\cdot\Pr[\rho=t-\Delta].
\end{equation}

Additionally, notice that
\begin{equation}\label{eq: **}
    (**)=\prod\limits_{i\in S_{\bot}}F(\tilde{T_i}-q_i(D^{\prime}))\prod\limits_{j\in S_{\top}}(1-F(\tilde{T_j}-2\Delta-q_j(D^{\prime}))),
\end{equation}
where $F(\cdot)$ is the cumulative distribution function of the random noise $v_i$ for $i\in [n]$ that is used for query results perturbation. If $F(\cdot)$ meets the condition stated by Inequality~\ref{eq: lispchitz} in Theorem~\ref{theo: privacy-svt-lips} and $k$ is set to $\frac{\varepsilon_2}{\lvert b\rvert\lvert S_{\top}\rvert}$, we obtain that for each $j\in S_{\top}$,
\begin{equation}
    1-F(\tilde{T_j}-2\Delta-q_j(D^{\prime}))\leq e^{\frac{\varepsilon_2}{c}}\cdot(1-F(\tilde{T_j}-q_j(D^{\prime})),
\end{equation}
where $c$ is one of the inputs in Algorithm~\ref{alg: basic-svt} and $\lvert S_{\top}\rvert\leq c$ holds. By converting the integral variable $u$ back to $t$ and putting all the pieces together, we finally have that
\begin{equation}\label{eq: complete_privay_proof}
    \begin{aligned}
        (\ref{eq: initial-privacy-analysis})\leq e^{\varepsilon_1+\varepsilon_2}\cdot&
        \int_{-\infty}^{+\infty}\Pr[\rho=t]\prod\limits_{i\in S_{\bot}}F(\tilde{T_i}-q_i(D^{\prime}))\\
        &\prod\limits_{j\in S_{\top}}(1-F(\tilde{T_j}-q_j(D^{\prime}))) dt=\Pr[\mathcal{M}(D^{\prime})=o],
    \end{aligned}
\end{equation}
which then completes the proof.
\end{proof}

Now, we show how we can extend our proof to the monotonic setting.
\begin{proof}[Proof for the monotonic setting.]
For the monotonic settings, the only difference lies in Inequality~\ref{eq: worst-case-replace}. As the worst case for this settings is that $0\leq\tilde{q}_i(D^{\prime})-\tilde{q}_{i}(D)\leq\Delta$ holds for $i\in S_{\bot}$, and $q_{j}(D)-q_{j}(D^{\prime})=0$ holds for $j\in S_{\top}$, we have that:
\begin{equation}\label{eq: worst-case-replace-mono}
    (*)\leq \prod\limits_{i\in S_{\bot}}\Pr[\tilde {q_i}(D^{\prime})-\Delta<\tilde{T_i}]\prod\limits_{j\in S_{\top}}\Pr[\tilde {q_j}(D^{\prime})\geq \tilde{T_j}].
\end{equation} 
Then, similar to what we do under the non-monotonic setting, by first inserting Inequality~\ref{eq: worst-case-replace-mono} into Equation~\ref{eq: initial-privacy-analysis} and then converting the integral variable $t$ into $u=t+\Delta$, we have that Equation~\ref{eq: initial-privacy-analysis} is no larger than
\begin{equation}\label{eq: after-change-integral-mono}
    \begin{aligned}
        &\int_{-\infty}^{+\infty}\Pr[\rho=t^{\prime}]\underbrace{\prod\limits_{i\in S_{\bot}}\Pr[\tilde {q_i}(D^{\prime})<\tilde{T_i}]\prod\limits_{j\in S_{\top}}\Pr[\tilde {q_j}(D^{\prime})\geq \tilde{T_j}-\Delta]}_{(**)} du,
    \end{aligned}
\end{equation}
where $t^{\prime}=u-\Delta$, $\tilde{T_i}=T_i+u$, and $\tilde{T_j}=T_j+u$. 
Then, we have that 
\begin{equation}
        (**)=\prod\limits_{i\in S_{\bot}}F(\tilde{T_i}-q_i(D^{\prime}))\prod\limits_{j\in S_{\top}}(1-F(\tilde{T_j}-\Delta-q_j(D^{\prime}))),
\end{equation}
holds. If $F(\cdot)$ meets the condition stated by Inequality~\ref{eq: lispchitz} in Theorem~\ref{theo: privacy-svt-lips}, by setting $k$ to $\frac{\varepsilon_2}{\lvert b\rvert\lvert S_{\top}\rvert}$ and changing the integral variable $u$ back to $t$, we are able to complete our proof as in Inequality~\ref{eq: complete_privay_proof}.
\end{proof}

\section{Proof of Theorem~\ref{theo: privacy-guarantee}}\label{sec: exp-privacy-proof}
Same as in Appendix~\ref{sec: proof-of-lips-privacy}, we provide our proof of Theorem~\ref{theo: privacy-guarantee} first under the non-monotonic setting, then demonstrate how we can adapt Algorithm~\ref{alg: exp-svt} as well as the privacy proof to the monotonic setting.
\begin{theorem}[Restate of Theorem~\ref{theo: privacy-guarantee}]
    Let $c$ denote the number of positive outcomes output by Algorithm~\ref{alg: exp-svt}. Algorithm~\ref{alg: exp-svt} satisfies $(\varepsilon_1+\varepsilon_2)$-differential privacy when \texttt{RESAMPLE} is set to \texttt{False}, and $(c\varepsilon_1,\varepsilon_2)$-differential privacy when \texttt{RESAMPLE} is set to \texttt{True}, where $\varepsilon_1$ and $\varepsilon_2$ are the privacy budgets for threshold and query perturbation, respectively.
\end{theorem}
\begin{proof}
If \texttt{RESAMPLE} is set to \texttt{FALSE}, then based on the proof of Theorem~\ref{theo: privacy-svt-lips}, we only need to prove that cumulative distribution function~(CDF) of the exponential distribution satisfies Equation~\ref{eq: lispchitz}. Since the CDF of the exponential distribution with parameter $\lambda=\frac{\varepsilon_2}{2c\Delta}$ is
\begin{equation}
    \begin{aligned}
        F_{exp}=
        \begin{cases}
            1-\exp(-\lambda x) & x\geq 0,\\
            0 & otherwise,
        \end{cases}
    \end{aligned}
\end{equation}
we have that for each $j\in S_{\top}$,
\begin{equation}\label{eq: exp-query-privacy}
\begin{aligned}
        &\left | \ln(1-F_{exp}(\tilde{T_j}-2\Delta-q_j(D^{\prime})))-\ln(1-F_{exp}(\tilde{T_j}-q_j(D^{\prime})))\right |\\
        \leq& \left |\ln\frac{\exp{(-\lambda(\tilde{T_j}-2\Delta-q_j(D^{\prime})))}}{\exp{(-\lambda(\tilde{T_j}-q_j(D^{\prime})))}}\right |
        \leq  2\Delta\lambda=\frac{\varepsilon_2}{c},
\end{aligned}
\end{equation}
where $c$ is the number of positive outcomes in Algorithm~\ref{alg: exp-svt}.
In addition, as the Laplace noise is used for threshold perturbation, according to the Laplace mechanism described in Section~\ref{subsec: background_dp}, we have that
\begin{equation}\label{eq: lap-threshold-privacy}
    \Pr[\rho=t]\leq e^\varepsilon_1\cdot\Pr[\rho=t-\Delta].
\end{equation}
By plugging Equation~\ref{eq: exp-query-privacy} and Equation~\ref{eq: lap-threshold-privacy} into Equation~\ref{eq: **} in Section~\ref{sec: privacy-revisit}, we arrive at the following inequalities:
\begin{equation}
    \Pr[\mathcal{M}(D)=o]\leq e^{\varepsilon_1+\varepsilon_2}\cdot\Pr[\mathcal{M}(D^{\prime})=o],
\end{equation}
which completes the proof when $\texttt{RESAMPLE}=\texttt{False}$.

If \texttt{RESAMPLE} is set to \texttt{True}, we can derive the privacy guarantee of Algorithm~\ref{alg: exp-svt} by considering a subsequence of queries $\{q_1,\ldots,q_k,q_{k+1}\}$ first. In particular, suppose the first $k$ queries output negative outcomes while the $k+1$th query output a positive outcome. Then, we have that
\begin{equation}\label{eq: subprivacy-analysis}
\begin{aligned}
&\Pr[\mathcal{M}_1(D)=o]\\
=&\int_{-\infty}^{+\infty}\Pr[\rho=t]\Pr[\max\limits_{i=1}^{k} \tilde {q_i}(D)<\tilde{T_i}]\Pr[\tilde {q_{k+1}}(D)\geq \tilde{T} _{k+1}] dt\\
=&\int_{-\infty}^{+\infty}\Pr[\rho=t]\prod\limits_{i=1}^{k}\Pr[\tilde {q_i}(D)<\tilde{T_i}]\Pr[\tilde{q}_{k+1}(D)\geq \tilde{T}_{k+1}] dt,
\end{aligned}
\end{equation}
where $\mathcal{M}_1$ is the the Algorithm~\ref{alg: exp-svt} but with the sub-sequence mentioned above as inputs.
Then, follow the same methodology as stated in Section~\ref{sec: privacy-revisit}, Equation~\ref{eq: lap-threshold-privacy}, and Equation~\ref{eq: exp-query-privacy}, it can be proved that
\begin{equation}
    \begin{aligned}
        \Pr[\mathcal{M}_1(D)=o_1]\leq e^{\varepsilon_1+\frac{\varepsilon_2}{c}}\cdot\Pr[\mathcal{M}_1(D)^{\prime}=o_1].
    \end{aligned}
\end{equation}
Notice that the original queries can be considered as a combination of $c$ sub-sequences. By adopting the sequential composition theorem~(\ie, Theorem~\ref{theo: composition}), we show that
\begin{equation}
    \begin{aligned}
        &\Pr[\mathcal{M}(D)=o]=\Pr[(\mathcal{M}_1(D),\ldots,\mathcal{M}_c(D))=o]\\
        \leq &e^{c\cdot(\varepsilon_1+\frac{\varepsilon_2}{c})}\cdot\Pr[(\mathcal{M}_1(D)^{\prime},\ldots,\mathcal{M}_c(D)^{\prime})=o]\\
        =&e^{c\varepsilon_1+\varepsilon_2}\cdot\Pr[\mathcal{M}(D^{\prime})=o],
    \end{aligned}
\end{equation}
which completes our proof.
\end{proof}

 Now, we elaborate on how we can adapt Algorithm~\ref{alg: exp-svt} and its privacy proof to the monotonic setting.

 As we have discussed in Appendix~\ref{sec: proof-of-lips-privacy}, the only difference lies in Inequality~\ref{eq: exp-query-privacy}. For the monotonic setting, for each $j\in S_{\top}$, it holds that:
 \begin{equation}\label{eq: exp-query-privacy-mono}
\begin{aligned}
        &\left | \ln(1-F_{exp}(\tilde{T_j}-\Delta-q_j(D^{\prime})))-\ln(1-F_{exp}(\tilde{T_j}-q_j(D^{\prime})))\right |\\
        \leq& \left |\ln\frac{\exp{(-\lambda(\tilde{T_j}-\Delta-q_j(D^{\prime})))}}{\exp{(-\lambda(\tilde{T_j}-q_j(D^{\prime})))}}\right |
        \leq  \Delta\lambda.
\end{aligned}
\end{equation}
Therefore, by perturbing each query result with random noise drawn from $\texttt{Exp}(\frac{1}{\lambda})$, where $\lambda=\frac{\varepsilon_2}{c\Delta}$, it is sufficient to provide a differential privacy guarantee stated in Theorem \ref{theo: privacy-guarantee}.

\section{Proof of Theorem~\ref{theo: utility-guarantee}}\label{sec: exp-utility-proof}

In this section, we first restate Theorem~\ref{theo: utility-guarantee} in the following, then prove it in three steps. First, we formalize the $(\alpha,\beta)$-accuracy in terms of Algorithm~\ref{alg: exp-svt} in a simple case where only one positive query is allowed to output~(\ie, $c=1$), $\Delta=1$, and $\varepsilon_1=\varepsilon_2=\frac{\varepsilon}{2}$ as in \cite{dwork2014algorithmic}. Then, we provide the utility guarantee for the case where the threshold correction term $r$ is set to $0$. After that, we show that our optimal threshold correction term derived Equation~\ref{eq: optimal-term} yields a accuracy no smaller than $r=0$. Also note that our proof is under the non-monotonic query setting, and is trivial to extend it to the monotonic query setting.
\label{sec: utility-guarantee-proof}

\begin{theorem}[Restate of Theorem~\ref{theo: utility-guarantee}]
    Given any $k$ records such that $\lvert \{i<k:d_i\geq t-\alpha\}\rvert=0$~(\ie, the record above and closest to the threshold is the last one), the Algorithm \ref{alg: exp-svt} at least $(\frac{4(\ln k+\ln \frac{2}{\beta})}{\varepsilon},\beta)$-accurate.
\end{theorem}

\begin{proof}
As mentioned above, for ease of computation and comparison, without loss of generality, we assume that $\Delta=1$, $c=1$ and $\varepsilon_1=\varepsilon_2=\frac{\varepsilon}{2}$ in Algorithm~\ref{alg: exp-svt}. 

Based on the definition of $(\alpha,\beta)$-accuracy~(Cf. Definition~\ref{def: svt-acc}), we need to compute $\beta$ given any $\alpha$ based on the following equation:
\begin{equation}\label{eq: acc-svt}
\begin{aligned}
        \beta &= 1-\prod\limits_{i=1}^{k}\Pr\left[\tilde{q}_{i}(D)-\alpha<\tilde{T}_i+r\right]\cdot\Pr\left[\tilde{q}_{j}(D)+\alpha\geq\tilde{T}_j+r\right]\\
        &=1-p(r),
\end{aligned}
\end{equation}
where $r$ is the threshold correction parameter and $p(r)$ is defined in Equation~\ref{eq: optimal_threshold_correction}.

We first show that when $r=0$, it holds that
\begin{equation}\label{eq: acc-r0}
    \beta\leq2\exp{\left(\frac{\ln{k}-\alpha\varepsilon}{4}\right)}
\end{equation}

Then, to prove Equation~\ref{eq: acc-r0}, it is sufficient for us to prove that with at least probability $1-\beta$, the following inequality holds:
\begin{equation}\label{eq: utility_initial}
    \max\limits_{i\in [k]}Exp_i(\frac{4}{\varepsilon})+\left|Lap(\frac{2}{\varepsilon})\right|\leq \alpha,
\end{equation}
where $Exp_i(\frac{1}{\lambda})$ is the random noise drawn from the exponential distribution with parameter $\lambda$ and $Lap(b)$ is the random noise drawn from the Laplace distribution with parameter $b$.

If the above equation holds, for any $i$ that $a_i=\top$, we have that
\begin{equation}
    \begin{aligned}
        q_{i}(D)+Exp_{i}(\frac{4}{\varepsilon})\geq \tilde{T}_i\geq T_i-\left|Lap(\frac{2}{\varepsilon})\right|.
    \end{aligned}
\end{equation}
Or in other words:
\begin{equation}
    \begin{aligned}
        q_{i}(D)\geq T_i-\left|Lap(\frac{2}{\varepsilon})\right|-Exp_{i}(\frac{4}{\varepsilon})\geq T_i -\alpha.
    \end{aligned}
\end{equation}

Similarly, for any $a_i=\bot$, we have that
\begin{equation}
    \begin{aligned}
        q_{i}(D)\leq T_i+\left|Lap(\frac{2}{\varepsilon})\right|+Exp_{i}(\frac{4}{\varepsilon})\leq T_i +\alpha.
    \end{aligned}
\end{equation}

An error may occur when either $Exp_i(\frac{4}{\varepsilon})$ or $\left|Lap(\frac{2}{\varepsilon})\right |$ is too large. Since both errors could happen, $\frac{\alpha}{2}$ is allocate to each type~\cite{dwork2014algorithmic}. Therefore, by letting the following inequalities hold:
\begin{equation}
    \begin{aligned}
        &\Pr[\min\limits_{i\in[k]]}Exp_{i}(\frac{4}{\varepsilon})\geq\frac{\alpha}{2}]\leq k\cdot\exp(-\frac{\alpha}{2})\leq \frac{\beta}{2},\\
        &\Pr[\left|Lap\right|\geq\frac{\alpha}{2}]\leq\frac{\beta}{2},
    \end{aligned}
\end{equation}
together with some calculus, we have that Inequality~\ref{eq: acc-r0} holds.

Now, we show that by taking the optimal threshold correction term $r^{op}$ into consideration, there is a $\beta^{\prime}$, such that
\begin{equation}\label{eq: alpha-reverse}
    \beta^{\prime}\leq\beta\leq2\exp{\left(\frac{\ln{k}-\alpha\varepsilon}{4}\right)}
\end{equation}
holds.

Note that in Equation~\ref{eq: optimal-term}, the optimal threshold correction term $r^{op}$ is obtained by maximizing $p(r)$. That is to say,
\begin{equation}
    r^{op}=\arg\min\left(1-\beta\right).
\end{equation}
Hence, by inserting $r^{op}$ back into Equation~\ref{eq: acc-svt}, we have that
\begin{equation}
    \begin{aligned}
        \beta^{\prime}=1-p\left(r^{op}\right)\leq1-p(r)=\beta
    \end{aligned}
\end{equation}
holds. Our proof of Theorem~\ref{theo: utility-guarantee} is then completed by computing 
$\alpha$ based on Inequality~\ref{eq: alpha-reverse}.
\end{proof}

\section{Analytical Expression of $p(r)$}\label{sec: pr_analytical}
By inserting $r+\alpha$ and $r-\alpha$ into Equation 10 in \cite{Craig}, we obtain the following equations immediately:
\begin{equation}\label{eq: success-rate-analytical}
\begin{aligned}
    &p(r) = \left(\Gamma\left(r+\alpha\right)\right)^{k}\cdot\left(1-\Gamma\left(r-\alpha\right)\right)\\
    = & \begin{cases}
        \left(\frac{b\exp{\left(\frac{r+\alpha}{b}\right)}}{2\left(1/\lambda+b\right)}\right)^{k}\left(1-\frac{b\exp{\left(\frac{r-\alpha}{b}\right)}}{2\left(1/\lambda+b\right)}\right) & r\leq-\alpha,\\

        \left(1+\frac{(1/\lambda)^2\exp{\left(-\frac{r+\alpha}{1/\lambda}\right)}}{b^2-(1/\lambda)^2}-\frac{b\exp{\left(-\frac{r+\alpha}{b}\right)}}{2(b-1/\lambda)}\right)^{k}\left(1-\frac{b\exp{\left(\frac{r-\alpha}{b}\right)}}{2\left(1/\lambda+b\right)}\right) & -\alpha\leq r\leq\alpha,\\

        \left(1+\frac{(1/\lambda)^2\exp{\left(-\frac{r+\alpha}{1/\lambda}\right)}}{b^2-(1/\lambda)^2}-\frac{b\exp{\left(-\frac{r+\alpha}{b}\right)}}{2(b-1/\lambda)}\right)^{k}
        \left(-\frac{(1/\lambda)^2\exp{\left(-\frac{r-\alpha}{1/\lambda}\right)}}{b^2-(1/\lambda)^2}+\frac{b\exp{\left(-\frac{r-\alpha}{b}\right)}}{2(b-1/\lambda)}\right) & r>\alpha,   
    \end{cases}
\end{aligned}
\end{equation}
where $b=\frac{\Delta}{\varepsilon_1}$ and $\lambda=\frac{\varepsilon_2}{2c\Delta}$ for the non-monotonic query setting, $b=\frac{\Delta}{\varepsilon_1}$ and $\lambda=\frac{\varepsilon_2}{c\Delta}$ for the monotonic query setting respectively.
\begin{figure}[htbp]
    \centering
    \includegraphics[width=0.6\linewidth]{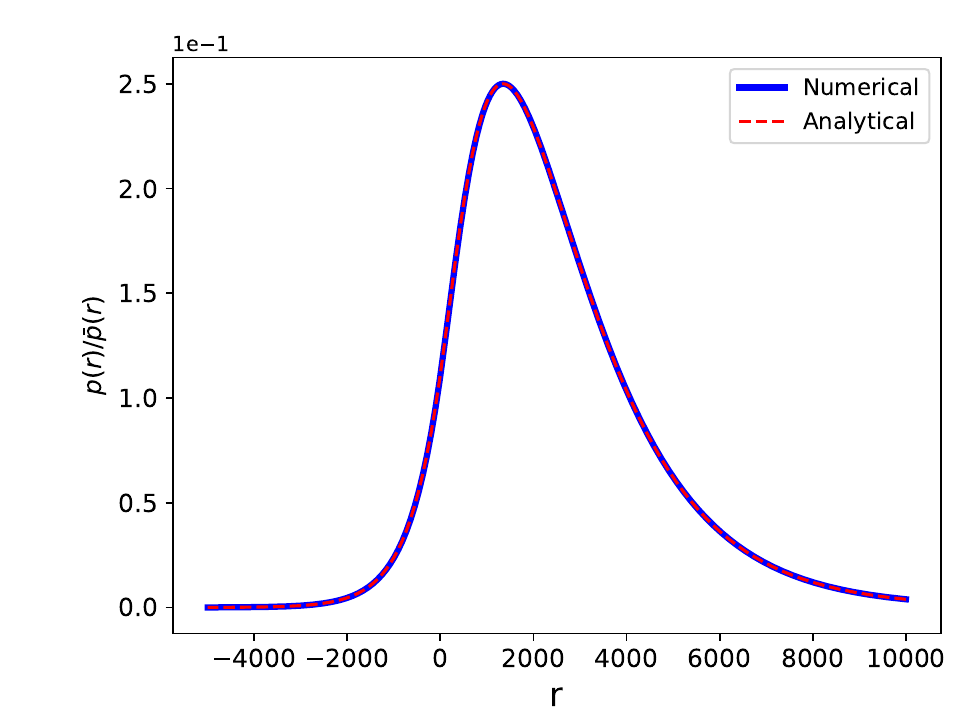}
    \caption{Comparison of $\bar{p}(r)$ computed using Algorithm~\ref{alg: threshold_correction} and Algorithm~\ref{alg: discretizer}. The bucket size $m$ Algorithm~\ref{alg: discretizer} is set to $20,001$. The other involved parameters are set as $\varepsilon=0.03$, $\alpha=0$, $k=10$.}
    \label{fig: threshold_analytical}
\end{figure}

We compare the analytical expression of $p(r)$ presented in Equation~\ref{eq: success-rate-analytical} with our numerical estimation computed using Algorithm~\ref{alg: threshold_correction} and Algorithm~\ref{alg: discretizer} in Figure~\ref{fig: threshold_analytical} to demonstrate the correctness of our numerical framework. The comparison shows that the numerical result align with the analytical expression, validating the correctness of both numerical framework and analytical result.

\section{Proof of the Maximum Value of $p(r)$}\label{sec: proof-of-maximum-pr}
To justify the design of our optimal threshold correction method, we prove in Theorem~\ref{theo: maximum-exists} that for any pair of noise distributions, the maximum value of $p(r)$ exists. Consequently, we can always obtain the optimal threshold correction term $r^{op}$ by maximizing $p(r)$.
\begin{theorem}\label{theo: maximum-exists}
    The maximum value of $p(r)$ exists and is no smaller than $\frac{k^{k}}{(k+1)^{k+1}}$.
\end{theorem}
\begin{proof}
    According to Equation~\ref{eq: success-rate}, we have:
    \begin{equation}
        \begin{aligned}
            p(r) = (\Gamma(r+\alpha))^{k}\cdot(1-\Gamma(r-\alpha)),
        \end{aligned}
    \end{equation}
    the differentiation of which is written as:
    \begin{equation}
        \begin{aligned}
            p^{\prime}(r) = (\Gamma(r+\alpha))^{k-1}\cdot(kf(r+\alpha)(1-\Gamma(r-\alpha))-\Gamma(r+\alpha)f(r-\alpha)),
        \end{aligned}
    \end{equation}
    where $f(\cdot)$ is the probability density function of the $\Gamma(\cdot)$. As $\Gamma(r+\alpha)^{k-1}>0$ holds, when $r\to -\infty$, both $\Gamma(r-\alpha)$ and $\Gamma(r+\alpha)$ grow closer to $0$. In other words:
    \begin{equation}
        \begin{aligned}
            \lim_{r\to -\infty} p^{\prime}(r) = f(r+\alpha)\geq 0.
        \end{aligned}
    \end{equation}
    Similarly, when $r\to +\infty$, we have $\Gamma(r+\alpha)\to 1$ and $\Gamma(r-\alpha)\to1$. Therefore, we have:
    \begin{equation}
        \begin{aligned}
            \lim_{r\to+\infty}p^{\prime}(r)=-f(r-\alpha)\leq0.
        \end{aligned}
    \end{equation}
    Since $p^{\prime}(r)$ is a continuous function for $r\in\mathbb{R}$, there must is a $r$ that has $p^{\prime}(r)=0$. Also, we have the following inequalities hold for any $r\in \mathbb{R}$ and $\alpha\geq 0$:
    \begin{equation}\label{eq: maximum-value}
        \begin{aligned}
             p(r) = &(\Gamma(r+\alpha))^{k}\cdot(1-\Gamma(r-\alpha))\\
             \geq & \underbrace{\Gamma(r)^{k}\cdot(1-\Gamma(r))}_{(***)}.
        \end{aligned}
    \end{equation}
    Since the maximum value of $(***)$ is $\frac{k^k}{(k+1)^{k+1}}$, which can be achieved when $\Gamma(r)=\frac{k}{k+1}$, we have that the maximum value of $p(r)$ is no less than $\frac{k^k}{(k+1)^{k+1}}$ according to Equation~\ref{eq: maximum-value}, which completes our proof.
\end{proof}

\section{Privacy Analysis of SVT-Gum}\label{sec: privacy-svt-gum}
To start with, we present the pseudo-code of the SVT algorithm with Gumbel noise in the Algorithm~\ref{alg: gum-svt}.
\begin{algorithm}
    \caption{SVT with Gumbel noise.}\label{alg: gum-svt}
    \LinesNumbered
    \KwIn{$Q=\{q_1,q_2,\ldots\}$,$\Delta$,$\varepsilon_1$,$\varepsilon_2$,$\lambda$,$c$,$T=\{T_1,T_2,\ldots\}$,$\gamma$, $\varepsilon$, $m$, option \texttt{RESAMPLE}}
    $\rho\sim Lap(0,\frac{\Delta}{\varepsilon_1})$\;
    $n_c=0$;$n_a$=0\;
       $r=\texttt{CorrectionTermOptimization}(b,\lambda,\gamma)$\;
    \For{$i=1,2,3,\ldots$}{
        $n_a=n_a+1$\;
        \underline{$\beta=\frac{2c\Delta}{\varepsilon_2}$;\ 
        $v_i\sim Gum(0,\beta)$}\;
    \eIf{$q_i(D)+v_i\geq \rho+T_i+r$}
        {Output $a_i=\top$\;$n_c=n_c+1$\;
        $S_{\top}=S_{\top}\cup i$\;
        if \texttt{RESAMPLE}, $\rho\sim Lap(0,\frac{\Delta}{\varepsilon_1})$\;
        \textbf{Abort} \text{ if $n_c\geq c$}\;}{Output $a_i=\bot$}}
\end{algorithm}

Notice that the only difference between the Algorithm~\ref{alg: exp-svt} and the Algorithm~\ref{alg: gum-svt} is that we use the random noise drawn from the Gumbel distribution with parameter $\beta=\frac{2c\Delta}{\varepsilon_2}$. In particular, the cumulative function of this distribution is:
\begin{equation}
    \begin{aligned}
        F_{Gum}(x)=\exp(-\exp(-\frac{x}{\beta})).
    \end{aligned}
\end{equation}
The privacy of Algorithm is given in the following theorem:
\begin{theorem}[Privacy guarantee]\label{theo: privacy-guarantee-gum}
    Algorithm~\ref{alg: gum-svt} satisfies $(\varepsilon_1+\varepsilon_2)$-differential privacy when \texttt{RESAMPLE} is set to \texttt{False}, and $(c\varepsilon_1,\varepsilon_2)$-differential privacy when \texttt{RESAMPLE} is set to \texttt{True}.
\end{theorem}
\begin{proof}
    To prove Theorem~\ref{theo: privacy-guarantee-gum}, we only need to show that $F_{Gum}(x)$ satisfies the condition in Theorem~\ref{theo: privacy-svt-lips}. In other words, we have to prove that
    \begin{equation}\label{eq: privacy-gum-inital}
        \begin{aligned}
            \left| \frac{1-\exp(-\exp(-\frac{x}{\beta}))}{1-\exp(-\exp(-\frac{x+2\Delta}{\beta}))}\right|\leq e^{\frac{\varepsilon_2}{c}}.
        \end{aligned}
    \end{equation}
    For simplicity, we remove the symbol of absolute value here as the $\left|\frac{1-F(x+c)}{1-F(x)}\right|\geq e^{-\frac{\varepsilon_2}{c}}$ follows by symmetry.

    For clarity, we let $a$ denote $-\exp(-\frac{x}{\beta})$ and $b$ denote $\exp(-\varepsilon)$. Therefore, to prove that Equation~\ref{eq: privacy-gum-inital} holds for any $\varepsilon_2\geq 0$, we only need to show that the following equation holds:
    \begin{equation}
        \begin{aligned}
            1-\exp(ab)-b(1-\exp(a))\geq 0.
        \end{aligned}
    \end{equation}
    In particular, we let:
    \begin{equation}
        \begin{aligned}
            f(b) = 1-\exp(ab)-b(1-\exp(a)).
        \end{aligned}
    \end{equation}
    Then, we have
    \begin{equation}
        \begin{aligned}
            f^{\prime}(b) = 1-a\exp(ab)-(1-\exp(a))\geq 0
        \end{aligned}
    \end{equation}
    holds for any $x\in \mathbb{R}$ and any $\varepsilon_2\geq0$. Therefore, we have that when $b=0$, $f(b)$ takes the minimum value, which is $0$. The Equation~\ref{eq: privacy-gum-inital} is then proved. The rest of the proof of Theorem~\ref{theo: privacy-guarantee-gum} follows the same methodology as the proof of Theorem~\ref{theo: privacy-guarantee}. We hence omit it from here for simplicity.
    Also note that our proof here is for the non-monotonic query setting. The privacy analysis for the monotonic query setting can be developed follow the same methodology as in Appendix~\ref{sec: exp-privacy-proof}, we also omit it from here.
\end{proof}

\section{Proof of Privacy Allocation}\label{sec: privacy-allocation-proof}
Overall, the optimal privacy budget allocation scheme follows the methodology proposed by Lyu \etal~\cite{lyu2017understanding}: to make the following comparison as accurate as possible, we need to minimize the variance $V=\texttt{Var}(N_1)+\texttt{Var}(N_2)$.
\begin{equation}
    q(D)+N_2(c,\varepsilon_2)\geq T+N_1(c,\varepsilon_1).
\end{equation}
In this section, we discuss the optimal privacy budget allocation results with different noise distributions for the non-monotonic query setting. We omit the proof for the monotonic query setting, as it shares the same methodology and similar calculations as its non-monotonic counterparts. However, we provide the optimal privacy allocation results for monotonic query settings in Table~\ref{tbl: privacy-alloation-mono} to offer additional insights.
\begin{table}[htbp]
\caption{Optimal Privacy Allocation for the Monotonic Query Setting. While $\varepsilon$ is the overall privacy budget consumption, $\varepsilon_1$ and $\varepsilon_2$ are the privacy budget for threshold perturbation and query perturbation, respectively. $w$ is an indicator for privacy budget allocation.}\label{tbl: privacy-alloation-mono}
    \centering
    \scalebox{1}{
    \begin{tabular}{ccccc}
    \toprule
        ~ & SVT-Exp & SVT-Gum & SVT-Lap & SVT-Gau\\ \midrule
        $\varepsilon$ & \multicolumn{4}{c}{$\varepsilon = \varepsilon_1+\varepsilon_2,\ \varepsilon_2=w\varepsilon_1$}\\ \midrule
        $w$ & $\left(\frac{c}{\sqrt{2}}\right)^{2/3}$ & $\left(\frac{\pi c}{2\sqrt{3}}\right)^{2/3}$ & $c^{2/3}$ &$c^{2/3}$ \\ \bottomrule
    \end{tabular}}
\end{table}
\subsection{SVT-Exp}
In SVT-Exp, $N_1(c,\varepsilon_1)=Lap(\frac{\Delta}{\varepsilon_1})$ and $N_2(c,\varepsilon_2)=Exp(\frac{2c\Delta}{\varepsilon_2})$. Thus, we have that
\begin{equation}
    V=2\left(\frac{\Delta}{\varepsilon_1}\right)^2+\left(\frac{2c\Delta}{\varepsilon_2}\right)^2.
\end{equation}
Let $\varepsilon=\varepsilon_1+\varepsilon_2$ and $\varepsilon_2=k\varepsilon_1$, $V$ can be written as
\begin{equation}
    V=2\left(\frac{(w+1)\Delta}{\varepsilon}\right)^2+\left(\frac{2c(w+1)\Delta}{w\varepsilon}\right)^2.
\end{equation}
After some calculation, it is derived that $V$ is minimized when $w$ is set to $(\sqrt{2}c)^{2/3}$.

\subsection{SVT-Gum}
In SVT-Gum, $N_1(c,\varepsilon_1)=Lap(\frac{\Delta}{\varepsilon_1})$ and $N_2(c,\varepsilon_2)=Gum(\frac{2c\Delta}{\varepsilon_2})$. Thus, we have that
\begin{equation}
    V=2\left(\frac{\Delta}{\varepsilon_1}\right)^2+\frac{\pi^2}{6}\left(\frac{2c\Delta}{\varepsilon_2}\right)^2.
\end{equation}
Let $\varepsilon=\varepsilon_1+\varepsilon_2$ and $\varepsilon_2=w\varepsilon_1$, $V$ can be written as
\begin{equation}
    V=2\left(\frac{(w+1)\Delta}{\varepsilon}\right)^2+\frac{\pi^2}{6}\left(\frac{2c(w+1)\Delta}{w\varepsilon}\right)^2.
\end{equation}
After some calculation, it is derived that $V$ is minimized when $w$ is set to $(\frac{\pi c}{\sqrt{3}})^{2/3}$.

\subsection{SVT-Lap}
In SVT-Lap, $N_1(c,\varepsilon_1)=Lap(\frac{\Delta}{\varepsilon_1})$ and $N_2(c,\varepsilon_2)=Lap(\frac{2c\Delta}{\varepsilon_2})$. Thus, we have that
\begin{equation}
    V=2\left(\frac{\Delta}{\varepsilon_1}\right)^2+2\left(\frac{2c\Delta}{\varepsilon_2}\right)^2.
\end{equation}
Let $\varepsilon=\varepsilon_1+\varepsilon_2$ and $\varepsilon_2=w\varepsilon_1$, $V$ can be written as
\begin{equation}
    V=2\left(\frac{(w+1)\Delta}{\varepsilon}\right)^2+2\left(\frac{2c(w+1)\Delta}{w\varepsilon}\right)^2.
\end{equation}
After some calculation, it is derived that $V$ is minimized when $w$ is set to $(2c)^{2/3}$.

\subsection{SVT-Gau}
In SVT-Gau, we assume that $\sigma=\frac{\alpha\Delta}{\varepsilon}$~\cite{dwork2014algorithmic}. Therefore, we have that the variance of $N_1(c,\varepsilon_1)$ is $\sigma_1^2=\left(\frac{\alpha\Delta}{\varepsilon_1}\right)^2$ and the variance of $N_2(c,\varepsilon_2)$ is $\sigma_2^2=\left(\frac{2\alpha c\Delta}{\varepsilon_2}\right)^2$. In other words, we have that
\begin{equation}
\begin{aligned}
    V=\sigma_1^2+\sigma_2^2=\left(\frac{\alpha\Delta}{\varepsilon_1}\right)^2+\left(\frac{2\alpha c\Delta}{\varepsilon_2}\right)^2.
\end{aligned}
\end{equation}
Let $\varepsilon=\varepsilon_1+\varepsilon_2$ and $\varepsilon_2=w\varepsilon_1$, $V$ can be written as
\begin{equation}
    V=\left(\frac{(w+1)\alpha\Delta}{\varepsilon}\right)^2+\left(\frac{2c(w+1)\alpha\Delta}{w\varepsilon}\right)^2.
\end{equation}
After some calculation, it is derived that $V$ is minimized when $w$ is set to $(2c)^{2/3}$.

\section{Evaluation Results over F1-score}\label{sec: eva_f1}
In this section, we first present the overall evaluation results for the F1-score across six datasets in Figure~\ref{fig: overall-50-f1}, which exhibits similar trends to those observed for NCR in Figure~\ref{fig: overall-50}. We then analyze how the F1-score varies with the increasing number of traverses for queries with negative noisy outcomes in Figure~\ref{fig: traverse-f1}. This analysis also reflects similar trends and provides comparable insights to those observed in Figure~\ref{fig: traverse}.
\begin{figure*}[tbp!]
    \centering
    \subfigure[Binary dataset]{
    \includegraphics[width=7cm]{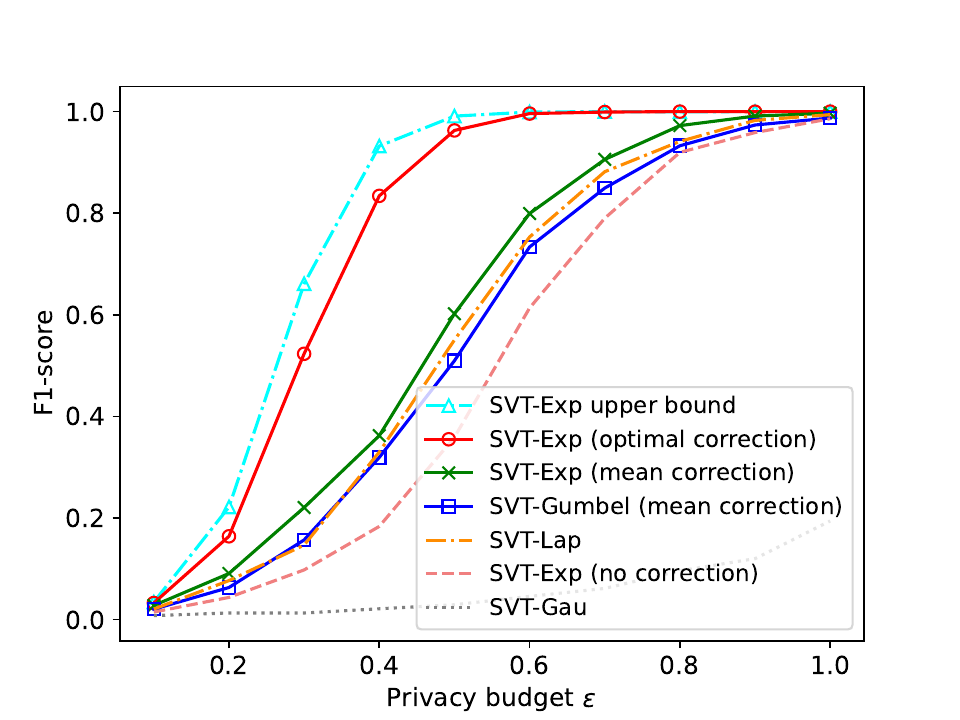}
    }
    \subfigure[Zipf dataset]{
    \includegraphics[width=7cm]{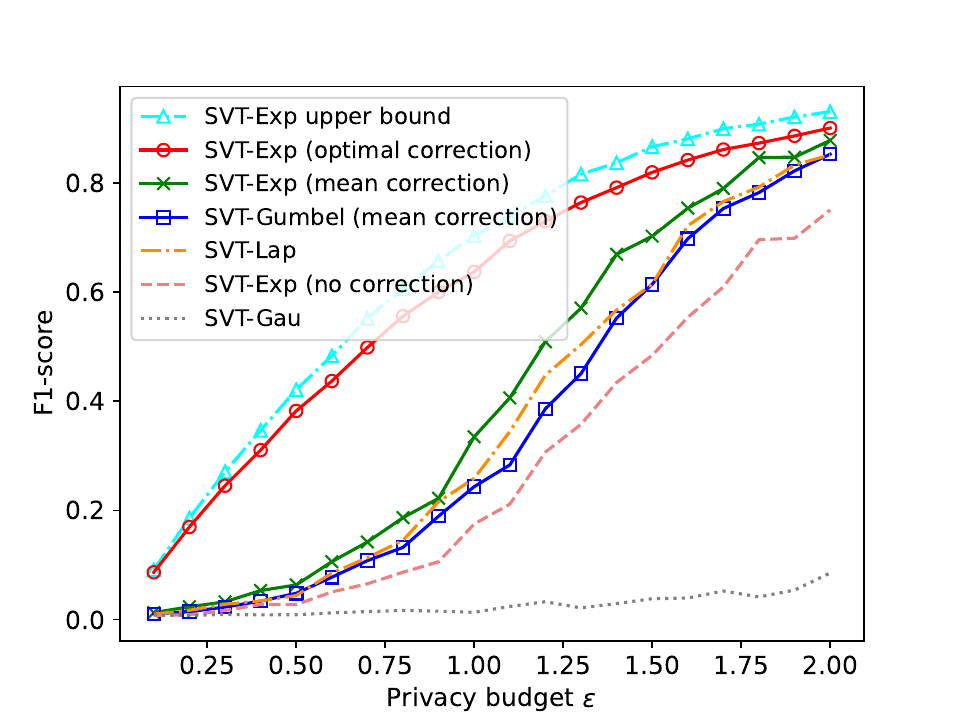}
    }
    \subfigure[BMS-POS dataset]{
    \includegraphics[width=7cm]{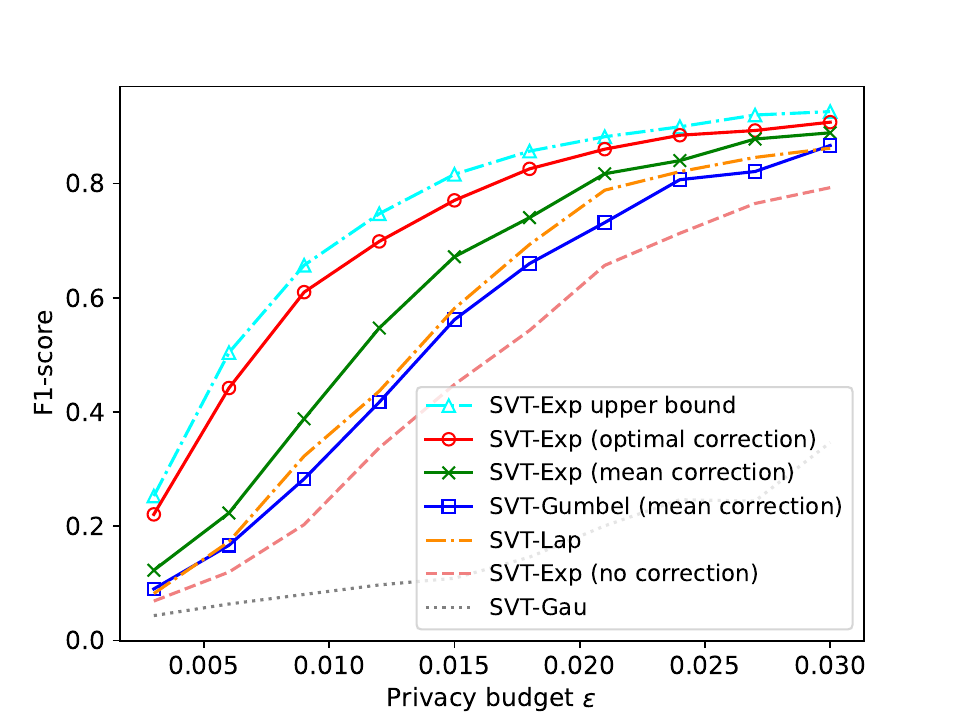}
    }
    \subfigure[Kosarak dataset]{
    \includegraphics[width=7cm]{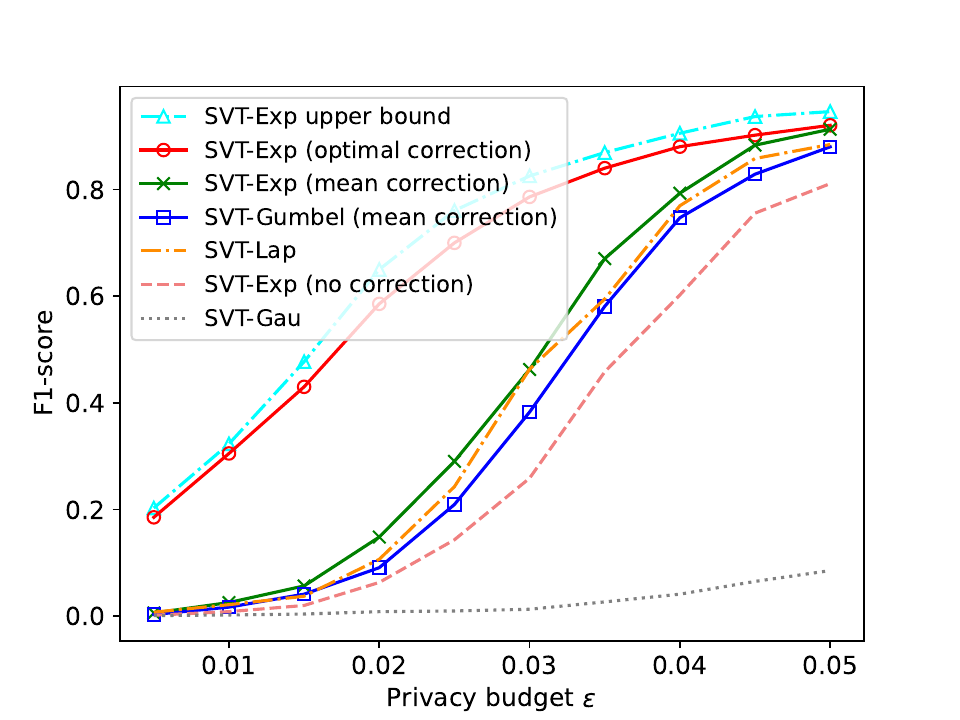}
    }
    \subfigure[Adult dataset]{
    \includegraphics[width=7cm]{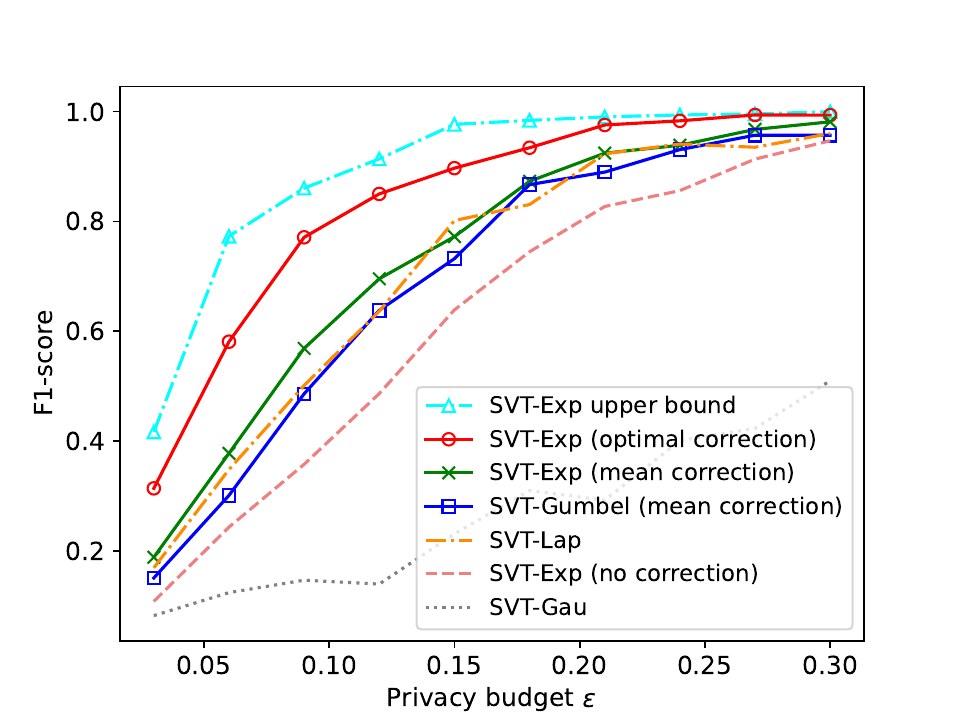}
    }
    \subfigure[T40I10D100K dataset]{
    \includegraphics[width=7cm]{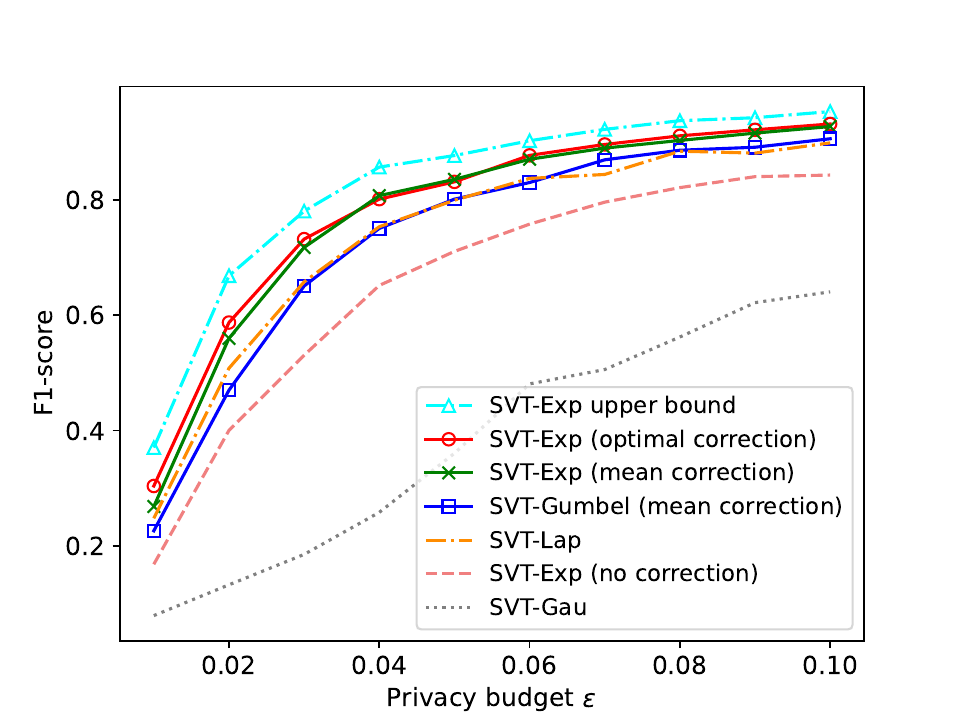}
    }
    \caption{F1-score on six datasets with $c=5$ for Adult and $c=50$ for the remaining datasets, $\alpha=0$, $k=\lfloor\frac{m}{c}\rfloor$, \texttt{RESAMPLE}=\texttt{False}, and \texttt{APPEND}=\texttt{True}. The sequential composition theorem is adopted for the overall privacy budget~(\ie, $\varepsilon$) computation.}
    \label{fig: overall-50-f1}
\end{figure*}

\section{Evaluation Results under RDP Composition}\label{sec: eva_rdp}
In this section, we provide the empirical evaluation results under RDP composition as mentioned in Section~\ref{subsec: background_dp}. We first specify the experiment settings, and then present the evaluation results in Figure~\ref{fig: overall-50-ncr-rdp}, Figure~\ref{fig: overall-50-f1-rdp}, Figure~\ref{fig: traverse-ncr-rdp}, and Figure~\ref{fig: traverse-f1-rdp}, respectively.

\textbf{Parameter settings.}
As stated by Zhu \etal~\cite{zhu2020improving}, using RDP composition in SVT results in a failure rate in Definition~\ref{def: DP} that is greater than $0$. In this section, we set the failure rate $\delta$ to $\frac{1}{n}$ for all compared baselines, following convention~\cite{dwork2014algorithmic}. For the overall privacy budget allocation, we adhere to the rules in Section~\ref{subsec: parameter_selection}. Specifically, for a fixed overall privacy budget $\varepsilon$, we use binary search to determine $\varepsilon_1$ and $\varepsilon_2$ such that $\varepsilon_1+\varepsilon_2=\varepsilon$, $\varepsilon_2=w\varepsilon_1$, and accumulated privacy budget for $c$ noisy positive outcomes equals to $\varepsilon_2$ after applying the RDP composition theorem when each query is assigned a privacy budget $\varepsilon_2^{\prime}$.
Additionally, we provide the evaluation results on only three datasets: Binary, BMS-POS, and Kosarak, as similar trends exhibit on other datasets. This selective presentation ensures clarity while maintaining the comprehensiveness of our analysis.

\textbf{Evaluation results.}
The results for NCR and F1-score of Algorithm~\ref{alg: exp-svt} under RDP composition are presented in Figure~\ref{fig: overall-50-ncr-rdp} and Figure~\ref{fig: overall-50-f1-rdp}, respectively. First, our proposed method consistently outperforms the other baselines across all datasets and privacy regions, similar to the trends observed in Figure~\ref{fig: overall-50} and Figure~\ref{fig: overall-50-f1}, where the sequential composition theorem is adopted. This further demonstrates the effectiveness of our proposed method as well as the generality of our method to other relevant techniques in the literature.
Second, the performance for both the F1-score and NCR of all compared baselines improves compared to their sequential composition counterparts. Notably, SVT-Gau shows a larger performance increase relative to the other baselines. This finding aligns with the analysis provided by Zhu~\etal~\cite{zhu2020improving}, which highlights that Gaussian noise achieves tighter composition results with RDP composition compared to other noise types. However, contrary to their results where SVT-Gau outperforms SVT-Lap, our results show that SVT-Lap still maintains a significant advantage. This discrepancy can be attributed to the different tested scenarios. As discussed by Zhu \etal, SVT-Gau is more suitable for scenarios where the gap between query results and predefined thresholds is relatively large, and relatively few queries are made. In our settings, however, many query results are quite close to the predefined threshold, and a large number of queries are conducted. This context reduces the relative effectiveness of SVT-Gau compared to SVT-Lap.

Figure~\ref{fig: traverse-ncr-rdp} and Figure~\ref{fig: traverse-f1-rdp} illustrate the trade-off between query accuracy and efficiency for our method under RDP composition. Consistent with the results shown in Figure~\ref{fig: traverse} and Figure~\ref{fig: traverse-f1}, our proposed method outperforms the other baselines across most datasets for all tested numbers of traversals. For the remaining datasets, our method demonstrates significant improvements in query accuracy and surpasses other baselines after a few additional rounds of traversals. This highlights the effectiveness of our appending strategy and confirms that our proposed method achieves superior query accuracy with only a marginal additional computational cost.
\begin{figure*}[htbp!]
    \centering
    \subfigure[Binary dataset]{
    \includegraphics[width=7cm]{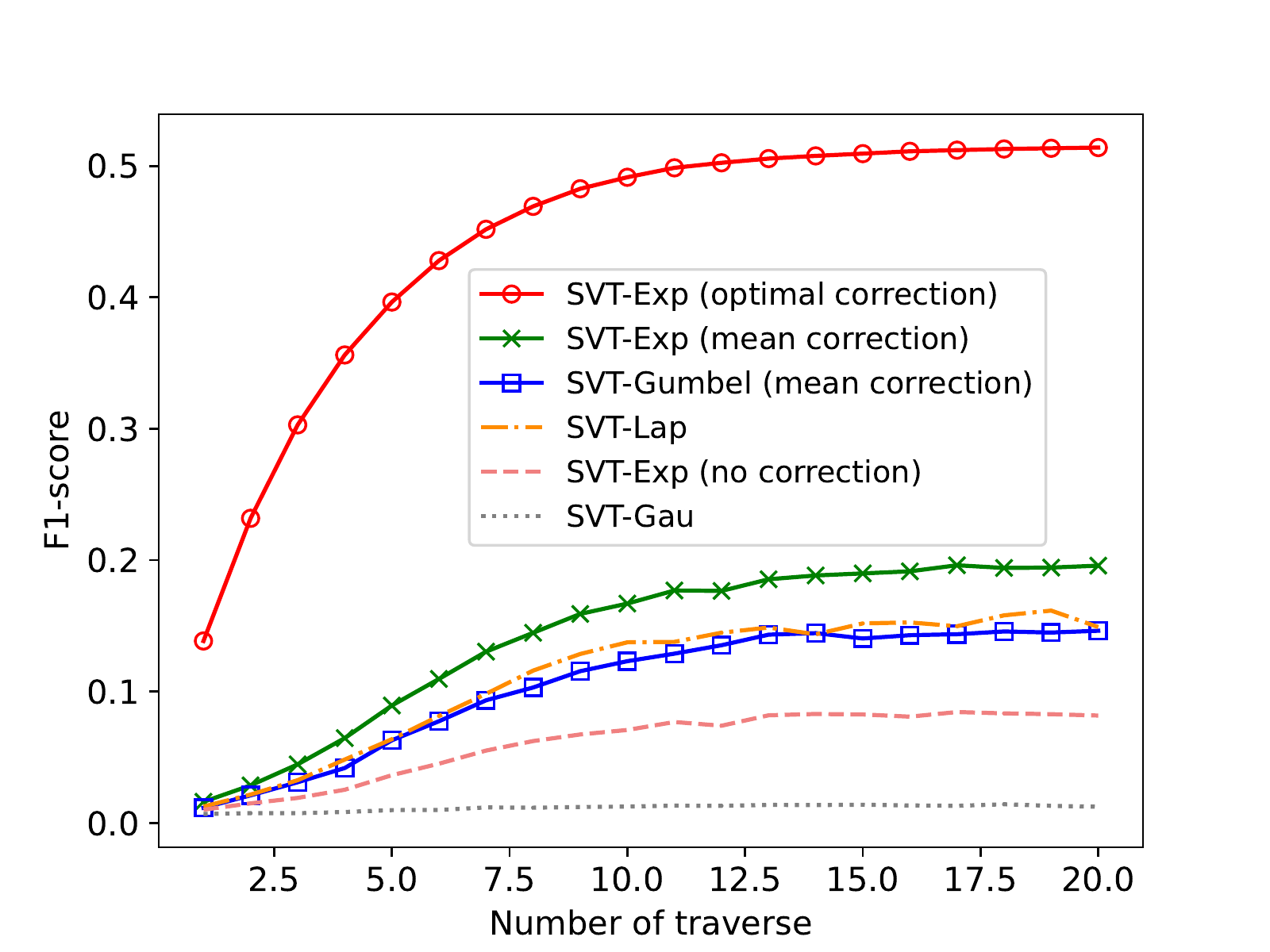}
    }
    \subfigure[Zipf dataset]{
    \includegraphics[width=7cm]{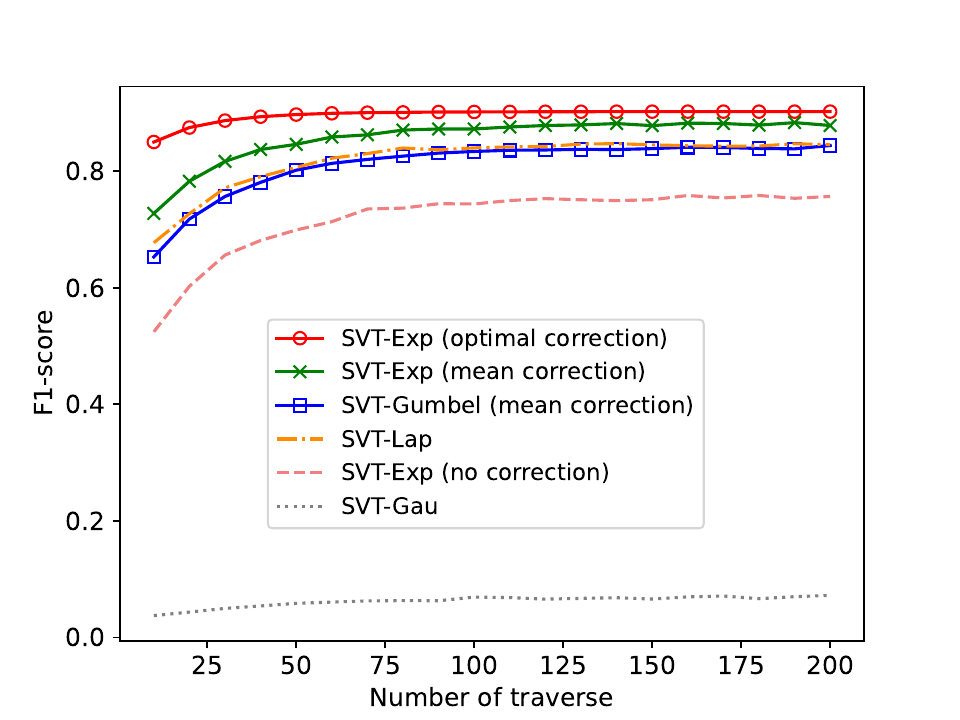}
    }
    \subfigure[BMS-POS dataset]{
    \includegraphics[width=7cm]{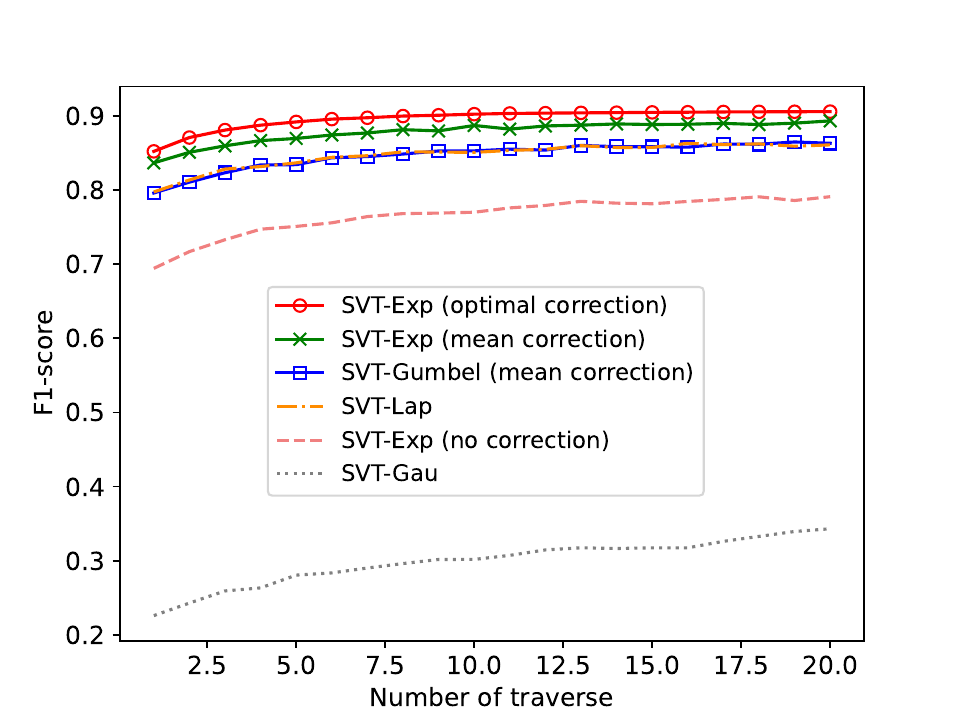}
    }
    \subfigure[Kosarak dataset]{
    \includegraphics[width=7cm]{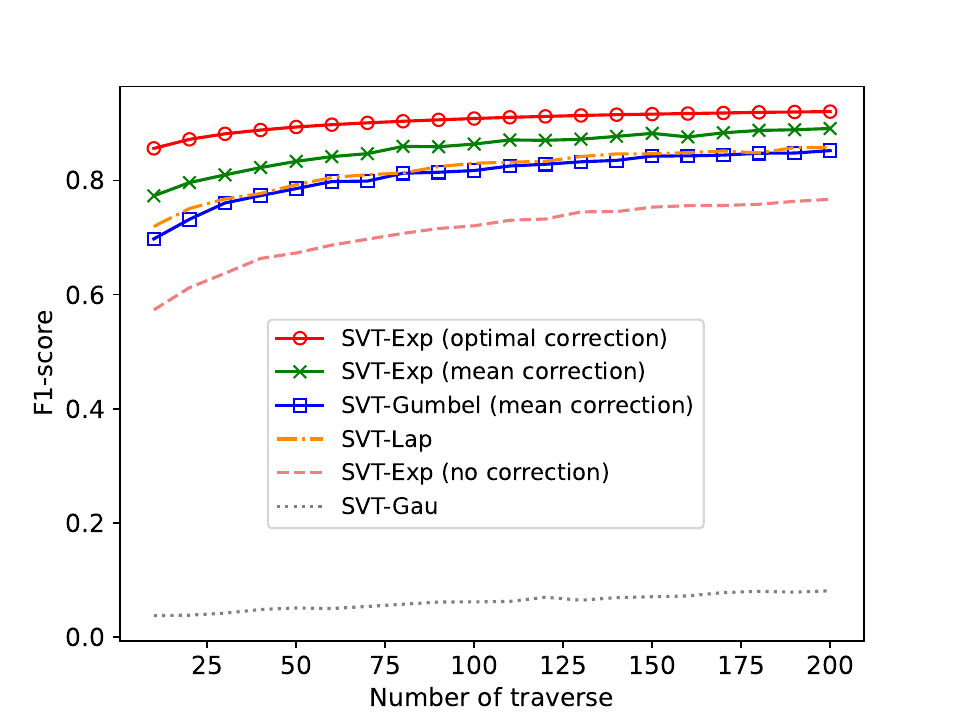}
    }
    \subfigure[Adult dataset]{
    \includegraphics[width=7cm]{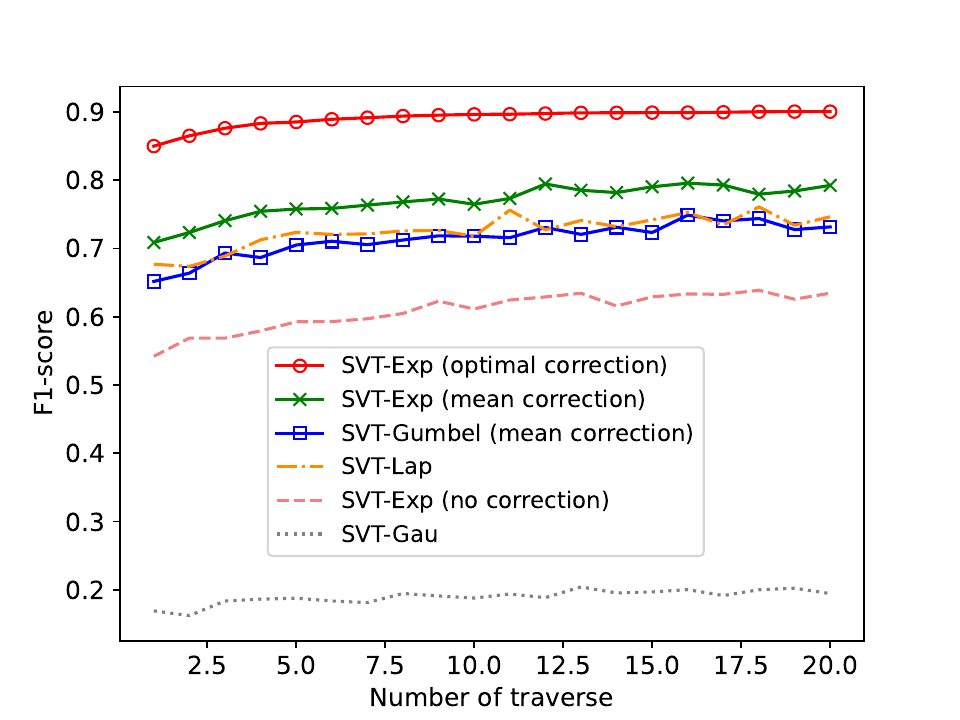}
    }
    \subfigure[T40I10D100K dataset]{
    \includegraphics[width=7cm]{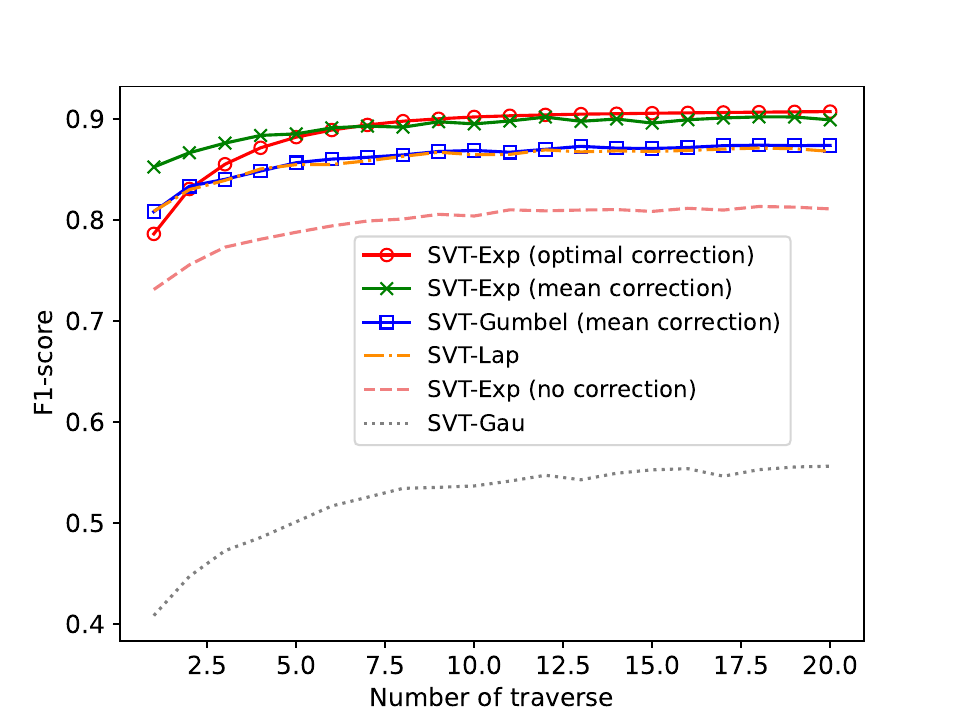}
    }
    \caption{F1-score on six datasets under varying number of traverses with $c=5$ for Adult and $c=50$ for the remaining datasets, $\alpha=0$, $k=\lfloor\frac{m}{c}\rfloor$. \texttt{RESAMPLE}=\texttt{False}, and \texttt{APPEND}=\texttt{True}.The sequential composition theorem is adopted for the overall privacy budget~(\ie, $\varepsilon$) computation.}
    \label{fig: traverse-f1}
\end{figure*}

\begin{figure*}[htbp!]
    \centering
    \subfigure[Binary dataset]{
    \includegraphics[width=5cm]{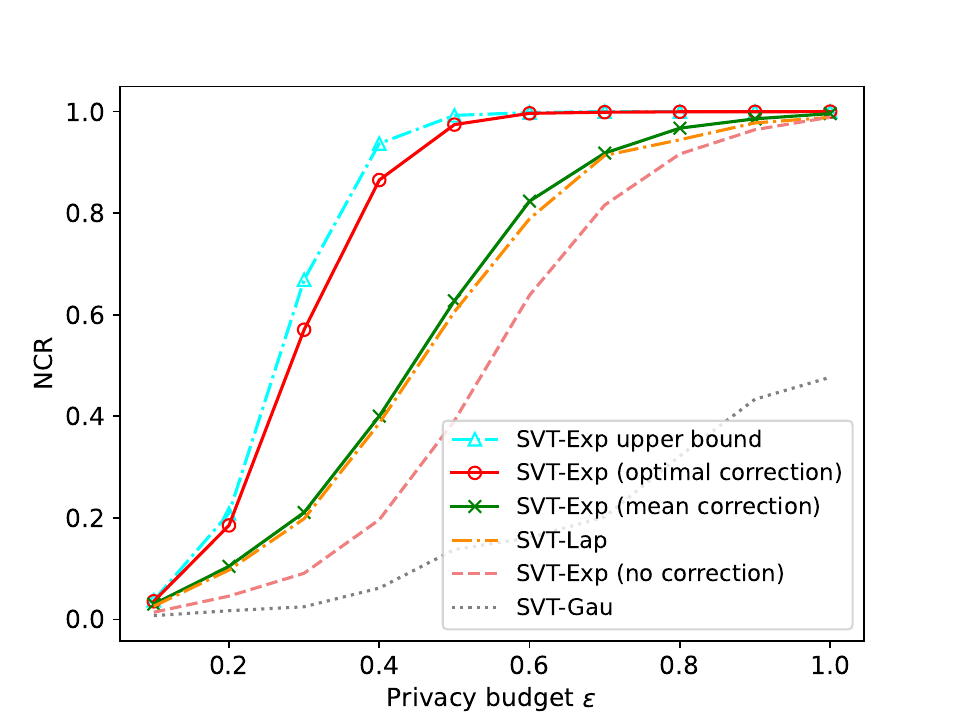}
    }
    \subfigure[BMS-POS dataset]{
    \includegraphics[width=5cm]{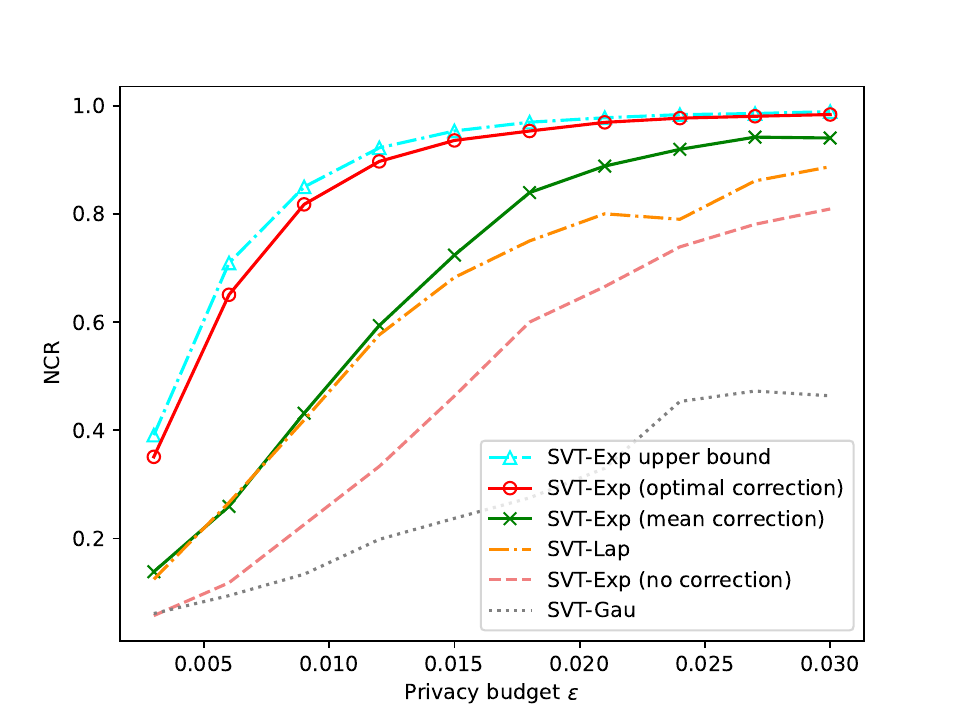}
    }
    \subfigure[Kosarak dataset]{
    \includegraphics[width=5cm]{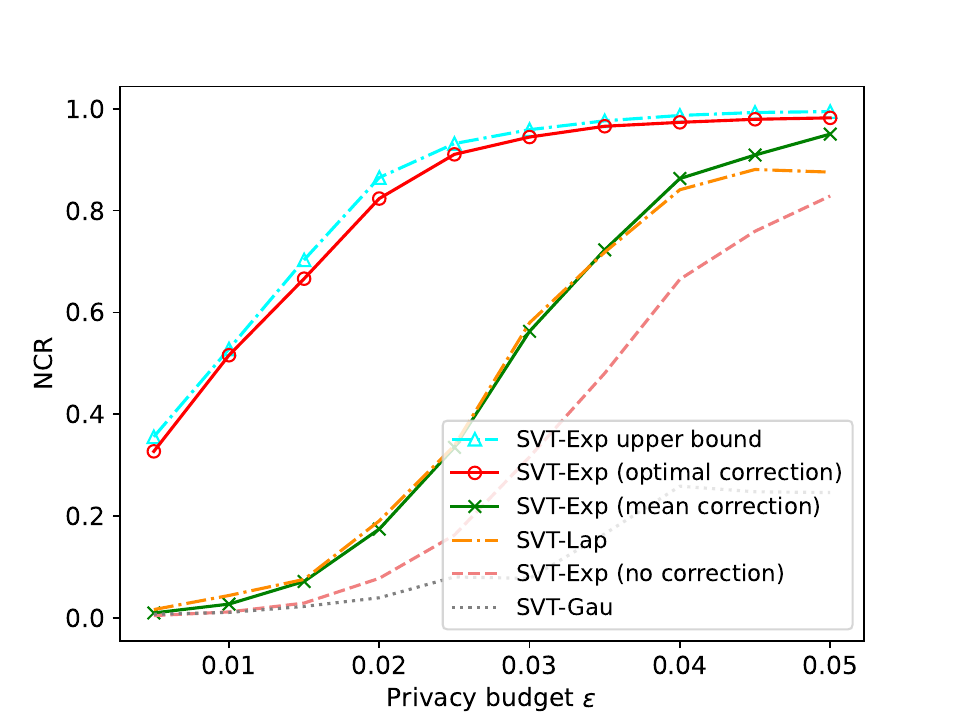}
    }
    \caption{NCR on three datasets with $c=50$, $\alpha=0$, $k=\lfloor\frac{m}{c}\rfloor$, \texttt{RESAMPLE}=\texttt{False}, and \texttt{APPEND}=\texttt{True}. The RDP composition theorem is adopted for the overall privacy budget~(\ie, $\varepsilon$) computation.}
    \label{fig: overall-50-ncr-rdp}
\end{figure*}

\begin{figure*}[htbp!]
    \centering
    \subfigure[Binary dataset]{
    \includegraphics[width=5cm]{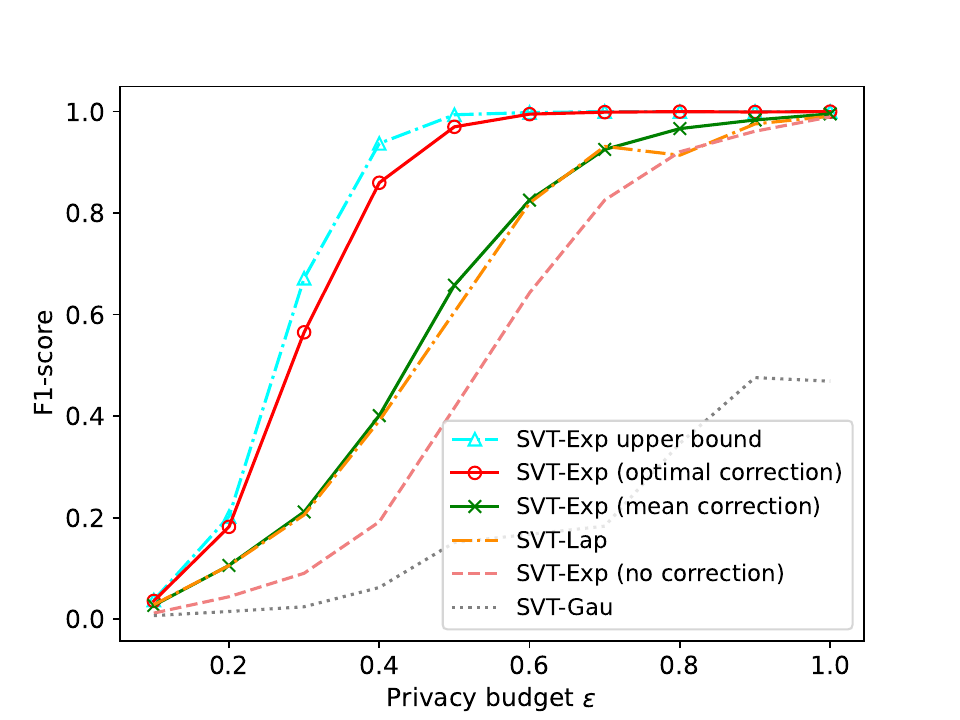}
    }
    \subfigure[BMS-POS dataset]{
    \includegraphics[width=5cm]{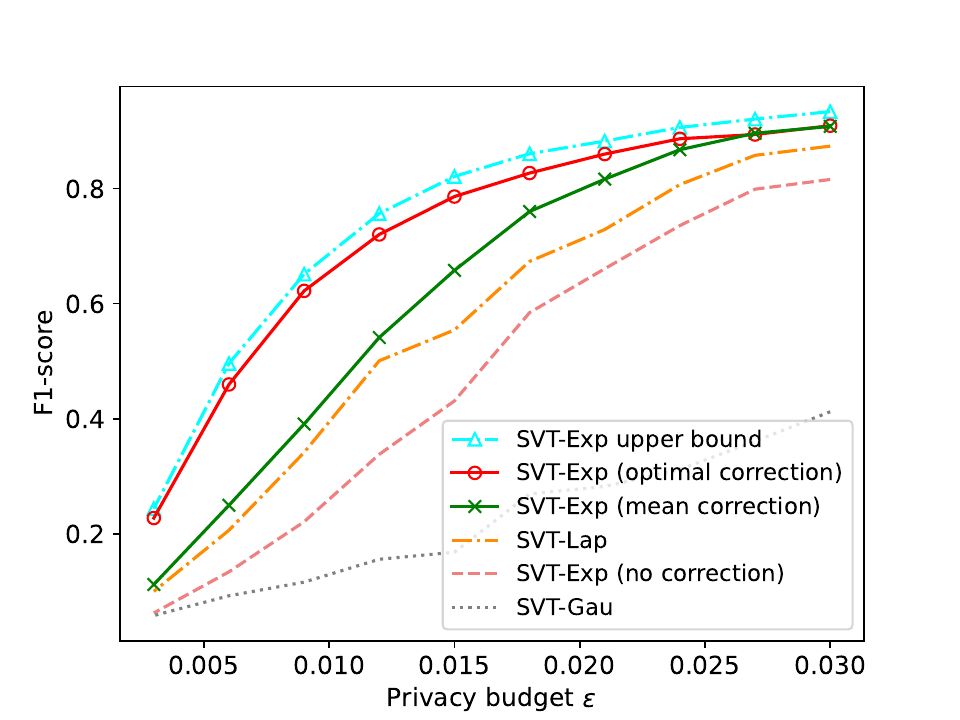}
    }
    \subfigure[Kosarak dataset]{
    \includegraphics[width=5cm]{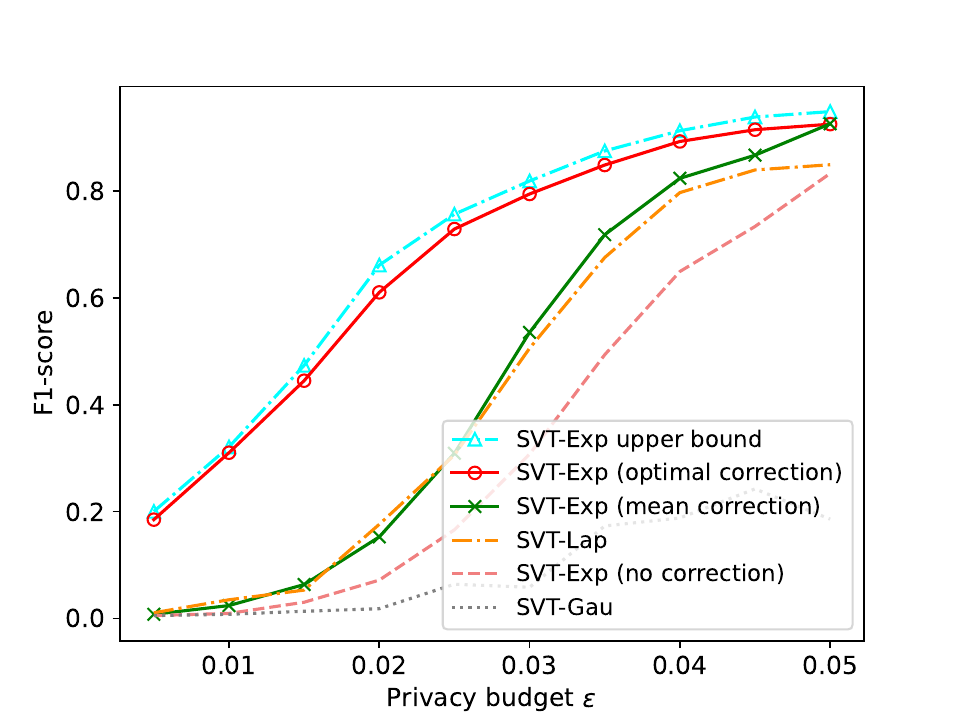}
    }
    \caption{F1-score on three datasets with $c=50$, $\alpha=0$, $k=\lfloor\frac{m}{c}\rfloor$, \texttt{RESAMPLE}=\texttt{False}, and \texttt{APPEND}=\texttt{True}. The RDP composition theorem is adopted for the overall privacy budget~(\ie, $\varepsilon$) computation.}
    \label{fig: overall-50-f1-rdp}
\end{figure*}

\begin{figure*}[htbp!]
    \centering
    \subfigure[Binary dataset]{
    \includegraphics[width=5cm]{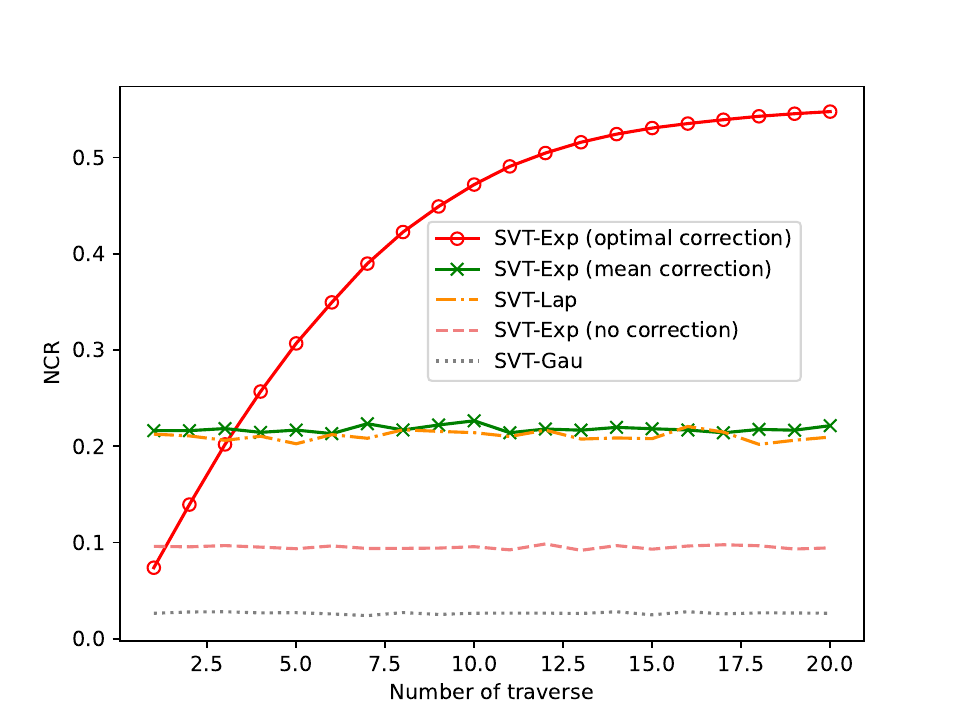}
    }
    \subfigure[BMS-POS dataset]{
    \includegraphics[width=5cm]{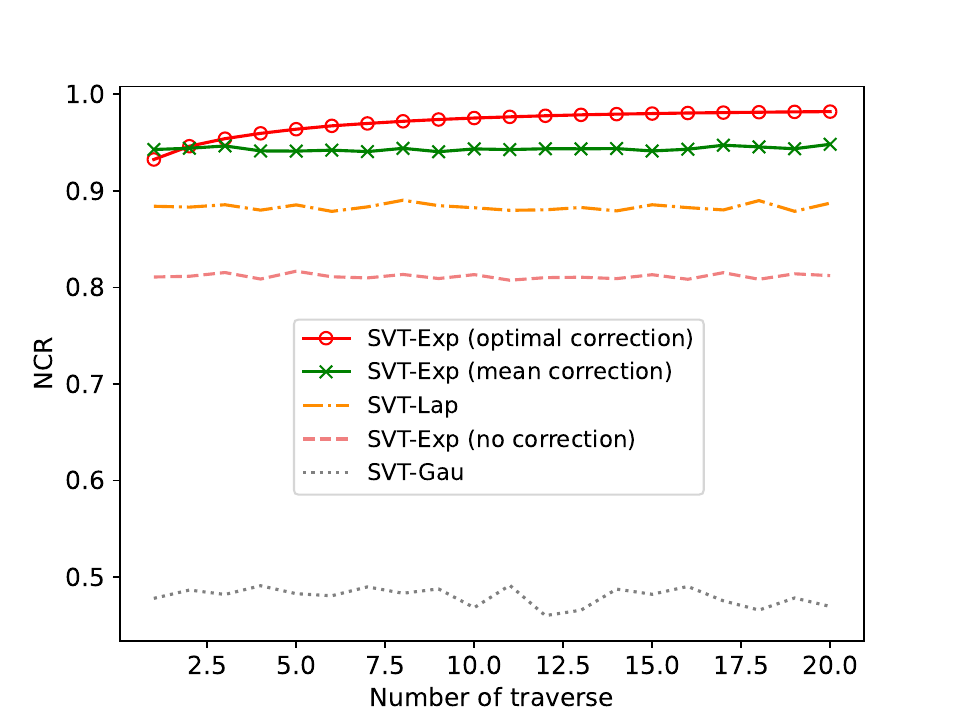}
    }
    \subfigure[Kosarak dataset]{
    \includegraphics[width=5cm]{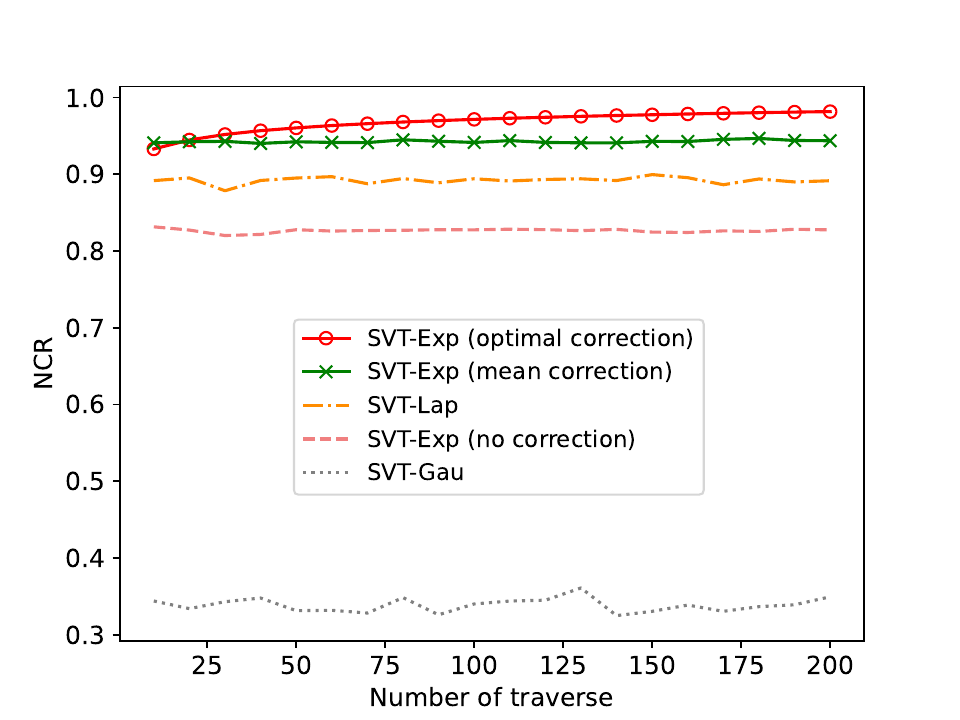}
    }
    \caption{NCR on three datasets under varying number of traverse with $c=50$, $\alpha=0$, $k=\lfloor\frac{m}{c}\rfloor$. The RDP composition theorem is adopted for the overall privacy budget~(\ie, $\varepsilon$) computation.}
    \label{fig: traverse-ncr-rdp}
\end{figure*}

\begin{figure*}[htbp!]
    \centering
    \subfigure[Binary dataset]{
    \includegraphics[width=5cm]{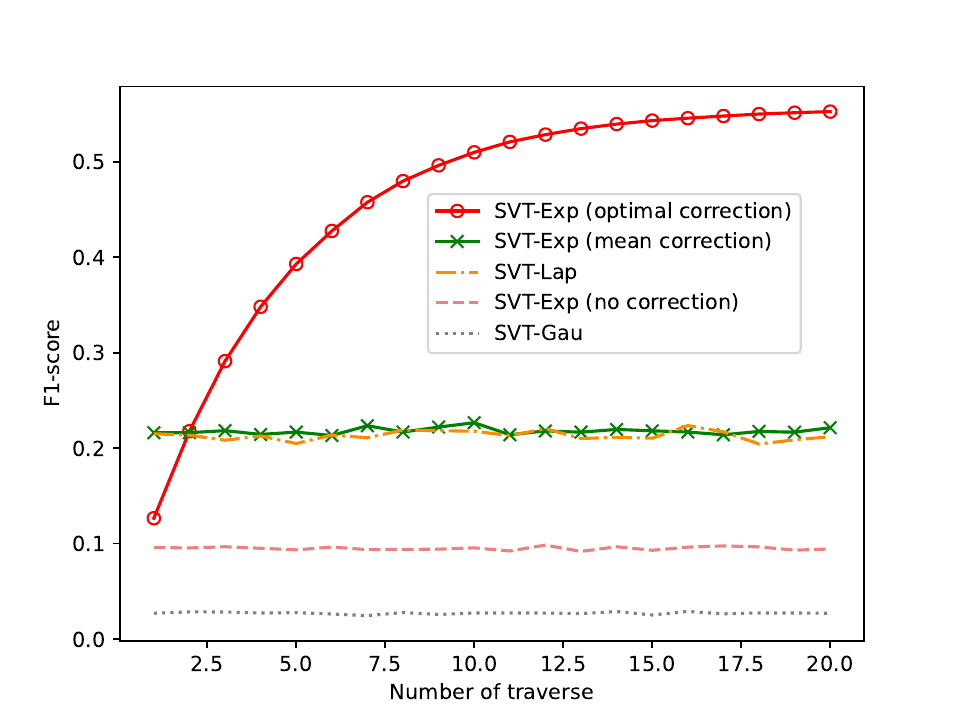}
    }
    \subfigure[BMS-POS dataset]{
    \includegraphics[width=5cm]{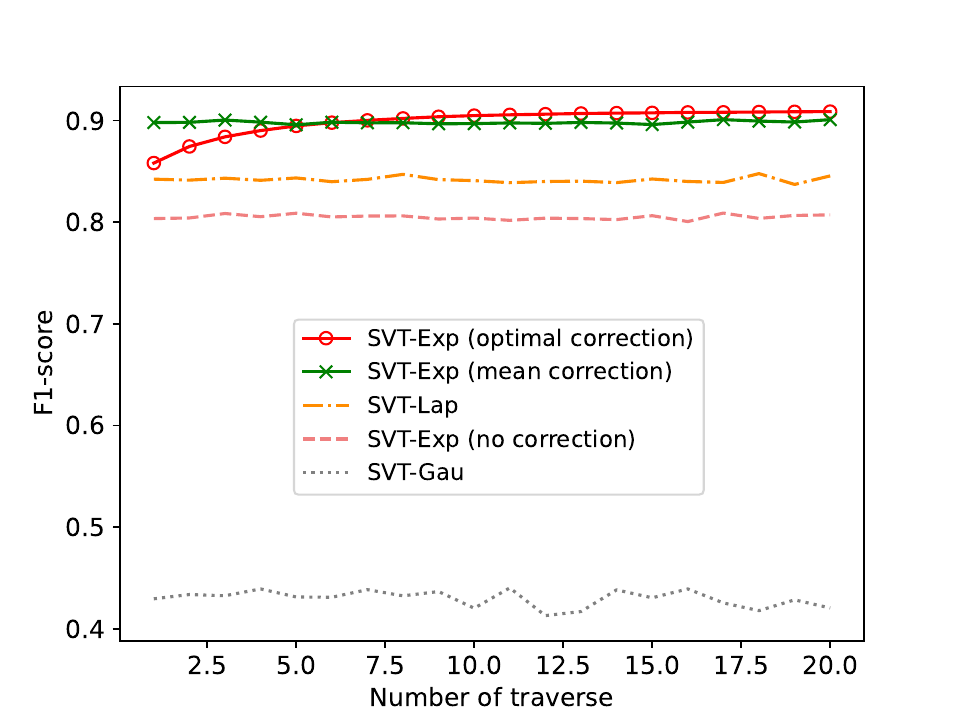}
    }
    \subfigure[Kosarak dataset]{
    \includegraphics[width=5cm]{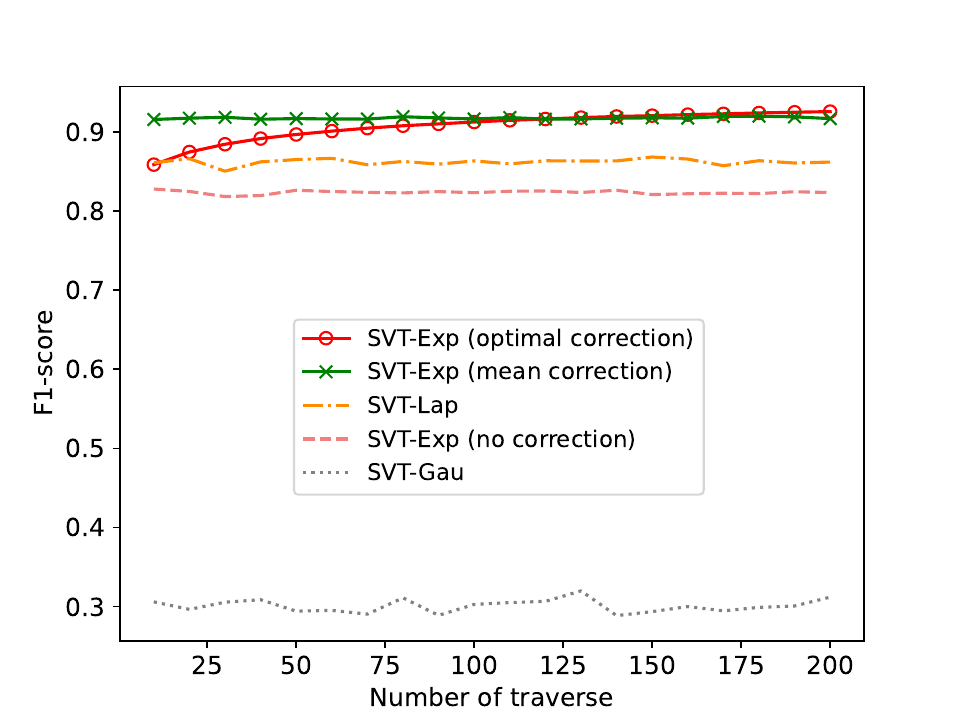}
    }
    \caption{F1-score on three datasets under varying number of traverse with $c=50$, $\alpha=0$, $k=\lfloor\frac{m}{c}\rfloor$. The RDP composition theorem is adopted for the overall privacy budget~(\ie, $\varepsilon$) computation.}
    \label{fig: traverse-f1-rdp}
\end{figure*}

\section{Impact of the Error Tolerance Parameter}
In this section, we analyze how the error tolerance parameter $\alpha$ in Equation~\ref{eq: success-rate} affects the query accuracy of our proposed method. Specifically, we present the NCR of our method with varying values of $\alpha$ on the Binary dataset, using the sequential composition theorem, as shown in Figure~\ref{fig: error_tolerance}.
Figure~\ref{fig: error_tolerance} demonstrates that our method remains robust to different values of $\alpha$ due to the effectiveness of the appending strategy. This robustness implies that setting $\alpha$ to $0$ across various data distributions generally yields satisfactory performance. However, it is important to note that a higher success probability in Figure~\ref{fig: correction-with-alpha} does not always correspond to a higher NCR. As discussed in Section~\ref{subsec: correction_methodology}, a larger $\alpha$ increases the success probability by allowing more queries to be mistakenly classified, which can lead to a decrease in NCR.
For the Binary dataset, setting $\alpha$ to $500$, which is equal to the gap between query results and the predefined threshold, results in slightly better performance compared to other settings. The rationale behind this is straightforward: increasing $\alpha$ within the range $[0,\lvert q(D)-T\rvert]$ does not negatively impact accuracy, as no query result falls within this interval. Thus, this choice of $\alpha$ optimizes the trade-off between accuracy and error tolerance in this particular context.
\begin{figure}[htbp]
    \centering
    \includegraphics[width=0.5\linewidth]{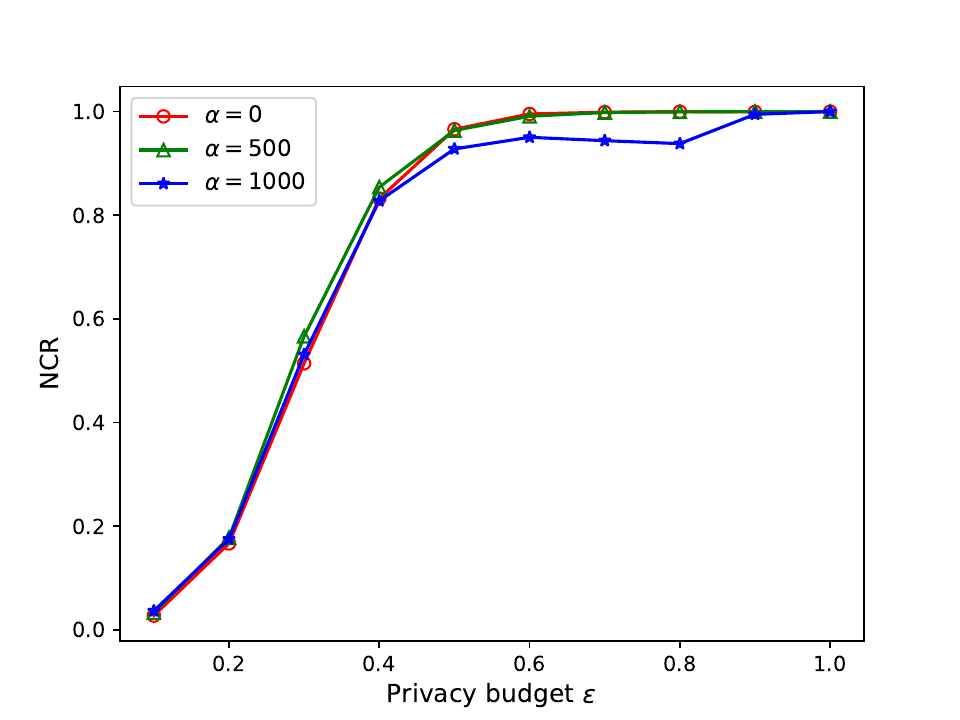}
    \caption{NCR of Algorithm~\ref{alg: exp-svt} with varying $\alpha$ on Binary dataset. The parameter $c$ is set to $50$, and the overall privacy budget $\varepsilon$ ranges from $0.1$ to $1$.}
    \label{fig: error_tolerance}
\end{figure}
\begin{figure}[htbp]
    \centering
    \includegraphics[width=0.5\linewidth]{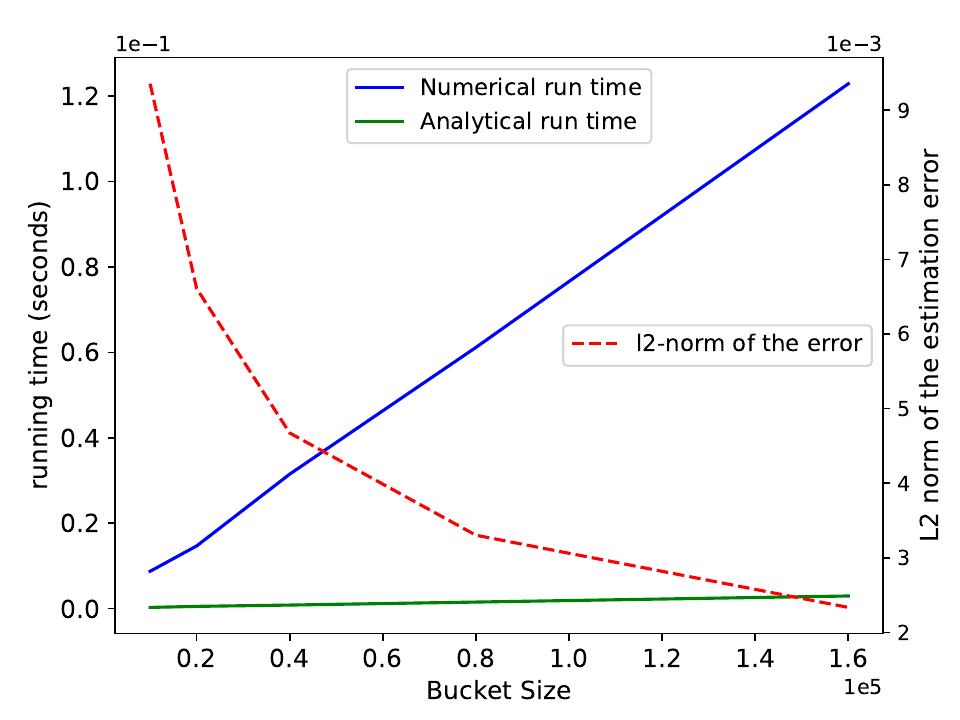}
    \caption{The trade-off of the numerical threshold correction method between estimation accuracy of $\bar{p}(r)$ in Algorithm \ref{alg: threshold_correction} and the algorithm running time.}
    \label{fig: threshold_efficiency}
\end{figure}
\section{Computation Efficiency of the Optimal Threshold Correction Term}\label{sec: efficiency_numerical}
In this section, we analyze the trade-off between the running time of our numerical optimal threshold correction framework and the accuracy of the estimated optimal threshold correction term $r^{op}$, which is derived from the success probability $p(r)$ in Equation~\ref{eq: success-rate}. The primary computational cost of our numerical framework arises from the convolution between two discretized distributions. The time complexity for this convolution is $o(m\log(m))$, where $m$ represents the bucket size in Algorithm~\ref{alg: threshold_correction} and Algorithm~\ref{alg: discretizer}.
Figure~\ref{fig: threshold_efficiency} illustrates the running time of our numerical framework as the bucket size increases, compared to the running time of the analytical computation using Equation~\ref{eq: success-rate-analytical}. The figure shows that while the analytical computation time remains constant regardless of bucket size, the running time of our numerical framework increases linearly with the bucket size.
Meanwhile, the $l_2$ norm of the estimation error of $\bar{p}(r)$ in Algorithm~\ref{alg: threshold_correction} decreases drastically as the bucket size increases.
However, it is also worth noting that even with the smallest bucket size tested in Figure~\ref{fig: threshold_efficiency} that incurs only less than $0.01$ seconds of additional computation cost, our numerical framework achieves a $l_2$ norm error in the estimation of $\bar{p}(r)$ of less than $0.01$. This estimation error scale on $\bar{p}(r)$ leads to almost no estimation error on $\bar{r}^{op}$. This is also to say, that our numerical framework achieves a very high estimation accuracy on $\bar{r}^{op}$ with only a marginal additional computation cost.